\documentclass[a4paper,twoside]{amsart}
\usepackage{amsmath,amssymb,amsthm}
\usepackage{amscd}
\usepackage[usenames,dvipsnames]{color}
\usepackage{pstricks}
\usepackage{graphicx}
\usepackage[all]{xy}
\usepackage{tabularx}

\newtheorem{prop}{Proposition}[section]

\newtheorem{defi}{Definition}[section]
\newtheorem{lemm}{Lemma}[section]
\newtheorem{thm}{Theorem}[section]
\newtheorem{coro}{Corollary}[section]
\newtheorem{ex}{Example}[section]

\begin{document}
%
%
\title[Complex powers of analytic functions and renormalization in QFT]{\textbf{Complex powers of analytic functions and meromorphic renormalization in QFT.}}
\author{
Nguyen Viet Dang\\
}%
\date{}
\thanks{This work was supported in part by the Labex CEMPI (ANR-11-LABX-0007-01).}
\dedicatory{To the memory of Louis Boutet de Monvel.}
\maketitle

\begin{abstract}
In this article, we study functional analytic properties 
of the meromorphic families of distributions
$(\prod_{i=1}^p(f_j+i0)^{\lambda_j})_{(\lambda_1,\dots,\lambda_p)\in\mathbb{C}^p}$
using Hironaka's resolution of singularities,
then using recent works on the
decomposition of meromorphic germs with linear poles,
we renormalize
products of powers of analytic functions
$\prod_{i=1}^p(f_j+i0)^{k_j}, k_j \in \mathbb{Z}$ in the space of distributions.  
We also study microlocal properties
of $(\prod_{i=1}^p(f_j+i0)^{\lambda_j})_{(\lambda_1,\dots,\lambda_p)\in\mathbb{C}^p}$ and
$\prod_{i=1}^p(f_j+i0)^{k_j}, k_j \in \mathbb{Z}$. 
In the second part, we argue that the above families
of distributions with \emph{regular holonomic singularities} provide
a universal model describing singularities of all Feynman amplitudes and 
give a new proof of renormalizability of quantum field theory on convex 
analytic
Lorentzian spacetimes as applications of ideas from the first part.
\end{abstract}


\tableofcontents

\section*{Introduction.}

To renormalize perturbative quantum field theories (QFT) on Minkowski space $\mathbb{R}^{n+1}$,
physicists often use a classical 
method, 
called \emph{dimensional regularization} and axiomatized by K. Wilson
\cite{Collins}, which can be roughly described as follows:
we work in momentum space 
and replace all integrals $\int_{\mathbb{R}^d} d^dpf(p)$ 
of rational functions $f(p)$
on $\mathbb{R}^d$
by integrals
$\int_{\mathbb{R}^{d+\varepsilon}} d^{d+\varepsilon}pf(p)$
on the ''space'' $\mathbb{R}^{d+\varepsilon}$
where the dimension is treated as a complex parameter.
For example, for a rotation invariant 
function $f$ on $\mathbb{R}^d$, 
$\int_{\mathbb{R}^d} d^dpf(p)=v_d\int_{\mathbb{R}_{\geqslant 0}}dr r^{d-1}f(r) $
where $v_d=\frac{(2\pi)^{d/2}}{\Gamma(\frac{d}{2})}$ is the $(d-1)$--volume
of the unit sphere which is calculated in such a way that
$\int_{\mathbb{R}^d} d^dp e^{-\frac{\vert p\vert^2}{2}}=\pi^{\frac{d}{2}}$.
By analytic continuation, 
these integrals depend meromorphically
in $\varepsilon$ and renormalization 
consists in subtracting
the poles in Feynman amplitudes
following the famous $R$--operation algorithm of Bogoliubov. 
Despite its efficiency, this procedure is difficult to interpret mathematically, due to the fact that
renormalization is performed in momentum space. However, the reason why 
\emph{dimensional regularization} works is intuitively quite clear since 
we integrate \emph{rational functions}
over \emph{semialgebraic sets}! 
This suggests
that in depth studies of \emph{dimensional regularization} make use
of algebraic geometry~\cite{ceyhan2012feynman2,ceyhan2012feynman,ceyhan2013algebraic}.

The purpose of the present paper
is to understand 
the meaning of analytic regularization techniques for QFT
on an analytic Lorentzian spacetime $M$
in the philosophy of Epstein--Glaser renormalization. 
In this point of view, we work in \emph{position space}
and interpret renormalization as the operation
of extension of distributions on the configuration spaces $(M^n)_{n\in\mathbb{N}}$.
At this point, we should refer to several exciting recent works
which explore analytic techniques 
in the Epstein--Glaser framework~\cite{Keller-PhD,Keller-09,dutsch2014dimensional}
in the \textbf{flat case}, especially the papers~\cite{bergbauer2009renormalization,belkale2003periods}
which, as in the present paper, use the resolution of singularities. 

In the physics terminology,
\emph{Feynman amplitudes} are formally defined as products of the form
$$\prod_{1\leqslant i<j\leqslant} G(x_i,x_j)^{n_{ij}},n_{ij}\in\mathbb{N} $$
of Feynman propagators
$G(x,y)$
which are distributions on the configuration space $M^2$, where $M$
is our Lorentzian spacetime.
The main idea of our work is to exploit
the fact that \emph{Feynman amplitudes} living on
configuration spaces $(M^n)_{n\in\mathbb{N}}$
have singularities of \textbf{regular holonomic type} i.e.
\begin{defi}
A function $u$ on some open set $U\subset \mathbb{C}^n$, 
is regular holonomic near a point $z_0$ of some smooth hypersurface
defined by some equation $\{\Gamma=0\},\Gamma(z_0)=0,d\Gamma(z_0)\neq 0$ if
$u$ is near $z_0$ a finite linear combination
with coefficients in $\mathcal{O}_{z_0}$ (the algebra of holomorphic germs at $z_0$) of functions of the form $\Gamma^\alpha,\Gamma^\alpha\log\Gamma$.
\end{defi}
These generalize meromorphic 
functions of several complex variables.
In modern
terms $\Gamma^\alpha$ (resp. $\log\Gamma$) would be defined as the distributions 
$(\Gamma + i0)^\alpha$ (resp. $\log(\Gamma + i0)$).  
Our approach, which goes back to Hadamard~\cite{Hadamard,BGP} and pre-dates
the Schwartz theory of distributions, uses the description of
the Feyman propagator as a branched meromorphic function
(possibly logarithmically branched) on the complexified spacetime.
Indeed, the singularity
of $G$ has the representation:
\begin{eqnarray}
G(x,y)=\frac{U}{\Gamma +i0}+V\log\left(\Gamma+i0\right) + W
\end{eqnarray}
where 
$\Gamma,U,V,W$ are analytic functions and
it follows that $G$ has \textbf{regular holonomic singularity}
along the null cone.
Inspired by the work of Borcherds~\cite{Borcherds-10}, our idea is to regularize $G$ by considering
the modified propagator:
\begin{eqnarray}
G_\lambda(x,y)=\left(\frac{U}{\Gamma +i0}+V\log\left(\Gamma+i0\right) + W\right)(\Gamma +i0)^\lambda
\end{eqnarray}
which is still of holonomic type.
Then we consider 
regularized Feynman amplitudes on configuration space $M^n$ 
depending
on several complex variables 
$(\lambda_{ij})_{1\leqslant i<j\leqslant n}\in\mathbb{C}^{\frac{n(n-1)}{2}}$:
$$\prod_{1\leqslant i<j\leqslant n} G_{\lambda_{ij}}(x_i,x_j)^{n_{ij}},n_{ij}\in\mathbb{N} $$
so our goal
in the present paper is to show that:
\begin{itemize}
\item the regularized Feynman amplitude $\prod_{1\leqslant i<j\leqslant n} G_{\lambda_{ij}}(x_i,x_j)^{n_{ij}},n_{ij}\in\mathbb{N} $
depends meromorphically on $(\lambda_{ij})_{1\leqslant i<j\leqslant n}\in\mathbb{C}^{\frac{n(n-1)}{2}}$
with value distribution. 
\item Outside the big diagonal $D_n=\{(x_1,\dots,x_n)\in M^n \text{ s.t. } \exists (i<j), x_i=x_j  \}$, 
it is holomorphic in $\lambda$
and 
$$ \lim_{\lambda\rightarrow 0} \prod_{1\leqslant i<j\leqslant n} G_{\lambda_{ij}}(x_i,x_j)^{n_{ij}}=\prod_{1\leqslant i<j\leqslant n} G(x_i,x_j)^{n_{ij}}  $$
where the above equality only holds in $\mathcal{D}^\prime(M^n\setminus D_n)$ i.e. on the configuration space of $n$-points which are all distinct. 
\item We can define a collection of renormalization maps
$\mathcal{R}_{M^n}$ which are linear maps from the space of Feynman amplitudes
to $\mathcal{D}^\prime(M^n)$ 
such that
$\mathcal{R}_{M^n}\left( \prod_{1\leqslant i<j\leqslant n} G(x_i,x_j)^{n_{ij}}\right) $
is a distributional extension of $\prod_{1\leqslant i<j\leqslant n} G(x_i,x_j)^{n_{ij}}  $
which satisfies the consistency axioms \ref{axiomsrenormmaps} (also elegantly described 
in~\cite{NST}) ensuring that the renormalization satisfies
physical requirements such as causality.
\end{itemize}

\subsubsection{Contents of the paper.}
Our paper is devoted to the realization of the above program and is divided in two parts:
the first part is of independent interest and of purely mathematical
nature
whereas the second part presents applications
of the first part to the renormalization of QFT
on analytic spacetimes.

Let us start with the first part.
In the first two sections, we study 
the universal model which describes the singularities of all Feynman 
amplitudes which consists in ill--defined products of powers of real analytic functions
of the form
$\prod_{i=1}^p(\log(f_j+i0))^{p_j}(f_j+i0)^{k_j}$ where $p_j$ are nonnegative integers and 
$k_j$ negative integers.
Then we show how to make
sense of the above ill--defined product of distributions
by analytic continuation as follows:
\begin{enumerate}
\item we consider the family $\left(\prod_{i=1}^p(f_j+i0)^{\lambda_j}\right)_{(\lambda_j)_j}$
where $(\lambda_1,\dots,\lambda_p)\in\mathbb{C}^p$ and
use the resolution of singularities of Hironaka
to show in Theorem \ref{Atiyahplus} that the family $\left(\prod_{i=1}^p(f_j+i0)^{\lambda_j}\right)_{(\lambda_j)_j}$
depends meromorphically on $(\lambda_1,\dots,\lambda_p)\in\mathbb{C}^p$ with linear poles
with value distribution.
\item Motivated by the problem of renormalization of conical multiple zeta functions at integers,
Guo--Paycha--Zhang~\cite{PaychaZhang2015} were able to generalize the Laurent series decomposition
to meromorphic germs with linear poles. Then we use their recent results to decompose
the meromorphic family $\left(\prod_{i=1}^p(f_j+i0)^{\lambda_j}\right)_{(\lambda_j)_j}$
in a regular part which is holomorphic in $\lambda$ and a singular part which contains the polar singularity
then
we define a renormalization $\mathcal{R}_\pi\left(\prod_{i=1}^p(f_j+i0)^{k_j}\right)$
by letting the complex parameter $(\lambda_1,\dots,\lambda_p)\in\mathbb{C}^p$ go to 
$(k_1,\dots,k_p)\in\mathbb{C}^p$ in the regular part.
\item $\mathcal{R}_\pi$ satisfies the following factorization identity of central
importance:
let $U,V$ be open sets in $\mathbb{R}^{n_1},\mathbb{R}^{n_2}$
respectively and
$f_1,\dots,f_p$ (resp $g_1,\dots,g_p$) real analytic functions on $U$
(resp $V$) then: 
\begin{equation}
\mathcal{R}_\pi\left(f_1^{k_1}\dots f_p^{k_p}  g_1^{l_1}\dots g_p^{l_p}\right)=
\mathcal{R}_\pi\left(f_1^{k_1}\dots f_p^{k_p}\right)\otimes \mathcal{R}_\pi\left(g_1^{l_1}\dots g_p^{l_p}\right).
\end{equation}
where the tensor product $\otimes$ is the exterior
tensor product: $\mathcal{D}^\prime(U)\otimes \mathcal{D}^\prime(V)\mapsto\mathcal{D}^\prime(U\times V) $. 
\end{enumerate}
Our philosophy is to hide the complicated combinatorics of renormalization behind two deep
results in analytic geometry: 
the resolution of singularities of Hironaka and the generalized
decomposition in Laurent series 
of~\cite{PaychaZhang2015}.

However, for our applications to QFT it is necessary to show that our renormalization
satisfies the axioms \ref{axiomsrenormmaps} hence we must
study the microlocal properties of the family $\left(\prod_{i=1}^p(f_j+i0)^{\lambda_j}\right)_{(\lambda_j)_j}$
and of the renormalized distribution
$\mathcal{R}_\pi\left(\prod_{i=1}^p(f_j+i0)^{k_j}\right)$.
We start in section $3$ by giving easy results on products of distributions in the setting of Sobolev
spaces and we give simple bounds in Theorem \ref{u=0thm} on the wave front of products. 
Then in section $4$,
we apply these tools  
to study the microlocal properties
of the family $((f+i0)^\lambda)_\lambda$.
In Theorem \ref{WFf+i0}, we bound the wave front set of $((f+i0)^\lambda)_\lambda$
for generic values of $\lambda$:
\begin{equation}
WF((f+i0)^\lambda)\subset \{ (x;\xi) \text{ s.t. }\exists \{(x_k,a_k)_k\}\in \left(\mathbb{R}^n\times \mathbb{R}_{>0}\right)^{\mathbb{N}} ,
x_k\rightarrow x, f(x_k)\rightarrow 0, a_kdf(x_k)\rightarrow\xi \}.
\end{equation}
In section $5$, based on the recent work~\cite{dabrowski2013functional}
we present a functional calculus of meromorphic functions with value
$\mathcal{D}^\prime_\Gamma$, where $\mathcal{D}^\prime_\Gamma$
is the space of distributions whose wave front set is contained
in the conic set $\Gamma$.
Using this functional calculus, we prove two Theorems
about functional analytic properties of 
the families $((f+i0)^\lambda)_\lambda$ and $\left(\prod_{i=1}^p(f_j+i0)^{\lambda_j}\right)_{(\lambda_j)_j}$.
In section $6$, we show that
\begin{thm}
Let $f$ be a real valued analytic function s.t. $\{df=0\}\subset\{f=0\}$,
$Z\subset \mathbb{C}$ a discrete subset containing the poles of
the meromorphic family $((f+i0)^\lambda)_\lambda$. Set
$$\Lambda_f=\{ (x;\xi) \text{ s.t. }\exists \{(x_k,a_k)_k\}\in \left(\mathbb{R}^n\times \mathbb{R}_{>0}\right)^{\mathbb{N}} ,
x_k\rightarrow x, f(x_k)\rightarrow 0, a_kdf(x_k)\rightarrow\xi \} .$$
For all $z\in Z$,
let $a_k$ to be the coefficients of the Laurent series
expansion of $\lambda\mapsto (f+i0)^\lambda$ around $z$
\begin{eqnarray*}
(f+i0)^\lambda=\sum_{k\in\mathbb{Z}} a_k(\lambda-z)^k.
\end{eqnarray*}
Then for all $k\in\mathbb{Z}$, $WF(a_k)\subset\Lambda_f$
and
if $k<0$ then $a_k$ is a distribution
\textbf{supported by the critical locus} $\{df=0\}$. 
\end{thm}
In the multiple functions case $(f_1,\dots,f_p)$, which is the case of interest, 
we describe in paragraph \ref{Threemainassumptions} 
geometric constraints on the zero sets of $(f_1,\dots,f_p)$
and the critical sets $\{df_1=0\},\dots,\{df_p=0\}$
which allow us to give an optimal result in Theorem \ref{functionalprod}: 
\begin{thm}
Under the assumptions of paragraph \ref{Threemainassumptions},
the family
$\left(\prod_{j=1}^p(f_j+i0)^{\lambda_j}\right)_{\lambda \in\mathbb{C}^p}$
depends meromorphically on $\lambda$ with linear poles
with value $\mathcal{D}^\prime_\Lambda$ where
\begin{eqnarray*}
\Lambda=\bigcup_{J}\{(x;\xi) |j\in J, f_j(x)=0,df_j(x)\neq 0, \xi=\sum_{j\in J} a_jdf_j(x),a_j>0  \}\cup N^*\Sigma_J,\\
\Sigma_J=\cap_{j\in J}\{df_j=0\}.
\end{eqnarray*}
The distribution
\begin{equation}
\mathcal{R}_\pi\left(\prod_{j=1}^p(f_j+i0)^{k_j} \right) \in\mathcal{D}^\prime(U)
\end{equation}
is a distributional extension
of $\prod_{j=1}^p(f_j+i0)^{k_j}\in\mathcal{D}^\prime(U\setminus X)$ and has wave front 
contained in $\Lambda$.
\end{thm}
The above bound on the wave front set of $\mathcal{R}_\pi\left(\prod_{j=1}^p(f_j+i0)^{k_j} \right)$
is quite natural from the point of view of symplectic geometry. 
Indeed, motivated by problems in representation theory, 
Aizenbud and Drinfeld~\cite{Aizenbud-12} 
introduced the class of WF-\textbf{holonomic} distribution
(which contains Fourier transform of algebraic measures for instance): 
\begin{defi}
A distribution $t$ on a smooth analytic manifold $M$ is called
WF-\textbf{holonomic} if $WF(t)$ is locally contained in some
finite union
of conormal bundles of some smooth analytic submanifolds of $M$, said differently,
for all bounded open set $U\subset M$, there is a finite number of analytic submanifolds
$(N_i)_{i}$ s.t. $WF(t)\subset \bigcup_{i\in I}N^*(N_i)$.
\end{defi}
The main Theorem of section $6$ shows that both
$\left(\prod_{j=1}^p(f_j+i0)^{\lambda_j}\right)_{\lambda \in\mathbb{C}^p}$
and $\mathcal{R}_\pi\left(\prod_{j=1}^p(f_j+i0)^{k_j} \right)$
are WF-\textbf{holonomic}. 
\begin{ex}
The Feynman propagator on $\mathbb{R}^{3+1}$ has the form
$G=C(Q+i0)^{-1}$ where $Q$ is the quadratic form of signature $(1,3)$ and its
wave front set is contained
in the union of the conormal $N^*(\{Q=0\}\setminus \{0\})$ of the cone $\{Q=0\}\setminus \{0\}$
(with vertex at the origin removed) and the conormal
of the origin $T_{\{0\}}^*\mathbb{R}^{3+1}=N^*(\{0\})$.
It follows that $G$ is WF-\textbf{holonomic}.
\end{ex}

In the second part of our paper, we apply all results
of the first part to prove the existence
in Theorem \ref{renormthmmain}
of renormalization maps $(\mathcal{R}_{M^n})_{n\in\mathbb{N}}$
compatible with the axioms \ref{axiomsrenormmaps}
following our philosophy 
of analytic continuation explained at the beginning of the introduction.
Let us explain the central novel feature of our approach:
unlike Borcherds~\cite{Borcherds-10}, we regularize with
\emph{as many complex variables as the number of propagators} in a given Feynman amplitude.
If we were to introduce only one regularization parameter
$\lambda$ like in classical QFT textbooks and Borcherds' work, then we would
be forced to subtract divergences in a hierarchical manner 
using either the St\"ueckelberg--Bogoliubov renormalization group
or the Bogoliubov R-operation since renormalization of Feynman
amplitudes must take into account subtle phenomena such as 
nested subdivergences, overlapping 
divergences...It is well known that a na\"ive subtraction of all poles
would not satisfy the axiom of causality in \ref{axiomsrenormmaps}.
However, the effect of introducing many regularization parameters
resolves the singularities and using the generalized 
decomposition in~\cite{PaychaZhang2015}, it is sufficient to 
subtract all singular parts all at once as done in our main Theorem \ref{renormthmmain}. 
To conclude our paper, we show that unlike the methods of Brunetti--Fredenhagen~\cite{Brunetti2} 
and of our thesis~\cite{Dangthese}, analytic techniques make no use of partitions of unity which shows 
that our meromorphic
renormalization is functorial when restricted to a 
category $\mathbf{M}_{ca}$ defined in subsection \ref{category} 
whose objects are geodesically convex analytic Lorentzian spacetimes $(M,g)$
equipped with a Feynman propagator $G$, this functoriality 
emphasizes the local character of our renormalization techniques.

\subsubsection{Future projects.}
In the sequel of the present paper~\cite{Viet-meromrenorm2}, 
we will relate our 
meromorphic regularization techniques with the
renormalization group of Bogoliubov, discuss the specific
examples of static spacetimes where our renormalization 
can be made global using the Wick rotation and finally, more importantly, we
plan to discuss important
extensions of our results to the 
case of \textbf{smooth globally hyperbolic spacetimes}
following suggestions of C. Guillarmou.

\subsubsection{Acknowledgements.}
First, we would like to thank Laura Desideri for her support and many 
fruitful discussions when we started this project together which
was initially supposed to be a joint work.
We thank Pierre Schapira, Daniel Barlet, Avraham Aizenbud, Sylvie Paycha, St\'ephane Malek,
Colin Guillarmou
and Alan Sokal 
for useful correspondance or discussions on the subject of
the present paper and especially
many thanks to Christian Brouder for his constant support
and for urging us to finish the present draft.
Finally, this work is dedicated to the memory of Louis Boutet de Monvel
who suggested us to look at the problem
of the renormalization in QFT from the point of view of
holonomic $\mathcal{D}$-modules with regular singularities and 
whose influence on us can be felt in every page of the present work.\\

\begin{center}
\textsc{\section*{Part I: analytic continuation techniques.}}
\end{center}

This part forms the analytical core of our paper
since all techniques like ``dimensional regularization''
in quantum field theory relie more or less
on the same idea of analytic continuation:
we introduce some parameter $\lambda$ that will smooth out singularities
of Feynman propagators
then we show that all quantities
depend meromorphically in the complex parameter 
$\lambda$. In mathematics, this is related
to Atiyah's approach~\cite{AtiyahHironaka} to the 
problem of division of distributions
and also the analytic continuation techniques
described in~\cite{bernstein1972analytic}
based on the existence of
Bernstein Sato polynomials.

\subsection{Meromorphic functions.}

 \subsubsection{Meromorphic functions in several variables.}
Before we move on, let us recall basic facts about meromorphic functions
in several complex variables.
To define meromorphic functions in several variables, we first need to define 
the notion of
thin set. A set $Z\subset \Omega$ is called a \emph{thin set}
if for all $x\in Z$, there is some neighborhood $V_x$ of $x$
such that $(V_x\cap Z)\subset \{g=0\}$ for some 
non zero holomorphic function $g$ defined on $V_x$. 
A function $f$ is meromorphic on $\Omega$ if there exists a thin set $Z\subset \Omega$
such that $f$ is holomorphic on $\Omega\setminus Z$ and
near any point $x\in\Omega$, there is some neighborhood $V_x$ of $x$
s.t. $f|_{V_x\setminus Z}= \frac{\varphi}{\psi}$ where 
$(\varphi,\psi)$ are holomorphic on $V_x$.
However in meromorphic regularization in QFT, we encounter
more restrictive classes of meromorphic functions.

\subsubsection{Meromorphic functions with linear poles.}

In our paper, all 
meromorphic functions of 
several variables $\lambda=(\lambda_1,\dots,\lambda_p)\in\mathbb{C}^p$
have polar singularities
along countable
union of affine hyperplanes of certain types. 
They are
\textbf{meromorphic functions with linear poles} 
in the terminology
of Guo--Paycha--Zhang~\cite{PaychaZhang2015}.

Consider the dual space $(\mathbb{C}^p)^*$ of $\mathbb{C}^p$ where each element $L\in (\mathbb{C}^p)^*$
defines a linear map $L: \lambda \in \mathbb{C}^p\mapsto L(\lambda)$.
Consider the \textbf{lattice} of covectors with integer coefficients $\mathbb{N}^p\subset (\mathbb{C}^p)^*$ then
to every element $L \in \mathbb{N}^p$, consider the linear map
$L: \lambda \in \mathbb{C}^p\mapsto L(\lambda)$.

\begin{defi}
Let $k=(k_1,\dots,k_p)$ be some element in $\mathbb{Z}^p$, 
then a germ of meromorphic 
function $f$ at $k$ has \textbf{linear poles} if there are $m$ vectors
$(L_i)_{1\leqslant i\leqslant m}\in (\mathbb{N}^p)^m$ in the lattice $\mathbb{N}^p$, 
such that
\begin{equation}
(\prod_{i=1}^m L_i(.+k))f  
\end{equation}
is a holomorphic germ at $k=(k_1,\dots,k_p)\in\mathbb{C}^p$.
An element $\frac{1}{\prod_{i=1}^m L_i(.+k) }$ is called a 
\textbf{simplicial fraction
of order} $m$ at $k$. 
\end{defi} 

Geometrically such meromorphic germ $f$ is singular 
along $m$ affine hyperplanes of equation
$\{\lambda\in\mathbb{C}^p \text{ s.t. } L_i(\lambda+k)=0 \}$
intersecting at point $k=(k_1,\dots,k_p)$ with integer coordinates in $\mathbb{C}^p$.

\subsubsection{Distributions depending meromorphically on extra parameters.}

The core of our analytic 
regularization method
in position space is the concept
of distribution depending holomorphically
(resp meromorphically)
w.r.t. some parameter $\lambda=(\lambda_1,\dots,\lambda_p)\in\mathbb{C}^p$
introduced in~\cite{Gelfand-ShilovI}:
\begin{defi}
Let $U$ be an open set in a smooth oriented manifold
$M$ and $\Omega$ an open subset of $\mathbb{C}^p$. 
Then a family $(t_\lambda)_\lambda$, $\lambda\in\Omega$ is holomorphic
(resp meromorphic) with value distribution if for all test function
$\varphi\in\mathcal{D}(U)$, $\lambda\in\Omega\mapsto t_\lambda(\varphi)\in\mathbb{C}$ 
is holomorphic
(resp meromorphic) in $\lambda\in\Omega$.
\end{defi}

If $(t_\lambda)_\lambda$ depends \textbf{holomorphically}
on $\lambda\in\Omega\subset\mathbb{C}^p$ with value $\mathcal{D}^\prime$,
let $\gamma=\gamma_1\times\dots\times \gamma_p$ be a cartesian product
where each $\gamma_i$ is a \emph{continuous curve} in $\mathbb{C}$,
then we can define weak integrals $\int_{\gamma\subset \mathbb{C}^p} d\lambda t_\lambda$ 
as limits of Riemann sums which converge
to some element in $\mathcal{D}^\prime$
since for all test function $\varphi\in\mathcal{D}$,
the element $\int_{\gamma} d\lambda t_\lambda(\varphi)$
exists as a limit of Riemann sums by continuity of $\lambda\in\gamma \mapsto t_\lambda(\varphi)$.

\subsubsection{A gain of regularity: when weak holomorphicity becomes strong 
holomorphicity.}

Now we give an easy 
\begin{prop}\label{weakstrongholo}
Let $U$ be an open set in $\mathbb{R}^n$, $\Omega\subset \mathbb{C}^p$,
$(t_\lambda)_{\lambda\in\Omega}$ a holomorphic family of distributions in $\mathcal{D}^\prime(U)$.
Then near every $z\in\Omega$, $t_\lambda$ admits a Laurent series expansion
$t_\lambda=\sum_{\alpha} (\lambda-z)^\alpha t_\alpha$
where $\alpha=(\alpha_1,\dots,\alpha_n)\in\mathbb{N}^n$ and each
coefficient $t_\alpha$ is a distribution in $\mathcal{D}^\prime(U)$ 
such that
for all test function $\varphi$, 
$\sum_{\alpha} (\lambda-z)^\alpha t_\alpha(\varphi)$
converges as power series near $z$.
\end{prop}
\begin{proof}
Without loss of generality
assume that $z=0$.
It suffices to observe that
by weak holomorphicity of $t$ and the multidimensional 
Cauchy's formula~\cite[p.~3]{GunningRossi} for any polydisk
$D_1\times \dots\times D_p$ such that 
$\partial D_i$ is a circle surrounding
$z_i$, for all test function $\varphi\in\mathcal{D}(U)$:
\begin{eqnarray}
t_\lambda(\varphi)=\frac{1}{(2i\pi)^p}\int_{\partial D_1 \times\dots\times \partial D_p} \frac{t_z (\varphi)dz_1\wedge \dots\wedge dz_p}{(z_1-\lambda_1)\dots (z_p-\lambda_p)}. 
\end{eqnarray}
For all test function, set $t_\alpha(\varphi) = \frac{\alpha !}{(2i\pi)^p}\int_{\partial D_1 \times\dots\times \partial D_p} \frac{t_z (\varphi)dz_1\wedge \dots\wedge dz_p}{(z_1-\lambda_1)^{\alpha_1+1}\dots (z_p-\lambda_p)^{\alpha_n+1}}$,
then $t_\alpha$ is linear on $\mathcal{D}(U)$.
Let us prove it defines a genuine distribution.
By a simple application of the uniform boundedness principle, for every compact
$K\subset U$ there exists a $C>0$ and some continuous seminorm $P$ for the Fr\'echet topology
of $\mathcal{D}_K(U)$ such that:
\begin{equation}
\forall \varphi\in\mathcal{D}_K(U), \sup_{\lambda\in \partial D_1 \times\dots\times \partial D_p}\vert t_\lambda(\varphi)\vert\leqslant CP(\varphi).
\end{equation} 
Assuming that all discs $\partial D_i$ have radius $r$, it immediately follows that
$t_\alpha$ satisfies a distributional version of Cauchy's bound:
\begin{equation}\label{Cauchybound}
\forall \varphi\in\mathcal{D}_K(U), \vert t_\alpha(\varphi)\vert\leqslant \frac{\alpha!}{r^{\vert\alpha\vert}}  CP(\varphi).
\end{equation}
This immediately implies that $(t_\alpha)_\alpha$ are distributions
and also that the power series 
$\sum_{\alpha} \lambda^\alpha t_\alpha(\varphi)$
converges near $0\in\Omega$.
\end{proof}

\subsubsection{Meromorphic functions with linear poles with value distribution.}

In the present work, we deal with families of distributions $(t_\lambda)_{\lambda\in\mathbb{C}^p}$
in $\mathcal{D}^\prime(U)$ depending meromorphically on
$\lambda\in\mathbb{C}^p$ with linear poles.
\begin{defi}
A family of distributions $(t_\lambda)_{\lambda\in\mathbb{C}^p}$
in $\mathcal{D}^\prime(U)$ depends meromorphically on
$\lambda\in\mathbb{C}^p$ with linear poles
if for every $x\in U$, there is a neighborhood $U_x$ of $x$, a collection $(L_i)_{1\leqslant i\leqslant m}\in(\mathbb{N}^p)^m\subset (\mathbb{C}^{p*})^m$ of linear functions with integer coefficients
on $\mathbb{C}^p$ such that for any element
$z=(z_1,\dots,z_p)\in \mathbb{Z}^p$,
there is a neighborhood $\Omega\subset \mathbb{C}^p$ of $z$, 
such that
\begin{equation}
\lambda\in\Omega \mapsto\prod_{i=1}^m(L_i(\lambda+z))t_\lambda
\end{equation}
is holomorphic with value distribution.
\end{defi}
The above expansion is a useful substitute to
the Laurent series expansion in the one variable case.
In particular, $(\prod_{i=1}^mL_i(\lambda+z))t_\lambda|_{U_x}$
is a holomorphic germ near $z$ with value distribution.
Locally near any element $z=(z_1,\dots,z_p)\in \mathbb{Z}^p$,
the polar set of $t$ is the union of exactly $m$ affine hyperplanes.

\subsection{The fundamental example of hypergeometric distributions.}

Next, we will study the fundamental 
example of such analytic continuation procedure for the simplest kind of hypergeometric
distributions, we work in
$\mathbb{R}^n$ with coordinates $(y_1,\dots,y_n)$:
\begin{lemm}\label{polesbasicexample}
Let $\Gamma\subset \mathbb{R}^n$ be a quadrant 
$\cap_{1\leqslant i\leqslant n}\{y_i\varepsilon_i\geqslant 0\}$ for $\varepsilon\in\{-1,1\}^n$.
The family of distributions $(t_\mu)_{\mu}$ defined as 
\begin{equation}
t_\mu=1_{\Gamma} 
y_1^{\mu_1} \dots y_n^{\mu_n} \text{ for }Re(\mu_i)>-1
\end{equation} 
extends meromorphically in $\mu=(\mu_1,\dots,\mu_n)\in\mathbb{C}^n$
with polar set $\cup_{1\leqslant i\leqslant n,k\in\mathbb{N}^*}\{ \mu_i+k=0 \}$.
\end{lemm}

\begin{proof}
The proof follows from an easy integration by parts argument, for all test function
$\varphi\in\mathcal{D}(\mathbb{R}^n)$, for $-1<Re(\mu_i)\leqslant 0$ and for any integers $(k_1,\dots,k_n)\in(\mathbb{N}^*)^n$:
\begin{eqnarray*}
t_\mu(\varphi)&=&\int_{\Gamma} dy_1\dots dy_n 
y_1^{\mu_1} \dots y_n^{\mu_n} \varphi(y_1,\dots,y_n)\\
&=&\left(\prod_{i=1}^n \frac{1}{\mu_i+k_i}\dots \frac{1}{\mu_i+1}   \right)\int_{\Gamma} dy_1\dots dy_n 
y_1^{\mu_1+k_1} \dots y_n^{\mu_n+k_n} \varphi(y_1,\dots,y_n) 
\end{eqnarray*}
where both sides are holomorphic
in the domain $-1<Re(\mu_i)$. However for 
$-k_i-1<Re(\mu_i)$, the right hand side
is well defined and meromorphic with poles at $\mu_i=-k_i,\dots,\mu_i=-1$. It is thus an analytic
continuation of the distribution $(t_\mu)_\mu$ on the right hand side which yields the desired 
result. 
\end{proof}

Moreover, the distribution $(t_\mu)_\mu$ exhibits an
interesting separation of variables property since
it admits a Laurent series expansion around elements of the form
$(k_1,\dots,k_n)\in(-\mathbb{N}^*)^n$ as the product of $n$
meromorphic functions in each variable $\mu_i$:
\begin{lemm}\label{keyLaurentseries}
Let us consider again the distribution $t_\mu$
of Lemma \ref{polesbasicexample}.
Near any element $(-k_1,\dots,-k_n)\in \mathbb{C}^n, k_i\in\mathbb{N}^*$,
the \textbf{polar set} of the family $(t_\mu)_\mu$ is a divisor with normal crossings
$\cup_{1\leqslant i\leqslant n}\{ \mu_i=-k_i \}$
i.e. it is the union of $n$ affine coordinates hyperplanes
and $t_\mu$ admits a Laurent series expansion in $(\mu_i+k_i),1\leqslant i\leqslant n$ of the form
\begin{equation}
t_\mu=\sum_\alpha u_\alpha\prod_{i=1}^n(\mu_i+k_i)^{\alpha_i-1} . 
\end{equation}
where $\alpha=(\alpha_1,\dots,\alpha_n)\in \mathbb{N}^n$ is a multi--index 
and $u_\alpha\in\mathcal{D}^\prime(U)$.
In particular, $t$ is meromorphic with linear poles
with value $\mathcal{D}^\prime(U)$.
\end{lemm}
\begin{proof}
Near an element $(-k_1,\dots,-k_n)\in (-\mathbb{N}^*)^n\subset \mathbb{C}^n$, $t_\mu$
writes as a product 
\begin{eqnarray*}
t_\mu=\prod_{i=1}^n \frac{u_\mu}{\mu_i+k_i} 
\end{eqnarray*} 
of a simplicial fraction $\prod_{i=1}^n \frac{1}{\mu_i+k_i} $
with the distribution $u_\mu$ defined as:
\begin{eqnarray*}
u_\mu(\varphi)= \left(\prod_{i=1}^n \frac{1}{\mu_i+k_i-1}\dots \frac{1}{\mu_i+1}   \right)\int_{\Gamma} dy_1\dots dy_n 
y_1^{\mu_1+k_1} \dots y_n^{\mu_n+k_n} \varphi(y_1,\dots,y_n)
\end{eqnarray*}
which is a distribution depending holomorphically
on $\mu$ provided that for all $i\in\{1,\dots,n\}$, $-k_i-1<Re(\mu_i)<-k_i+1$.
It means that for every test function $\varphi$, $\mu\mapsto u_\mu(\varphi)$
is a \textbf{holomorphic germ} near $(-k_1,\dots,-k_n)$.

We restrict to a small polydisk near $(-k_1,\dots,-k_n)$ and 
by Lemma \ref{weakstrongholo}, $u_\mu$ admits a 
power series expansion $u_\mu=\sum_\alpha (\mu+k)^\alpha u_\alpha $ near 
$(-k_1,\dots,-k_n)$
where $\alpha=(\alpha_1,\dots,\alpha_n)\in \mathbb{N}^n$ is a multi--index 
and $u_\alpha$ are distributions.  
Finally, we deduce that
\begin{eqnarray*}
t_\mu&=& \left(\prod_{i=1}^n \frac{1}{\mu_i+k_i}\right)\sum_\alpha (\mu+k)^\alpha u_\alpha\\
&= & \sum_\alpha (\prod_{i=1}^n(\mu_i+k_i)^{\alpha_i-1}) u_\alpha. 
\end{eqnarray*}
\end{proof}

\section{The meromorphic family $\left(\prod_{j=1}^p
(f_j+i0)^{\lambda_j}\right)_{\lambda \in\mathbb{C}^p}$.}

Let $U$ be some open set in $\mathbb{R}^n$ and $f_1,\dots,f_p$ be some real valued
analytic functions on $U$. 
The goal of the first part of our paper is to
show that the family of distributions
$\left(\prod_{j=1}^p
(f_j+i0)^{\lambda_j}\right)_{\lambda \in\mathbb{C}^p}$ depends meromorphically on $\lambda$, our proof
relies on Hironaka's resolution of singularities.
Let us quote the content of the
resolution Theorem as it is stated
in Atiyah's paper~\cite[p.~147]{AtiyahHironaka}:
\begin{thm}\label{Hironaka}
Let $F\neq 0$ be a real analytic function  defined in a
neighborhood of $0\in\mathbb{R}^n$. Then there exists an open neighborhood $U$ of $0$, a real analytic manifold
$\tilde{U}$ and a proper analytic map $\varphi:\tilde{U}\mapsto U$ such that
\begin{enumerate}
\item $\varphi:\tilde{U}\setminus \{F\circ \varphi=0\}\mapsto U\setminus \{F=0\}$ is an isomorphism,
\item for each $p\in U$, there are local analytic coordinates $(y_1, \dots ,y_n)$ centered at $p$
so that, locally near $p$, we have
$$ F\circ \varphi=\varepsilon\prod y_i^{k_i} $$
where
$\varepsilon$
is an invertible analytic function and $k_i$ are non negative integers.
\end{enumerate}
\end{thm}

This Theorem is central for QFT applications since 
it explains why regularized Feynman amplitudes should
depend meromorphically on the regularization
parameter $\lambda$.

\begin{thm}\label{Atiyahmeromextension}
Let $U$ be some open set in $\mathbb{R}^n$ and $(f_1,\dots,f_p)$ be some real valued
analytic functions on $U$. 
For every $(k_1,\dots,k_p)\in\mathbb{N}^p$,
the map
$\left(\lambda_1,\dots,\lambda_p\right)\in\mathbb{C}^p \mapsto \prod_{j=1}^p\log^{k_j}(f_j+i0)(f_j+i0)^{\lambda_j}$
is meromorphic in $\mathbb{C}^p$ with value distribution. 
\end{thm}
\begin{proof}
We closely follow Atiyah's exposition~\cite{AtiyahHironaka} 
based on Hironaka's Theorem \ref{Hironaka}
of resolution of singularities. The proof is essentially
local hence we might reduce to a smaller open set $U$ on which
Theorem \ref{Hironaka} applies.

Step 1 note that
$\prod_{j=1}^p\log^{k_j}(f_j+i0)(f_j+i0)^{\lambda_j}=\frac{d}{d\lambda_1}^{k_1}\dots \frac{d}{d\lambda_p}^{k_p} \prod_{j=1}^p(f_j+i0)^{\lambda_j}$, therefore it suffices
to prove the claim for $\prod_{j=1}^p(f_j+i0)^{\lambda_j}$.

Step 2 recognize that for complex $\lambda$, we choose the determination of the $\log$
which gives the identity
\begin{eqnarray}
(f+i0)^{\lambda}=1_{\{f\geqslant 0\}} f^\lambda + 1_{\{f\leqslant 0\}}e^{i\pi\lambda} (-f)^\lambda.
\end{eqnarray}

Step 3 therefore by expanding brutally the product:
$$\prod_{j=1}^p(f_j+i0)^{\lambda_j}=\prod_{j=1}^p  \left(1_{\{f_j\geqslant 0\}} f_j^{\lambda_j} + 1_{\{f_j\leqslant 0\}}e^{i\pi\lambda_j} (-f_j)^{\lambda_j}\right) $$
$$=\sum_{\varepsilon\in\{-1,1\}^p}\prod_{j=1}^p \left(1_{\{\varepsilon_j f_j\geqslant 0\}} (\varepsilon_j)^{\lambda_j}\left(\varepsilon_j f_j\right)^{\lambda_j} \right)$$
we may reduce to the problem
of meromorphic extension 
of a product of the form
$$\prod_{j=1}^p g_j^{\lambda_j}1_{\Gamma},\text{ where }\Gamma=\underset{1\leqslant j\leqslant p}{\bigcap}\{ g_j\geqslant 0 \} $$
where $(g_j)_j$ are real analytic, $1_{\Gamma}$ is the
indicator function of the domain 
$\Gamma=\underset{1\leqslant j\leqslant p}{\bigcap}\{ g_j\geqslant 0 \}$
and all functions $g_j\geqslant 0$ on $\Gamma$.

Step 4. Following Atiyah, we shall apply Hironaka's Theorem
\ref{Hironaka} to the function $F=\prod_j g_j $
to resolve
simultaneously the collection of real analytic functions $(g_j)_j$.
Assume $\forall j, g_j\neq 0$.
Denote by $\Sigma=\bigcup_{j\in \{1,\dots,p\}}\{g_j=0\}$ 
the zero set of all the above functions.
Then there is a proper analytic map $\varphi:\tilde{U}\mapsto U$, 
coordinate functions $(y_i)_i$ on $\tilde{U}$ such that $\varphi^{-1}(\Sigma)=\{\prod_i y_i=0\}$,
$\varphi$ is a diffeomorphism from $\tilde{U}\setminus \{\prod_i y_i=0\}\mapsto U\setminus \Sigma$
and for all $j$, every pulled--back function $\varphi^*g_j$ 
has the form $\varepsilon(y)y^{\alpha^j} $ where $\alpha^j=
(\alpha^j_1,\dots,\alpha^j_n)$ is a multi--index, 
$y^{\alpha^j}=\prod_{i=1}^n y_i^{\alpha^j_i}$
and $\varepsilon$ does not vanish in some neighborhood of $0$. 

Step 5 the above means that each pulled--back function $\varphi^*g_j$ 
reads $\varphi^*g_j=\varepsilon_jy^{\alpha^j}$ hence
the pulled--back product $\varphi^*\left(\prod_{j=1}^p g_j^{\lambda_j}1_{\Gamma} \right)$ can be 
further be expressed as a finite sum
of products of the form: 
$$ \prod_{j=1}^p \left( \varepsilon_j y^{\alpha^j}\right)^{\lambda_j}1_{\Gamma}=
\prod_{j=1}^p  \varepsilon_j^{\lambda_j} \prod_{j=1}^py^{\alpha^j\lambda_j} 1_{\Gamma},\Gamma=\{  y^{\alpha^j}\geqslant 0,\forall j \}.$$
Dropping the factor $\prod_{j=1}^p  \varepsilon_j^{\lambda_j}$ which does not vanish near $0$
and is analytic for all $\lambda\in\mathbb{C}^p$ we are reduced to study
the singular term: 
\begin{eqnarray*}
\left(y^{\sum_{j=1}^p\alpha^j\lambda_j} 1_{\Gamma}\right)= 1_{\Gamma}\prod_{j=1}^p y^{\alpha^j\lambda_j} ,\Gamma=\{  y^{\alpha^j}\geqslant 0,\forall j \}
\end{eqnarray*}
where for every $j\in\{1,\dots,p\}$, $\alpha^j=(\alpha^j_{1},\dots,\alpha^j_{n})$ 
is a multi--index and $\lambda_j$ a complex number.
The above distribution is a typical example of hypergeometric
distributions.
And it is immediate to prove
that the above expression is meromorphic in $\lambda$
with value $\mathcal{D}^\prime(U)$ by successive integration
by parts as in Lemma \ref{polesbasicexample} (see also~\cite{Gelfand-ShilovI}) or 
by the existence of the functional equation
\begin{eqnarray*}
\frac{d}{dy}^\beta \left(y^{\sum_{j=1}^p\alpha^j\lambda_j} 1_{\Gamma}\right)&=&
\prod_{i=1}^n \left(\frac{d}{dy_i}\right)^{\beta_i}\left( 1_{\Gamma}\prod_{j=1}^p y^{\alpha^j\lambda_j}\right)\\  
&=&
\prod_{i=1}^n \frac{\Gamma(\sum_{j=1}^p\lambda_j\alpha^j_i)}{\Gamma(\sum_{j=1}^p\lambda_j\alpha^j_i-\beta_i)}
\left( y^{\sum_{j=1}^p\alpha^j\lambda_j-\beta} 1_{\Gamma} \right),
\end{eqnarray*}
and the poles come from the poles at negative integers of the Euler $\Gamma$
function.

Step 6 We admit that $\Sigma=\{g_j=0,h_j=0\}$ has null measure as a consequence of Lemma \ref{thinset}.

Step 7 Let $u_\lambda$ denote the pulled--back distribution $\varphi^*\prod_{j=1}^p g_j^{\lambda_j}1_{\Gamma}$
on $\tilde{U}$. Then for $Re(\lambda_j)_j$ large enough both distributions
$\varphi_*u_\lambda$ and $\prod_{j=1}^p g_j^{\lambda_j}1_{\Gamma}$ are holomorphic in $\lambda$
and coincide on $U\setminus\Sigma$. However when $Re(\lambda_j)_j$ are large enough,
both distributions are locally integrable and since $\Sigma$ has \textbf{null measure},
the equality $\varphi_*u_\lambda=\prod_{j=1}^p g_j^{\lambda_j}1_{\Gamma}$ 
holds in $L^1_{loc}(U)$ hence in $\mathcal{D}^\prime(U)$ and both  sides are holomorphic in $\lambda$ with value
$\mathcal{D}^\prime(U)$. Finally we proved in Step 5 that $ \left(y^{\sum_{j=1}^p\alpha^j\lambda_j} 1_{\Gamma}\right)\in \mathcal{D}^\prime(\tilde{U})$ extends meromorphically
in $\lambda\in\mathbb{C}^p$ 
hence so does $\varphi_*u_\lambda=\prod_{j=1}^p g_j^{\lambda_j}1_{\Gamma}$.
By uniqueness of the analytic continuation process, this proves the claim. 
\end{proof}

\subsubsection{More general examples of hypergeometric distributions.}

The next result refines on Theorem \ref{Atiyahmeromextension} and
concerns the location of the poles of the meromorphic continued distributions.
\begin{lemm}\label{hypergeomdistribaffinehyperplanes}
Let us work in $\mathbb{R}^n$ with coordinates $(y_1,\dots,y_n)$. 
Consider the meromorphic family of distributions:
\begin{eqnarray*}
\left(\left(y^{\sum_{j=1}^p\alpha^j\lambda_j} 1_{\Gamma}\right)= 1_{\Gamma} 
\prod_{i=1}^ny_i^{\sum_{j=1}^p\alpha^j_i\lambda_j}\right)_{\lambda\in\mathbb{C}^p} ,
\Gamma=\{  y^{\sum_{j=1}^p\alpha_j}\geqslant 0,\forall j\}
\end{eqnarray*}
where for every $j\in\{1,\dots,p\}$, $\alpha^j=(\alpha^j_{1},\dots,\alpha^j_{n})\in\mathbb{N}^{n}$ 
is a multi--index and $\lambda_j\in\mathbb{C},1\leqslant j\leqslant p$.
Then define the collection $(\mu_i)_{1\leqslant i\leqslant n}\in (\mathbb{N}^p)^n$ 
of linear functions
on $\mathbb{C}^p$: 
\begin{equation}
\left(\mu_i:\lambda \in\mathbb{C}^p\mapsto \sum_{j=1}^p\alpha^j_i\lambda_j\right)_{1\leqslant i\leqslant n}
\end{equation}
then:
\begin{enumerate}
\item the polar set $Z$ 
of the family $\left(y^{\sum_{j=1}^p\alpha^j\lambda_j} 1_{\Gamma}\right)_{\lambda}$
is contained in the 
\textbf{union of affine hyperplanes}
\begin{equation}
Z=\bigcup_{1\leqslant i\leqslant n,k\in\mathbb{N}^*} \{\lambda \text{ s.t. }\mu_i(\lambda)=-k\},
\end{equation}
\item in some neighborhood of any element 
$z=(z_1,\dots,z_p)\in \mathbb{Z}^n$ 
there is a neighborhood $\Omega\subset \mathbb{C}^p$ of $z$, 
some distributions $(u_\beta)_{\beta\in\mathbb{N}^n}$ in $\mathcal{D}^\prime(U)$ such that:
\begin{equation}
\forall \lambda\in\Omega, \left(y^{\sum_{j=1}^p\alpha^j\lambda_j} 1_{\Gamma}\right)=\sum_{\beta\in\mathbb{N}^n} \prod_{i=1}^n (\mu_i(\lambda+z))^{\beta_i-1}u_\beta. 
\end{equation}
\end{enumerate}
\end{lemm}
\begin{proof}
It is an easy consequence of Lemmas \ref{polesbasicexample} and \ref{keyLaurentseries}
for $(\mu_i=\sum_{j=1}^p\alpha^j_i\lambda_j)_{1\leqslant i\leqslant n}$.
\end{proof}

The above result yields that
the hypergeometric distributions $ \left(y^{\sum_{j=1}^p\alpha^j\lambda_j} 1_{\Gamma}\right)_\lambda$
depend \textbf{meromorphically of $\lambda$ with linear poles}.
Finally, we can state an extended version of our 
main Theorem
\begin{thm}\label{Atiyahplus}
Let $U$ be some open set in $\mathbb{R}^n$ and $(f_1,\dots,f_p)$ be some real valued
analytic functions on $U$. 
Then 
the family of distributions 
$\prod_{j=1}^p(f_j+i0)^{\lambda_j}$
depends meromorphically on $\lambda$
with linear poles.
\end{thm}
\begin{proof}
In fact we prove the following
stronger result:
for all $x\in U$, there is a neighborhood $U_x$ of $x$, $n2^p$ linear functions with integer coefficients
$(\mu_{i,\varepsilon})_{1\leqslant i\leqslant n,\varepsilon\in \{-1,1\}^p}$ s.t. for all $z\in \mathbb{Z}^p$, there 
is a neighborhood $\Omega\subset \mathbb{C}^p$ of $z$ 
and
distributions $(u_{\beta,\varepsilon}), \beta\in\mathbb{N}^n,\varepsilon\in \{-1,1\}^p$ 
such that
\begin{equation}
\prod_{j=1}^p(f_j+i0)^{\lambda_j}|_{U_x}=\sum_{\varepsilon\in \{-1,1\}^p,\beta} u_{\beta,\varepsilon}\prod_{i=1}^n\mu_{i,\varepsilon}(\lambda + z)^{\beta_i-1}.
\end{equation} 
The result follows from Step 3 of the proof of Theorem \ref{Atiyahmeromextension}
where we decomposed $(\prod_{j=1}^p(f_j+i0)^{\lambda_j})$ as a sum
of $2^p$ elementary distributions of the form
$\prod_{j=1}^p g_j^{\lambda_j}1_\Gamma $ where every elementary
distribution $\prod_{j=1}^p g_j^{\lambda_j}1_\Gamma $ is the pushforward
by the resolution $\varphi$ of a
hypergeometric distribution of the form 
studied in Lemma \ref{hypergeomdistribaffinehyperplanes}. 
\end{proof}

The main result of the above Theorem
is the existence of a natural
Laurent series expansion in $(\lambda_1,\dots,\lambda_p)\in\mathbb{C}^p$ 
for the family 
$\prod_{j=1}^p(f_j+i0)^{\lambda_j}$.

\subsubsection{Appendix to section $1$: analytic sets have measure zero.}
We give here the key easy Lemma which states that the zero set
of a non zero real valued analytic
function has measure zero on $U$.
\begin{lemm}\label{thinset}
Let $F$ be a nonzero real analytic function on $U\subset \mathbb{R}^n$ 
then $\{F=0\}$ has zero Lebesgue measure.
\end{lemm}
\begin{proof}
The proof can be found in Federer~\cite{federer1969geometric}, 
but we sketch 
a simple proof following Atiyah~\cite{AtiyahHironaka} 
based on Hironaka's resolution of singularities.
It suffices to show that
near any point $x\in\{F=0\}\cap U$
there is some neighborhood $V_x$
of $x$ s.t. $V_x\cap \{F=0\}$ has measure zero.
Then it follows by paracompactness of $U$ that
$\{f=0\}\cap U$ can be covered by a countable number
of zero measure sets hence it has measure zero !
Locally near any $x\in U\cap \{F=0\}$, there is a proper analytic map $\varphi:\tilde{U}\subset\mathbb{R}^n \mapsto U$
such that the set $\tilde{\Sigma}=\varphi^{-1}\left( \{F=0\} \right)$
is contained in the coordinate cross of the
form $D=\{\prod_{i=1}^n t_i=0 \}$ and the set $\tilde{\Sigma}\subset D$ 
has zero measure since $D$ has measure zero. 
Therefore by~\cite[Proposition 1.3 p.~30]{GuilleminGolubitsky}, its
image by the $C^1$ map $\varphi$ has measure zero in particular it contains
$\{F=0\}\subset \varphi(D)$ which therefore has zero measure.
\end{proof}

\section{The main construction.}

The main problem of renormalization in QFT is to
define $\prod_{j=1}^p(f_j+i0)^{-k_j}$ for values of $k_j$
which are \textbf{positive integers} which boils down to
evaluate the \textbf{meromorphic} family $\prod_{j=1}^p(f_j+i0)^{\lambda_j}$
exactly at its poles.
Motivated by exciting recent works of Paycha--Guo--Zhang~\cite{PaychaZhang2015},
we follow in this section their definition 
of regularization and construct an abstract framework in which
one can regularize meromorphic functions with integral linear poles.
This construction will be used in the second part of our paper 
to renormalize quantum field theories.
The philosophy is to introduce
as many complex variables in our problem as there are
propagators and renormalize with meromorphic functions
with integral linear poles of an arbitrary number of variables.

\subsection{Algebras of cylindrical functions.}
Our goal is to construct an algebra of functions
$\mathcal{M}_{k}(\mathbb{C}^{\mathbb{N}})$ depending
on arbitrary number of complex variables $(\lambda_1,\dots,\lambda_p)$
which contains all meromorphic germs
obtained by meromorphic regularization of the first section. More precisely,
for all real analytic functions $(f_1,\dots,f_p)$ on some open set
$U$, for all test function $\varphi\in\mathcal{D}(U)$, the meromorphic germ $\lambda\mapsto \prod_{i=1}^j(f_j+i0)^{\lambda_j}(\varphi)$ at $(k_1,\dots,k_p)$
whose existence is guaranteed by Theorem \ref{Atiyahplus} 
is contained in the algebra $\mathcal{M}_{k}(\mathbb{C}^{\mathbb{N}})$.
We also construct a subalgebra $\mathcal{O}_{k}(\mathbb{C}^{\mathbb{N}})$ of $\mathcal{M}_{k}(\mathbb{C}^{\mathbb{N}})$
which contains all regular elements i.e. holomorphic germs
$f(\lambda_1,\dots,\lambda_p)\in \mathcal{M}_{k}(\mathbb{C}^{\mathbb{N}})$ whose limit exists at $(k_1,\dots,k_p)$.

Let us consider the space $\mathbb{C}^{\mathbb{N}}$ of 
sequences of complex numbers and a fixed sequence of integers $k\in\mathbb{Z}^{\mathbb{N}}$.
We construct an algebra
of cylindrical functions on $\mathbb{C}^{\mathbb{N}}$ as follows.
Let $p$ be a fixed integer. Let $k_{\leqslant p}=(k_1,\dots,k_p)$ be the first $p$ coefficients 
of the sequence $k$ viewed as an element 
of $\mathbb{C}^p$ then we define
two algebras $\mathcal{O}_{k_{\leqslant p}}(\mathbb{C}^p)$ and $\mathcal{M}_{k_{\leqslant p}}(\mathbb{C}^p)$
of germs of functions at $k_{\leqslant p}=(k_1,\dots,k_p)$.
\begin{defi}
$\mathcal{O}_{k_{\leqslant p}}(\mathbb{C}^p)$ is the
algebra of
holomorphic germs
$f$ at $k_{\leqslant p}$.
$\mathcal{M}_{k_{\leqslant p}}(\mathbb{C}^p)$ is the
algebra of meromorphic germs
at $k_{\leqslant p}$ with linear poles,  
$f$ belongs to $\mathcal{M}_{k_{\leqslant p}}(\mathbb{C}^p)$
if there are $m$ integral vectors
$(L_i)_{1\leqslant i\leqslant m}\in (\mathbb{N}^p)^m$ 
such that
\begin{equation}
\lambda\mapsto f(\lambda)(\prod_{i=1}^m L_i(\lambda+k))  
\end{equation}
is a holomorphic germ at $k_{\leqslant p}=(k_1,\dots,k_p)\in\mathbb{C}^p$.
\end{defi}

For all integer $p$, a germ $f(\lambda_1,\dots,\lambda_p)$ can always be viewed as a function
of the $p+1$ variables $(\lambda_1,\dots,\lambda_{p+1})$ which does not depend on the last variable
$\lambda_{p+1}$. It follows that
there are obvious inclusions $\mathcal{O}_{k_{\leqslant p}}(\mathbb{C}^p)\hookrightarrow \mathcal{O}_{k_{\leqslant p+1}}(\mathbb{C}^{p+1}) $ and $\mathcal{M}_{k_{\leqslant p}}(\mathbb{C}^p)\hookrightarrow \mathcal{M}_{k_{\leqslant p+1}}(\mathbb{C}^{p+1})$
which imply the existence of the
inductive limits $\mathcal{O}_{k}(\mathbb{C}^{\mathbb{N}})=\underset{\rightarrow}{\lim} \mathcal{O}_{k_{\leqslant p}}(\mathbb{C}^p)$
and $\mathcal{M}_{k}(\mathbb{C}^{\mathbb{N}})=\underset{\rightarrow}{\lim} \mathcal{M}_{k_{\leqslant p}}(\mathbb{C}^p)$.
It is simple to check the following
properties
\begin{prop}
Both
$\mathcal{O}_{k}(\mathbb{C}^{\mathbb{N}}),\mathcal{M}_{k}(\mathbb{C}^{\mathbb{N}})$ are algebras,
$\mathcal{M}_{k}(\mathbb{C}^{\mathbb{N}})$ is a $\mathcal{O}_{k}(\mathbb{C}^{\mathbb{N}})$ module and
contains
$\mathcal{O}_{k}(\mathbb{C}^{\mathbb{N}})$
as a subalgebra.
\end{prop}

\subsection{A projector and the factorization property.}
By definition of the
inductive limit,
elements of $\mathcal{M}_{k}(\mathbb{C}^{\mathbb{N}})$ are 
meromorphic germs with integral linear poles 
depending on a finite number of variables.
\subsubsection{The notion of independence.}
We will say that two elements $(f,g)\in \mathcal{M}_{k}(\mathbb{C}^{\mathbb{N}})^2$
are \emph{independent} if they depend on different sets of variables.
It follows that if $(f,g)$ are independent, then
they satisfy condition $(c)$ of \cite[Theorem 4.4]{PaychaZhang2015}.

\subsubsection{Subtraction of poles and projectors.}

Recall that our final goal is to evaluate $\prod_{j=1}^p(f_j+i0)^{k_j}$ for values of $k_j$
which are \textbf{negative integers} which requires to subtract the poles
of elements from $\mathcal{M}_{k}(\mathbb{C}^{\mathbb{N}})$. 
An elegant way to reformulate the operation
of subtraction of poles is in terms of a projection
\begin{equation}
\pi: \mathcal{M}_{k}(\mathbb{C}^{\mathbb{N}})\mapsto \mathcal{O}_{k}(\mathbb{C}^{\mathbb{N}}).
\end{equation}

\subsubsection{The factorization condition.}

\begin{defi}\label{factconditionmeroabstract}
A projection $\pi:\mathcal{M}_{k}(\mathbb{C}^{\mathbb{N}})\mapsto \mathcal{O}_{k}(\mathbb{C}^{\mathbb{N}})$ 
satisfies the factorization condition if
for all $(f,g)\in \mathcal{M}_{k}(\mathbb{C}^{\mathbb{N}})^2$, if $f$ and $g$
are independent then
\begin{equation}
\pi(fg)=\pi(f)\pi(g).
\end{equation}
\end{defi}
\subsection{The main existence Theorem.}

In this subsection, we explain the existence of a projection
which satisfies the factorization condition.
This is exactly the content
of \cite[Theorem 4.4]{PaychaZhang2015}.
Let us state their Theorem in our notations:
\begin{thm}{\textbf{Guo--Paycha--Zhang}}\label{PaychaZhang}
\\
Let $Q$ be the quadratic form defined
on all the vector spaces $\mathbb{C}^p$ for $p\in\mathbb{N}$
as $Q(z_1,\dots,z_p)=\sum_{i=1}^p \vert z_i\vert^2$.
\begin{enumerate}
\item For all $p\in\mathbb{N}$, we have the direct sum decomposition
\begin{equation}
\mathcal{M}_{k_{\leqslant p}}(\mathbb{C}^p)=\mathcal{O}_{k_{\leqslant p}}(\mathbb{C}^p)\oplus\mathcal{M}_{-,k_{\leqslant p}}(\mathbb{C}^p)
\end{equation}
where the space $\mathcal{M}_{-,k_{\leqslant p}}(\mathbb{C}^p)$
contains all singular functions, in 
particular any element $f=\frac{h}{L_1\dots L_n}\in\mathcal{M}_{k_{\leqslant p}}(\mathbb{C}^p) $
can be written as a sum 
\begin{equation}
f=\sum_i\frac{h_i(\ell_{i(n_i+1)},\dots,\ell_{ip})}{L_{i1}^{s_{i1}}\dots L_{in_i}^{s_{in_i}}}+\phi_i(L_{i1},\dots,L_{in_i},\ell_{i(n_i+1)},\dots,\ell_{ip})
\end{equation}
where for each $i$, $(s_{i1},\dots,s_{in_i})\in\mathbb{N}^{n_i}$, 
the collection of linear forms $(L_{i1},\dots,L_{in_i})$ is a linearly independent
subset of $(L_1,\dots,L_n)$, the collection of linear forms
$(\ell_{i(n_i+1)},\dots,\ell_{ip})$ is a basis of the orthogonal complement (for $Q$)
of the subspace spanned by the $(L_{i1},\dots,L_{in_i})$, $h_i$ is holomorphic
in the independent variables $\ell_i$ so that $\frac{h_i(\ell_{i(n_i+1)},\dots,\ell_{ip})}{L_{i1}^{s_{i1}}\dots L_{in_i}^{s_{in_i}}}$ belongs to $\mathcal{M}_{-,k_{\leqslant p}}(\mathbb{C}^p)$.
\item The coefficients 
\begin{equation}
(h_i,\phi_i)_i
\end{equation} 
depend linearly on finite number
of partial derivatives of $h$
\item Taking a direct limit yields
\begin{equation}
\mathcal{M}_{k}(\mathbb{C}^{\mathbb{N}})=\mathcal{O}_{k}(\mathbb{C}^{\mathbb{N}})\oplus\mathcal{M}_{-,k}(\mathbb{C}^{\mathbb{N}})
\end{equation}
\item The projection map
$\pi: \mathcal{M}_{k}(\mathbb{C}^{\mathbb{N}})\mapsto\mathcal{O}_{k}(\mathbb{C}^{\mathbb{N}})$
onto $\mathcal{O}_{k}(\mathbb{C}^{\mathbb{N}})$ along the subspace
$\mathcal{M}_{-,k}(\mathbb{C}^{\mathbb{N}})$ factorizes
on independent functions. 
If $(f,g)\in \mathcal{M}_{k}(\mathbb{C}^{\mathbb{N}})^2$
are independent then
\begin{equation}
\pi(fg)=\pi(f)\pi(g).
\end{equation}
\end{enumerate}
\end{thm}
\begin{proof}
We refer to~\cite{PaychaZhang2015} for the proof of this beautiful Theorem but
will only show the property $(2)$ which explains how to define
$\pi$ in an algorithmic fashion closely following the original proof in~\cite{PaychaZhang2015}.
Thanks to \cite[Lemma 4.1]{PaychaZhang2015}, without loss of generality we can reduce the proof to germs of functions
of the type
$$f=\frac{h}{L_1^{s_1}\dots L_m^{s_m}}$$
with $h$ holomorphic, linearly independent linear forms $(L_1,\dots,L_m)$ and 
$(s_1,\dots,s_m)$
positive integers. The system $(L_1,\dots,L_m,\ell_{m+1},\dots,\ell_{p})$ is a coordinate 
system on $\mathbb{C}^p$.
Consider a partial Taylor expansion with remainder of $h$
in the first $m$ coordinates $(L_1,\dots,L_m)$:
\begin{eqnarray*}
h=\sum_{k<s} \frac{L_1^{k_1}\dots L_m^{k_m}}{k_1!\dots k_m!}\partial^{k_1}_{L_1}\dots \partial^{k_m}_{L_m}h(0,\ell_{m+1},\dots,\ell_p)+L_1^{s_1}\dots L_m^{s_m}\phi(L_1,\dots,L_m,\ell_{m+1},\dots,\ell_{p})
\end{eqnarray*}
where $\phi$ is \textbf{holomorphic} and $k=(k_1,\dots,k_m)<s=(s_1,\dots,s_m)$ 
means that for some $i\in\{1,\dots,m\}$,
$k_i<s_i$ and $k_j\leqslant s_j,\forall j\neq i$.
Then it follows that
\begin{eqnarray*}
\frac{h}{L_1^{s_1}\dots L_m^{s_m}}=\sum_{k<s} \frac{1}{k_1!\dots k_m!} \frac{\partial^{k_1}_{L_1}\dots \partial^{k_m}_{L_m}h(0,\ell_{m+1},\dots,\ell_p)}{L_1^{s_1-k_1}\dots L_m^{s_m-k_m}}
+ \phi(L_1,\dots,L_m,\ell_{m+1},\dots,\ell_{p})
\end{eqnarray*}
hence:
\begin{eqnarray}
\pi(\frac{h}{L_1^{s_1}\dots L_m^{s_m}})&=&\frac{h}{L_1^{s_1}\dots L_m^{s_m}}-\sum_{k<s} \frac{1}{k_1!\dots k_m!} \frac{\partial^{k_1}_{L_1}\dots \partial^{k_m}_{L_m}h(0,\ell_{m+1},\dots,\ell_p)}{L_1^{s_1-k_1}\dots L_m^{s_m-k_m}}.
\end{eqnarray}
\end{proof}
\begin{thm}
Let $U$ be an open set in $\mathbb{R}^n$, $\Omega\subset \mathbb{C}^p$ open and 
$(t_\lambda)_{\lambda\in\Omega}$ a meromorphic family with linear poles
at $k\in\Omega$ with value $\mathcal{D}^\prime(U)$.
Then the family $\pi(t_\lambda)_\lambda$
defined
as
\begin{equation}
\forall \varphi\in\mathcal{D}(U),\pi(t_\lambda)(\varphi)=\pi(t_\lambda(\varphi))
\end{equation}
is holomorphic at $k$ with value $\mathcal{D}^\prime(U)$.
\end{thm}
\begin{proof}
Proposition \ref{weakstrongholo}
implies that if $h(\lambda)_{\lambda}$ is holomorphic
in $\lambda\in\mathbb{C}^p$ with value $\mathcal{D}^\prime(U)$
then the truncated Laurent series
$$\sum_{k\geqslant s}  \frac{L_1^{k_1}\dots L_m^{k_m}}{k_1!\dots k_m!} \partial^{k_1}_{L_1}\dots 
\partial^{k_m}_{L_m}h_k(0,\ell_{m+1},\dots,\ell_p) $$
absolutely converge
in $\mathcal{D}^\prime(U)$ by Cauchy's bound (\ref{Cauchybound}).
Then dividing the above truncated Laurent series
by $L_1^{s_1}\dots L_m^{s_m}$ and by definition of the projection $\pi$ of Theorem \ref{PaychaZhang},
we find that:
\begin{eqnarray}
\pi(\frac{h}{L_1^{s_1}\dots L_m^{s_m}})&=&\phi(L_1,\dots,L_m,\ell_{m+1},\dots,\ell_{p})\\
&=&\frac{h}{L_1^{s_1}\dots L_m^{s_m}}-\sum_{k<s} \frac{1}{k_1!\dots k_m!} \frac{\partial^{k_1}_{L_1}\dots \partial^{k_m}_{L_m}h(0,\ell_{m+1},\dots,\ell_p)}{L_1^{s_1-k_1}\dots L_m^{s_m-k_m}}
\end{eqnarray}
is also holomorphic 
in $\lambda\in\mathbb{C}^p$ with value $\mathcal{D}^\prime(U)$.
\end{proof}
The above Theorem allows us to define a renormalization operator $\mathcal{R}_\pi$ of
the complex powers 
$\prod_{j=1}^p(f_j+i0)^{k_j}$ for $k_j\in-\mathbb{N}^*$ as follows:
\begin{defi}\label{defrenormoperator}
For all test function $\varphi\in\mathcal{D}(U)$,
\begin{equation}
\mathcal{R}_\pi(\prod_{j=1}^p(f_j+i0)^{k_j})(\varphi)=\pi( \lambda\mapsto \prod_{j=1}^p(f_j+i0)^{\lambda_j}(\varphi))(k).
\end{equation}
\end{defi}
\subsubsection{The fundamental tensor factorization property.}
It is immediate by construction that the renormalization operator $\mathcal{R}_\pi$
satisfies the following factorization identity:
let $U,V$ be open sets in $\mathbb{R}^{n_1},\mathbb{R}^{n_2}$
respectively and
$f_1,\dots,f_p$ (resp $g_1,\dots,g_p$) real analytic functions on $U$
(resp $V$) then 
\begin{equation}\label{mostgeneralfactorizationidentityrenormmaps}
\mathcal{R}_\pi\left(f_1^{k_1}\dots f_p^{k_p}  g_1^{l_1}\dots g_p^{l_p}\right)=
\mathcal{R}_\pi\left(f_1^{k_1}\dots f_p^{k_p}\right)\otimes \mathcal{R}_\pi\left(g_1^{l_1}\dots g_p^{l_p}\right).
\end{equation}
where the tensor product $\otimes$ is the exterior
tensor product: $\mathcal{D}^\prime(U)\otimes \mathcal{D}^\prime(V)\mapsto\mathcal{D}^\prime(U\times V) $.

%

\section{$u=0$ theorem.}

\subsubsection{Motivation for these Theorems.}

In QFT, we need to multiply 
Feynman propagators, which are distributions,
in order to define Feynman amplitudes.
The control of their wave front sets give sufficient conditions
under which one can multiply these distributions.
Therefore we are let to study the wave front set
of the family $(f+i0)^\lambda$.
Unfortunately to bound the wave front of 
the family $(f+i0)^\lambda$, we must bound wave front sets
of products of distributions which are well defined
but fail to satisfy H\"ormander's transversality condition
on wave front sets.
The $u=0$ Theorem which originates
from the work of Iagolnitzer 
will help us 
give bounds on wave front sets
of products of distributions $(uv)$ which are well defined but 
whose wave front set 
fail to satisfy
the transversality condition
$WF(u)\cap -WF(v)=\emptyset$
of H\"ormander.
\subsection{Products in Sobolev spaces.}
The goal of this part is to recall some well known
results on Sobolev spaces.
We denote by $H^s(\mathbb{R}^d)$ the usual
$L^2$ Sobolev space and $t\in\mathcal{D}^\prime(\mathbb{R}^d)$ 
belongs to $H_{loc}^s(\mathbb{R}^d)$ if for all test function
$\varphi\in\mathcal{D}(\mathbb{R}^d)$, $t\varphi\in H^s(\mathbb{R}^d)$.
Recall that the usual multiplication
of smooth functions extends 
naturally to $H_{loc}^{s_1}(\mathbb{R}^d)\times H_{loc}^{s_2}(\mathbb{R}^d)$
when $s_1+s_2\geqslant 0$. 
Indeed
\begin{lemm}\label{Prodsobolev}
Let $(u,v)\in H_{loc}^{s_1}(\mathbb{R}^d)\times H_{loc}^{s_2}(\mathbb{R}^d)$ 
for 
$s_1+s_2\geqslant 0, s_1\leqslant 0 \leqslant s_2$ then
the product $uv$ makes sense in $\mathcal{D}^\prime(\mathbb{R}^d)$
and for all test function $\varphi$, 
the Fourier transform
$\widehat{uv\varphi^2}$ is 
well defined by an \textbf{absolutely convergent} 
convolution integral which satisfies the bound:
\begin{equation}\label{estimateuvsobolev}
\vert\widehat{uv\varphi^2}(\xi)\vert\leqslant \int_{\mathbb{R}^d}d^d\eta\vert \widehat{u\varphi}(\xi-\eta)\widehat{v\varphi}(\eta)\vert \leqslant (1+\vert\xi\vert)^{-s_1} \Vert u\varphi\Vert_{H^{s_1}} 
\Vert v\varphi\Vert_{H^{s_2}}.
\end{equation}
\end{lemm}
\begin{proof}
Let $\varphi$ be a test function then
from $\widehat{uv\varphi^2}=\widehat{u\varphi}*\widehat{v\varphi}$, 
we deduce the estimates:
\begin{eqnarray*}
\vert \widehat{uv\varphi^2} \vert(\xi) &\leqslant &\int_{\mathbb{R}^d}d^d\eta\vert \widehat{u\varphi}(\xi-\eta)\widehat{v\varphi}(\eta)\vert\\
&\leqslant &\sup_{\eta}(1+\vert\xi-\eta\vert)^{-s_1}(1+\vert\eta\vert)^{-s_2} 
\int_{\mathbb{R}^d} d^d\eta
\vert (1+\vert\xi-\eta\vert)^{s_1}\widehat{u\varphi}(\xi-\eta)(1+\vert\eta\vert)^{s_2}
\widehat{v\varphi}(\eta)\vert\\
&\leqslant & \sup_{\eta}\{\frac{(1+\vert\xi- \eta\vert)^{-s_1}}{(1+\vert\eta\vert)^{-s_1}}(1+\vert\eta\vert)^{-(s_1+s_2)}\} 
\Vert u\varphi\Vert_{H^{s_1}} 
\Vert v\varphi\Vert_{H^{s_2}}\text{ by Cauchy--Schwartz} \\
&\leqslant & \sup_{\eta}\{\frac{(1+\vert\xi- \eta\vert)^{-s_1}}{(1+\vert\eta\vert)^{-s_1}}\} 
\Vert u\varphi\Vert_{H^{s_1}} 
\Vert v\varphi\Vert_{H^{s_2}}\text{ since }s_1+s_2\geqslant 0\\
&\leqslant &(1+\vert\xi\vert)^{-s_1} \Vert u\varphi\Vert_{H^{s_1}} 
\Vert v\varphi\Vert_{H^{s_2}}\text{ since }\frac{(1+\vert\xi- \eta\vert)^{-s_1}}{(1+\vert\eta\vert)^{-s_1}(1+\vert\xi\vert)^{-s_1}}\leqslant 1 .
\end{eqnarray*}
The above
shows that
$\widehat{uv\varphi^2}$ is 
well defined by an \textbf{absolutely convergent} 
convolution integral and has polynomial growth in $\xi$.
Hence $uv\varphi^2=\mathcal{F}^{-1}\left(\widehat{uv\varphi^2}\right)$ is a well defined
distribution in $\mathcal{E}^\prime(\mathbb{R}^d)$.
Now let $(\varphi_j)_j$ be a partition of unity
of $\mathbb{R}^d$
such that $\forall j, \varphi_j\in\mathcal{D}(\mathbb{R}^d)$
and $\sum_j\varphi_j^2=1$ where the sum is locally finite. 
Then the identity
\begin{eqnarray}
uv=\sum_j(uv\varphi_j^2)=\sum_j\mathcal{F}^{-1}\left(\widehat{uv\varphi_j^2}\right)
\end{eqnarray}
shows that the product 
$uv$ makes sense in $\mathcal{D}^\prime(\mathbb{R}^d)$.
\end{proof}

Denote by $H_0^s\left(\Omega \right)$ 
the space of functions in $H^s(\mathbb{R}^d)$ 
whose support is contained in $\Omega$ endowed with the
topology of the Sobolev space $H^s(\mathbb{R}^d)$.
\begin{prop}\label{multsobolevcontinuous}
Let $s_1+s_2\geqslant 0, s_1\leqslant 0 \leqslant s_2$. Then
the multiplication $(u,v)\in H^{s_1}_0(\Omega)\times H^{s_2}_0(\Omega)\mapsto 
(uv)\in \mathcal{E}^\prime(\mathbb{R}^d)$ is bilinear continuous
where $\mathcal{E}^\prime(\mathbb{R}^d)$ is endowed
with the strong topology.
\end{prop}  
\begin{proof}
Recall that the strong topology
of $\mathcal{E}^\prime(\mathbb{R}^d)$ is the topology of uniform
convergence on 
bounded sets of $C^\infty(\mathbb{R}^d)$.
Let $B$ be some arbitrary bounded set
in $C^\infty(\mathbb{R}^d)$ for its Fr\'echet 
space topology.
Pick a test function $\chi\in\mathcal{D}(\mathbb{R}^d)$
s.t. $\chi=1$ on $\Omega$. Then $\forall\varphi\in B,$
\begin{eqnarray*}
\vert \left\langle uv,\varphi \right\rangle \vert &=&
\vert\int_{\mathbb{R}^d} \widehat{uv \chi^2}(\xi) \widehat{\varphi\chi}(\xi) d^d\xi\vert\\
&\leqslant &\Vert u\chi\Vert_{H^{s_1}} 
\Vert v\chi\Vert_{H^{s_2}}  \int_{\mathbb{R}^d} d^d\xi
(1+\vert\xi\vert)^{-s_1}\vert \widehat{\varphi\chi}(\xi)\vert\\
&\leqslant & \Vert u\chi\Vert_{H^{s_1}} 
\Vert v\chi\Vert_{H^{s_2}}  \int_{\mathbb{R}^d} d^d\xi
(1+\vert\xi\vert)^{-d-1}  \vert (1+\vert\xi\vert)^{d+1-s_1} \widehat{\varphi\chi}(\xi)\vert \\
&\leqslant & C \Vert u\Vert_{H^{s_1}} 
\Vert v\Vert_{H^{s_2}} \sup_{x\in\Omega,\vert\alpha\vert\leqslant m}\vert \varphi(x)\vert
\end{eqnarray*}
for $m\geqslant d+1-s_1$ and where $C$ does not depend on $\varphi$.
But $\sup_{\varphi\in B}  \sup_{x\in\Omega,\vert\alpha\vert\leqslant m}\vert \varphi\vert<+\infty$
therefore $\exists C>0, \sup_{\varphi\in B}\vert \left\langle uv,\varphi \right\rangle \vert\leqslant C \Vert u\Vert_{H^{s_1}} 
\Vert v\Vert_{H^{s_2}}$ which yields the desired result.
\end{proof}

\subsubsection{The Fourier transform of compactly 
supported
Sobolev distributions.}
We will need to compare
$C^k$ norms and Sobolev norms
and we also often use the following local embeddings:
\begin{prop}\label{embeddings}
Let $\Omega$ be some bounded open set.
Denote by $H_0^s\left(\Omega \right)$ (resp $C^k_0(\Omega)$)
the space of functions in $H^s(\mathbb{R}^d)$ (resp $C^k(\mathbb{R}^d)$)
whose support is contained in $\Omega$.  
If $k+\frac{d}{2}<s$ then the map:
\begin{eqnarray}\label{sobembedding}
u\in H_0^s\left(\Omega \right) \longmapsto u\in C^k_0(\Omega)
\end{eqnarray}
is continuous.

Conversely let $k\in\mathbb{N}$ then for all $s$ such that $s+\frac{d}{2}<k$ 
the map:
\begin{eqnarray}\label{ckembedding}
u\in C^k_0(\Omega)\longmapsto u\in H_0^s\left(\Omega \right)
\end{eqnarray}
is continuous.
\end{prop}
\begin{proof}
The embedding \ref{sobembedding} results from the elementary
estimates:
\begin{eqnarray*}
\forall x\in \Omega, \vert \partial^ku(x) \vert\leqslant 
\int_{\mathbb{R}^d}d^d\xi \vert\xi\vert^k\vert \widehat{u}(\xi) \vert
\leqslant \int_{\mathbb{R}^d}d^d\xi (1+ \vert\xi\vert)^k\vert \widehat{u}(\xi) \vert\\
\leqslant \int_{\mathbb{R}^d}d^d\xi (1+ \vert\xi\vert)^{s}\vert \widehat{u}(\xi)\vert(1+ \vert\xi\vert)^{k-s} 
\leqslant \Vert u\Vert_{H^s} \left(\int_{\mathbb{R}^d}d^d\xi (1+ \vert\xi\vert)^{2(k-s)}\right)^{\frac{1}{2}}
\end{eqnarray*}
where the last estimate follows from Cauchy Schwartz inequality and the fact
that $\vert(1+ \vert\xi\vert)^{k-s} \in L^2(\mathbb{R}^d)$ since $k-s<-\frac{d}{2}$.

 Conversely if $k>\frac{d}{2}$ then: 
\begin{eqnarray*} 
u\in C^k_0(\Omega) \implies \vert(1+\vert\xi\vert)^{k}\widehat{u}(\xi)\vert\leqslant C \sup_{x\in\Omega,\vert\alpha\vert\leqslant k}\vert u(x)\vert\\
\implies \forall\varepsilon>0, \exists C^\prime>0, \Vert (1+\vert\xi\vert)^{k-(\frac{d}{2}+\varepsilon)}\widehat{u}(\xi)\Vert_{ L^2(\mathbb{R}^d)}\leqslant  C^\prime \sup_{x\in\Omega,\vert\alpha\vert\leqslant k}\vert u(x)\vert.
\end{eqnarray*}
Finally this means $\forall k\geqslant 0$, 
$C^k_0(\Omega)$ injects continuously in $H^s_0(\Omega),\forall s<k-\frac{d}{2}$.
\end{proof}

The embedding \ref{ckembedding} will be important
for us since it states that a very regular function
in $C^k$ for large $k$ will belong to all Sobolev space
$H^s$ for $s<k-\frac{d}{2}$ and that the embedding is continuous. 
The next lemma
gives us a way to control
weighted norms of Fourier
transform of compactly
supported distributions
of Sobolev regularity
$H^s(\mathbb{R}^d)$.
\begin{lemm}\label{sobolevfouriersupnorm}
Let $u$ be a distribution in $H^s(\mathbb{R}^d)$ and $B$ the ball
of radius $R$. There exists $M>0$ s.t. for all $u$ supported
in $B$, $u$ satisfies the estimate 
\begin{equation}
\exists M\geqslant 0,\vert\widehat{u}(\xi)\vert\leqslant M\Vert u\Vert_{H^s(\mathbb{R}^d)} 
(1+\vert\xi\vert)^{k}
\end{equation}
for all $k\geqslant 0$ if $s\geqslant 0$, $s+k>\frac{d}{2}$ if $s<0$.
\end{lemm}
\begin{proof} 
First note that $\widehat{u}$ is real analytic by Paley--Wiener--Schwartz.
If $s\geqslant 0$ then $u$ is a compactly supported $L^2$ function, hence a distribution
of order $0$ and thus $k=0$ which means that $\widehat{u}$ is bounded.
Moreover, we have the explicit estimate:
\begin{eqnarray*}
\vert\widehat{u}(\xi)\vert=\vert u(\chi e^{i\langle .,\xi\rangle})\vert &\leqslant &
\Vert u\Vert_{L^2(\mathbb{R}^d)} \Vert \chi e^{i\langle .,\xi\rangle}\Vert_{L^2(\mathbb{R}^d)} \\
&\leqslant & (2R)^{\frac{d}{2}} \Vert u\Vert_{L^2(\mathbb{R}^d)}\leqslant (2R)^{\frac{d}{2}} \Vert u\Vert_{H^s(\mathbb{R}^d)}.
\end{eqnarray*}
If $s<0$,
by duality of Sobolev spaces~\cite[Proposition 13.7]{Eskin}, we find that for all test function
$\varphi$:
$$\vert \left\langle u,\varphi \right\rangle \vert\leqslant \Vert u \Vert_{H^s}
 \Vert\varphi \Vert_{H^{-s}} .$$
Hence by the embedding \ref{ckembedding}, for all
$k$ satisfying $k>-s+\frac{d}{2}$ there exists $C>0$ s.t. :
$$\Vert\varphi \Vert_{H^{-s}(\mathbb{R}^d)}\leqslant  C\Vert\varphi \Vert_{C^k(\Omega)} $$
therefore:
\begin{eqnarray*}
\vert \left\langle u,\varphi \right\rangle \vert&\leqslant &
\Vert u \Vert_{H^s}
\Vert\varphi \Vert_{H^{-s}}\\
&\leqslant &
C \Vert u \Vert_{H_0^s(\Omega)}\Vert \varphi\Vert_{C^k_0(\Omega)}
\end{eqnarray*} 
therefore choosing
$\varphi=\chi e^{i\langle\xi,.\rangle}$
where $\chi\in\mathcal{D}(\mathbb{R}^d),\chi|_{B}=1$
yields
\begin{eqnarray*}
\vert\widehat{u}(\xi)\vert&=&\vert \left\langle u,\chi e^{i\langle\xi,.\rangle}\right\rangle \vert\\
&\leqslant &
C \Vert u \Vert_{H_0^s(\Omega)}\Vert \chi e^{i\langle\xi,.\rangle}\Vert_{C^k_0(\Omega)}\\
&\leqslant & C^\prime \Vert u \Vert_{H_0^s(\Omega)} (1+\vert\xi\vert)^k
\end{eqnarray*}
for some constant $C^\prime$ independent of
$u$.
\end{proof}

\subsection{The $\widehat{+}_i$ operation of Iagolnitzer.}
We first introduce the $\widehat{+}_i$ operation of Iagolnitzer
on closed conic sets. Actually,
this operation originates
from the $u=0$ Theorems
of Iagolnitzer~\cite{iagolnitzer1978theu} which aim to study the
analytic wave front set of products $uv$
s.t. $WF_A(u)$ and $WF_A(v)$ are not transverse.
 
\subsubsection{Definition.}
In what follows we define $\widehat{+}_i$ following Iagolnitzer~\cite{iagolnitzer1978theu}. Our definition
of $\widehat{+}_i$ is
weaker than the $\widehat{+}$ operation defined 
by Kashiwara--Schapira~\cite{KS}
and gives a larger conic set for the WF of the product.
Let $\Gamma_1,\Gamma_2$ be two closed conic sets
in $T^\bullet \mathbb{R}^d$, then
\begin{eqnarray*}
\Gamma_1\widehat{+}_i\Gamma_2=\{(x;\xi) \text{ s.t. }\exists \{(x_{1,n};\xi_{1,n}),(x_{2,n};\xi_{2,n})\}_{n\in\mathbb{N}} \in \left(\Gamma_1\times \Gamma_2\right)^{\mathbb{N}} , x_{i,n}\rightarrow x,\xi_{1,n}+\xi_{2,n}\rightarrow \xi, \xi\neq 0 \}
\end{eqnarray*}

\begin{lemm}
If $\Gamma_1\cap-\Gamma_2=\emptyset$
then $\Gamma_1\widehat{+}_i\Gamma_2=(\Gamma_1+\Gamma_2)\cup\Gamma_1\cup\Gamma_2 $.
\end{lemm}
\begin{proof}
The proof follows from the definition of $\widehat{+}_i$.
\end{proof}

\subsubsection{A $u=0$ Theorem.}

We want to show that
\begin{thm}
Let $(u,v)\in H_{loc}^{s_1}(\mathbb{R}^d)\times H_{loc}^{s_2}(\mathbb{R}^d)$ 
for 
$s_1+s_2\geqslant 0, s_1\leqslant 0 \leqslant s_2$ then
the product $uv$ makes sense in $\mathcal{D}^\prime(\mathbb{R}^d)$
and
\begin{eqnarray}
WF(uv)\subset WF(u)\widehat{+}_i WF(v).
\end{eqnarray}
\end{thm}
\begin{proof}
The existence of the product $uv$ in $\mathcal{D}^\prime$
follows from Lemma \ref{Prodsobolev}.
We use the notation of H\"ormander and denote by
$\Sigma(u\varphi)\subset\mathbb{R}^{n*}$ the closed cone
which is the complement
of the codirections where $\widehat{u\varphi}$
has fast decrease. We denote by $\pi_2$
the projection $(x;\xi)\in T^*\mathbb{R}^d\mapsto \xi\in\mathbb{R}^{d*}$. 
From H\"ormander~\cite[]{HormanderI}, the cone
$\Sigma(u\varphi)$
can be expressed
in terms of the wave front set of $u$: 
\begin{equation}
\Sigma(u\varphi)=\pi_2\left(WF(u)\cap T_{\text{supp }\varphi}^*\mathbb{R}^d\right).
\end{equation}

If $(x;\xi)\notin WF(u)\widehat{+}_i WF(v)$, then we claim
that there is a closed
conic neighborhood $V$ of $\xi$ and 
a small
ball $B_\varepsilon(x)$ centered at $x$ such that
for all $\varphi\in \mathcal{D}(B_\varepsilon(x))$,
\begin{eqnarray}\label{keyconditionWFsum}
\left(\left(\Sigma(u\varphi)\cup\{0\}\right)+\left(\Sigma(v\varphi)\cup\{0\}\right)\right)\cap V=\emptyset.
\end{eqnarray}
By contradiction assume the above claim
is not true. Then for all closed conic neighborhood $V$ of $\xi$ such that 
$\left(\{x\}\times V\right)\cap (WF(u)\widehat{+}_i WF(v)|_{x})=\emptyset$ where $WF(u)\widehat{+}_i WF(v)|_{x}$ 
lives in the fiber $T^*_x\mathbb{R}^d$,
there is some sequence $\varepsilon_n\rightarrow 0$
such that for every $n$, there are two elements
$(x_{1,n};\xi_{1;n})\in WF(u),(x_{2,n};\xi_{2;n})\in WF(v)$,
$(x_{1,n},x_{2,n})\in B_{\varepsilon_n}(x)^2$ 
such that $\xi_{1,n}+\xi_{2,n}\in V $. 
Therefore
we have a pair of sequences 
$(x_{1,n};\frac{\xi_{1;n}}{\vert\xi_{1,n}+\xi_{2,n}\vert})\in WF(u),(x_{2,n};
\frac{\xi_{2;n}}{\vert\xi_{1,n}+\xi_{2,n}\vert})\in WF(v)$
such that $\frac{\xi_{1,n}+\xi_{2,n}}{\vert\xi_{1,n}+\xi_{2,n}\vert}\in V\cap \mathbb{S}^{d-1}$
and $(x_{1,n},x_{2,n})\rightarrow (x,x)$. The set $V\cap \mathbb{S}^{d-1}$ is compact, therefore by extracting a subsequence,
we can assume that the sequence $\left(\frac{\xi_{1,n}+\xi_{2,n}}{\vert\xi_{1,n}+\xi_{2,n}\vert}\right)_n$ 
converges
to $\xi\in V$ which implies that 
$(x;\xi)\in \text{supp }\chi \times V$ and
$(x;\xi )\in WF(u)\widehat{+}_iWF(v)|_x$
which contradicts the assumption that $\text{supp }\chi\times V$ does not meet $WF(u)\widehat{+}_i WF(v)$.


We are reduced to study the localized product
$(u\varphi)(v\varphi)$ which is supported in a ball $B_\varepsilon$ around $x$.
 We enlarge $\Sigma(u\varphi),\Sigma(v\varphi)$ and choose
functions $\alpha_{1},\alpha_{2}$ smooth in $C^\infty(\mathbb{R}^d\setminus \{0\})$
and homogeneous of degree $0$
s.t. $((\text{supp }\alpha_{1})+(\text{supp }\alpha_{2}))\cap V=\emptyset$.

Following the method in Eskin~\cite{Eskin} (see also \cite{Dangthese}), we decompose
the convolution product in four parts:
\begin{eqnarray*}
\widehat{uv\varphi^2}|_V(\xi)&=&I_1(\xi)+I_2(\xi)+I_3(\xi)+I_4(\xi)\\
I_1(\xi)&=&\int_{\mathbb{R}^d} \alpha_{1}\widehat{u\varphi}(\xi-\eta)\alpha_{2}\widehat{v\varphi}(\eta)d\eta \\
I_2(\xi)&=&\int_{\mathbb{R}^d} (1-\alpha_{1})\widehat{u\varphi}(\xi-\eta)\alpha_{2}\widehat{v\varphi}(\eta)d\eta \\
I_3(\xi)&=&\int_{\mathbb{R}^d} (1-\alpha_{2})\widehat{v\varphi}(\xi-\eta)\alpha_{1}\widehat{u\varphi}(\eta)d\eta \\
I_4(\xi)&=&\int_{\mathbb{R}^d} (1-\alpha_{2})\widehat{v\varphi}(\xi-\eta)(1-\alpha_{1})\widehat{u\varphi}(\eta)d\eta \\
\end{eqnarray*}
Note that
$((\text{supp }\alpha_{1})+(\text{supp }\alpha_{2}))\cap V=\emptyset\implies \forall \xi\in V, I_1(\xi)=0$
hence $I_1$ vanishes and we are thus reduced
to estimate the remaining terms.
Denote by $\delta$ the distance in the unit sphere between $\left(\text{supp }\alpha_{1}\cup\text{supp }\alpha_{2}\right)\cap\mathbb{S}^{d-1}$
and $V\cap\mathbb{S}^{d-1}$. Then we have the following estimates:
\begin{eqnarray*}
\vert I_2(\xi)\vert &\leqslant & \Vert u\Vert_{2N,\text{supp }(1-\alpha_{1}),\varphi} (1+\sin\delta\vert\xi\vert)^{-N}
\int_{\mathbb{R}^d} d\eta (1+\sin\delta \vert\eta\vert)^{-N} \vert\widehat{v\varphi}(\eta)\vert\\
\vert I_3(\xi)\vert &\leqslant & \Vert v\Vert_{2N,\text{supp }(1-\alpha_{2}),\varphi} (1+\sin\delta\vert\xi\vert)^{-N}
\int_{\mathbb{R}^d} d\eta (1+\sin\delta \vert\eta\vert)^{-N} \vert\widehat{u\varphi}(\eta)\vert\\
\vert I_4(\xi)\vert &\leqslant & (1+\vert\xi\vert)^{-N}\Vert v\Vert_{2N,\text{supp }(1-\alpha_{2}),\varphi}\Vert u\Vert_{N,\text{supp }(1-\alpha_{1}),\varphi} 
\int_{\mathbb{R}^d}\frac{(1+\vert\xi\vert)^N}{(1+\vert\xi-\eta\vert)^{2N}(1+\vert\eta\vert)^{N}} d\eta 
\end{eqnarray*}
$(u\varphi,v\varphi)$ are compactly supported
distributions in $H^{s_1}(\mathbb{R}^d)\times H^{s_2}(\mathbb{R}^d)$
hence by Lemma
\ref{sobolevfouriersupnorm}
there are integers
$m_1,m_2$ and constants 
$C_1,C_2$
such that:
\begin{eqnarray*}
\Vert (1+\vert\xi\vert)^{-m_1}\widehat{u\varphi}\Vert_{L^\infty}\leqslant C_1
\Vert u\varphi\Vert_{H^{s_1}(\mathbb{R}^d)}\\
\Vert (1+\vert\xi\vert)^{-m_2}\widehat{v\varphi}\Vert_{L^\infty}\leqslant C_2\Vert v\varphi\Vert_{H^{s_2}(\mathbb{R}^d)}.
\end{eqnarray*}
Hence, we can recover our estimates
in terms of Sobolev norms:
\begin{eqnarray*}
\vert I_2(\xi)\vert &\leqslant & \Vert u\Vert_{2N,\text{supp }(1-\alpha_{1}),\varphi}\Vert (1+\vert\xi\vert)^{-m_1}\widehat{u\varphi}\Vert_{L^\infty} (1+\sin\delta\vert\xi\vert)^{-N}
\int_{\mathbb{R}^d} d\eta (1+\sin\delta \vert\eta\vert)^{-N}(1+\vert\eta\vert)^{m_1} \\
&\leqslant &  \Vert u\Vert_{2N,\text{supp }(1-\alpha_{1}),\varphi}C_1
\Vert u\varphi\Vert_{H^{s_1}(\mathbb{R}^d)} (1+\sin\delta\vert\xi\vert)^{-N}
\int_{\mathbb{R}^d} d\eta (1+\sin\delta \vert\eta\vert)^{-N}(1+\vert\eta\vert)^{m_1}\\
\vert I_3(\xi)\vert &\leqslant & \Vert v\Vert_{2N,\text{supp }(1-\alpha_{2}),\varphi}\Vert (1+\vert\xi\vert)^{-m_2}\widehat{v\varphi}\Vert_{L^\infty} (1+\sin\delta\vert\xi\vert)^{-N}
\int_{\mathbb{R}^d} d\eta (1+\sin\delta \vert\eta\vert)^{-N} (1+\vert\eta\vert)^{m_2}\\
 &\leqslant & \Vert v\Vert_{2N,\text{supp }(1-\alpha_{2}),\varphi}
 C_2\Vert v\varphi\Vert_{H^{s_2}(\mathbb{R}^d)}(1+\sin\delta\vert\xi\vert)^{-N}
\int_{\mathbb{R}^d} d\eta (1+\sin\delta \vert\eta\vert)^{-N} (1+\vert\eta\vert)^{m_2}\\
\vert I_4(\xi)\vert &\leqslant & (1+\vert\xi\vert)^{-N}\Vert v\Vert_{2N,\text{supp }(1-\alpha_{2}),\varphi}\Vert u\Vert_{N,\text{supp }(1-\alpha_{1}),\varphi} 
\int_{\mathbb{R}^d}\frac{(1+\vert\xi\vert)^N}{(1+\vert\xi-\eta\vert)^{2N}(1+\vert\eta\vert)^{N}} d\eta 
\end{eqnarray*}

Set $\Gamma_1,\Gamma_2$ to be two closed conic sets.
Hence for all $(x;\xi)\notin \Gamma_1\widehat{+}_i\Gamma_2$, for all $N>d+m_1+m_2$, 
there is a
closed cone $V\subset \mathbb{R}^{d*}$ and $\varphi\in\mathcal{D}(\mathbb{R}^d)$
such that $(x;\xi)\in\text{supp }\varphi\times V$ and $\text{supp }\varphi\times V$ 
does not meet $\Gamma_1\widehat{+}_i\Gamma_2$, and there 
are some seminorms of $\mathcal{D}^\prime_{\Gamma_1},\mathcal{D}^\prime_{\Gamma_2}$ and
some constant $C_N$ which does not depend on $u,v$
such that
\begin{eqnarray}\label{estprod2}
\Vert uv\Vert_{N,V,\varphi^2}\leqslant C_N \left(\Vert u\Vert_{2N,\text{supp }(1-\alpha_{1}),\varphi}\Vert+
\Vert u\varphi\Vert_{H^{s_1}(\mathbb{R}^d)} \right)\left(\Vert v\Vert_{2N,\text{supp }(1-\alpha_{2}),\varphi}\Vert+\Vert v\varphi\Vert_{H^{s_2}(\mathbb{R}^d)}\right)
\end{eqnarray}
\end{proof}

We define functional spaces which are Sobolev spaces
of compactly supported distributions whose wave front set
is contained in a closed cone $\Gamma\subset T^\bullet\Omega$.
\begin{defi}
Let $\Omega$ be a bounded open set in $\mathbb{R}^d$, $\Gamma$
a closed conic set in $T^\bullet\Omega$, then
a distribution $t\in\mathcal{E}^\prime(\Omega)$ belongs
to $H^s_{0,\Gamma}(\Omega)$ if
$t\in H^{s}_0(\Omega)\cap \mathcal{E}_{\Gamma}^\prime(\Omega)$.
We equip $H^s_{0,\Gamma}(\Omega)$ with the weakest topology
which makes the injections 
$H^s_{0,\Gamma}(\Omega)\hookrightarrow H^s(\mathbb{R}^d)$
and $H^s_{0,\Gamma}(\Omega)\hookrightarrow \mathcal{D}^\prime_\Gamma(\Omega)$
continuous. Equivalently, the topology
of $H^s_{0,\Gamma}(\Omega)$ is defined by the Sobolev norm
of $H^s$
and the seminorms $\Vert t\Vert_{N,V,\chi}=\sup_{\xi\in V} (1+\vert\xi\vert)^N
\vert \widehat{t\chi}(\xi) \vert $
for all $\chi\in\mathcal{D}(\Omega)$
and cone $V$ of $\mathbb{R}^d\setminus \{0\}$
s.t. $\left(\text{supp }\chi\times V\right)\cap \Gamma=\emptyset$.
\end{defi}

It follows from Proposition \ref{multsobolevcontinuous} and estimate (\ref{estprod2}) that:
\begin{thm}\label{u=0thm}
Let $\Omega$ be a bounded open set, $(s_1,s_2)$ real numbers s.t. $s_1+s_2>0$ and
$(\Gamma_1,\Gamma_2)$ two closed conic sets in $T^\bullet\Omega$.
Then
the product $$(u,v)\in H^{s_1}_{0,\Gamma_1}(\Omega) \times H^{s_2}_{0,\Gamma_2}(\Omega)
\mapsto uv\in\mathcal{E}_\Gamma^\prime(\mathbb{R}^d)$$
is continuous
where $\Gamma=\Gamma_1\widehat{+}_i\Gamma_2$.
\end{thm}

\section{The wave front set of $(f+i0)^\lambda$.}

Recall that our goal is to study from the microlocal point of view models for the
singularity of Feynman amplitudes of the form
$\prod_{i=1}^p (f_j+i0)^{\lambda_j}$. Since the proof is quite involved, we will
start smoothly by investigating the complex power $(f+i0)^\lambda$ for only one
analytic function $f$ where all the main ideas
can already be found.

Let $U$ be some open set in $\mathbb{R}^n$ and $f$ be some real valued
analytic function on $U$. 
The goal of this section is to
provide a relatively simple geometric
bound on $WF(f+i0)^\lambda$. 
Our main result in this section is related to works of Kashiwara, Kashiwara--Kawai
on the characteristic variety of the
$\mathcal{D}$-module $\mathcal{D}f^\lambda$. 
Our proof
relies on the existence of the Bernstein Sato polynomial~\cite{Grangerbernsteinsato}
and the
bounds on the wave front set
of products 
given by Theorem
\ref{u=0thm}.

We start with a useful Lemma. 
\begin{lemm}\label{BernsteinSatoWF}
Let $f$ be a real valued analytic function on an open set $U\subset \mathbb{R}^n$, 
then there is a discrete set $Z\subset \mathbb{C}$ s.t. meromorphic family $((f+i0)^\lambda)_\lambda$
satisfies the identity:
\begin{eqnarray}
\forall \lambda\in \mathbb{C}\setminus Z,\forall k\in\mathbb{N},  WF\left(f+i0\right)^{\lambda+k}=WF\left(f+i0\right)^\lambda.
\end{eqnarray}
\end{lemm}
\begin{proof}
To determine the wave front set over $U$, 
it suffices to determine it locally
in some neighborhood
of any point $x\in U$.
Following the lecture notes of Granger~\cite{Grangerbernsteinsato}, 
we must complexify
the whole situation and consider the holomorphic extension of $f$ to some complex
neighborhood $V\subset\mathbb{C}^n$ of $U$ 
and use existence of a \textbf{local} Bernstein Sato polynomial on $\mathbb{C}^n$.
 
 Let us first discuss 
some issues about complexification. Assume $f$ 
was extended by holomorphic continuation on $V\subset \mathbb{C}^n$, 
consider the open set $V=f^{-1}\left(\mathbb{C}\setminus i\mathbb{R}_{< 0} \right)$, this set contains $U$ since $f|_U$ is \textbf{real valued}, then 
we choose the branch of the $\log$ which avoids the negative
imaginary axis $i\mathbb{R}_{\leqslant 0}$ in the complex plane.
Therefore for $\varepsilon>0$, we can define the complex powers $(f+i\varepsilon)^{\lambda}=e^{\lambda\log(f+i\varepsilon)}$ for $\lambda\in\mathbb{C}$ on $V\setminus\{f=0\}$. 
When $Re(\lambda)> 0$, 
$(f+i\varepsilon)^{\lambda}$ has unique extension
as a continuous function on $V$ letting $\varepsilon$ goes to zero.
Indeed $(f+i0)^\lambda=0$ on
$\{f=0\}$ and $(f+i0)^\lambda$ equals $f^\lambda$ on $V\setminus \{f=0\}$ and 
$(f+i0)^\lambda$ is thus
\textbf{holomorphic} on $V\setminus \{f=0\}$.
In the sequel, we denote by $x=(x_1,\dots,x_n)$ the coordinates in the real open set
$U$ and by $z=(z_1,\dots,z_n)$ the complex coordinates in $V$.

 Assuming that $(U,V)$ are chosen small
enough, 
by the local existence of
the Bernstein Sato polynomial~\cite[Theorem 5.4 p.~257]{Grangerbernsteinsato}, there exists
a holomorphic differential operator
$P(z,\partial_z)$ with holomorphic coefficients 
and a polynomial $b(\lambda)$ 
s.t.
\begin{equation*}
P(z,\partial_z)f^{\lambda+1}=b(\lambda)f^\lambda.
\end{equation*}
This relation is valid on $V\setminus \{f=0\}$. 

Going back to the real case, we have an equation 
\begin{equation*}
P(x,\partial_x)f^{\lambda+1}=b(\lambda)f^\lambda.
\end{equation*}
on $U\setminus\{f=0\}$ where  
the real analytic set
$\{f=0\}$ has null measure in $U$ by Lemma \ref{thinset}, when $Re(\lambda)$
is strictly larger than the order
of the differential operator $P$, both
$P(x,\partial_x)f^{\lambda+1}$ and $f^{\lambda}$ have unique continuation as functions of regularity $C^0$ and $C^k$
on $W$ respectively
and
the above identity holds true in the sense of distributions.

Since $(f+i0)^\lambda$
extends meromorphically in $\lambda$ with value distribution by Theorem \ref{Atiyahmeromextension}, 
the following equation holds true
at the distributional level:
\begin{equation}
P(x,\partial_x)\left(f+i0\right)^{\lambda+1}=b(\lambda)\left(f+i0\right)^\lambda
\end{equation}
for all $\lambda$ avoiding 
the poles of $f^{\lambda+1},f^\lambda$ and the zeros of $b$.
Therefore for such $\lambda$, one has
\begin{eqnarray*}
b^{-1}(\lambda)P(x,\partial_x)\left(f+i0\right)^{\lambda+1}&=&\left(f+i0\right)^\lambda\\
\implies WF\left(f+i0\right)^\lambda&=& WF\left(b^{-1}(\lambda)P(x,\partial_x)\left(f+i0\right)^{\lambda+1}\right)\\
 \implies WF\left(f+i0\right)^\lambda &\subset &
WF\left(f+i0\right)^{\lambda+1}.
\end{eqnarray*}
We used the classical bound on the wave front set $WF(Pu)\subset WF(u)$ where $u\in\mathcal{D}^\prime$
and $P$ is a differential operator. 
On the other hand $\left(f+i0\right)^{\lambda+1}=f\left(f+i0\right)^\lambda$
which implies that $WF\left(f+i0\right)^{\lambda+1}\subset WF\left(f+i0\right)^\lambda$,
finally set $Z$ to be equal to $\left(
\{\text{poles of }((f+i0)^\lambda)_\lambda\}\cup \{\text{ zeros of }b\}\right)-\mathbb{N}$, this yields
\begin{eqnarray}
\forall \lambda\notin Z,  WF\left(f+i0\right)^{\lambda+1}=WF\left(f+i0\right)^\lambda.
\end{eqnarray}
\end{proof}
\textbf{Morality: it suffices to bound $WF(f+i0)^\lambda$ for $Re(\lambda)$ then we would bound $WF(f+i0)^\lambda$ for all $\lambda \notin Z$.}
Now we state and prove the main Theorem 
of this section. The proof relies on the $u=0$ Theorem.

\begin{thm}\label{WFf+i0}
Let $f$ be a real valued analytic function s.t. $\{df=0\}\subset\{f=0\}$, assume $f$ is proper then
for all $\lambda\notin Z$,
\begin{equation}
WF((f+i0)^\lambda)\subset \{ (x;\xi) \text{ s.t. }\exists \{(x_k,a_k)_k\}\in \left(\mathbb{R}^n\times \mathbb{R}_{>0}\right)^{\mathbb{N}} ,
x_k\rightarrow x, f(x_k)\rightarrow 0, a_kdf(x_k)\rightarrow\xi \}.
\end{equation}
\end{thm}
\begin{proof}
We use the very simple idea
to convert $(f+i0)^\lambda$ into a slightly more complicated integral
which is easier to control:
\begin{eqnarray}
(f+i0)^\lambda= \int_\mathbb{R} dt (t+i0)^\lambda\delta_{t-f}.
\end{eqnarray}
Let $\pi$ be the 
projection $\pi:(t,x)\in\mathbb{R}\times \mathbb{R}^n\mapsto x\in\mathbb{R}^n$.
The above integral formula for $(f+i0)^\lambda$ 
can also be conveniently reformulated as a
pushforward $\pi_*\left((t+i0)^\lambda\delta_{t-f}\right)$. 

Step 1 First, let us show that for $Re(\lambda)$ large enough
the product $(t+i0)^\lambda\delta_{t-f}$ makes sense in $\mathcal{D}^\prime$.
Let $(U_i)_i$ be an open cover of $\mathbb{R}\times U$ by
bounded open sets and
$(\varphi_i)_i$ a subordinated partition of unity $\sum \varphi_i^2=1$.
Then it is enough to consider 
$$\sum_i ((t+i0)^\lambda\varphi_i)(\delta_{t-f}\varphi_i) .$$
The delta function $\delta_{t-f}$ is supported by 
the hypersurface 
$\{t-f=0\}$, by the usual Sobolev trace Theorem~\cite[Theorem 13.6]{Eskin}, any function
in $H^{s}(\mathbb{R}^{n+1})$ for $s>\frac{1}{2}$ can be restricted
on $\{t-f=0\}$ therefore
by duality of
Sobolev space~\cite[Proposition 13.7]{Eskin}, $\delta_{t-f}$ belongs to $H^s(\mathbb{R}^{n+1})$ 
for all $s<-\frac{1}{2}$.
If $Re(\lambda)\geqslant m\in\mathbb{N}$
then $((t+i0)^\lambda\varphi_i)$ is a
compactly supported
function with regularity $C^{m}$, 
hence it belongs to the Sobolev space
$H^s(\mathbb{R}^{n+1})$ for $m>s+\frac{n+1}{2}$ by the continuous injection (\ref{ckembedding}) of Proposition \ref{embeddings}. 
It follows
that for every $s$, $((t+i0)^\lambda\varphi_i)
\in H^s(\mathbb{R}^{n+1})$ for $Re(\lambda)$ large enough.

Step 2 To study $WF\left((t+i0)^\lambda\delta_{t-f}\right)$,
we will use the $u=0$ Theorems to give bounds
on the wave front set
of the product of $(t+i0)^\lambda$ with $\delta_{t-f}$.
Since $\delta_{t-f}\in H^{-\frac{1}{2}-\varepsilon}, \forall\varepsilon>0$,
the product
$\left((t+i0)^\lambda\delta_{t-f}\right)$ makes sense
for all $\lambda$ s.t. $Re(\lambda)>\frac{n}{2}+1$, and by the $u=0$ Theorem \ref{u=0thm}, 
\begin{eqnarray}
WF\left((t+i0)^\lambda\delta_{t-f}\right)\subset WF(t+i0)^\lambda\widehat{+}_iWF\left(\delta_{t-f}\right).
\end{eqnarray}

We start from the elementary wave front sets:
\begin{eqnarray*}
WF\left(\delta_{t-f}\right)=\{(t,x;\tau,\xi) \text{ s.t. } f(x)=t, \xi=-\tau df ,\tau\neq 0 \}\\
WF(t+i0)^\lambda=\{(0,x;\tau,0)\text{ s.t. } \tau> 0 \},
\end{eqnarray*}
and by definition of the $\widehat{+}_i$ operation
of Iagolnitzer,
it is obvious that
outside $t=0$, 
\begin{eqnarray*}
WF(t+i0)^\lambda\widehat{+}_iWF\left(\delta_{t-f}\right)|_{\{t\neq 0\}}=WF\left(\delta_{t-f}\right)|_{\{t\neq 0\}}.
\end{eqnarray*}
At $t=0$,
set 
\begin{equation}
\Gamma_f=\{ (0,x;\tau,\xi) \text{ s.t. }\exists (x_n,\tau_n,\tau_n^\prime)_{n\in\mathbb{N}},
x_n\rightarrow x, f(x)=0, \xi_n=-\tau_ndf(x_n)\rightarrow\xi, \tau_n+\tau^\prime_n\rightarrow \tau,\tau^\prime_n>0 \}.
\end{equation}
Then we find that
\begin{eqnarray*}
&&WF(t+i0)^\lambda\widehat{+}_iWF\left(\delta_{t-f}\right)|_{t=0}\\
&=& \{ (0,x;\tau,\xi) \text{ s.t. }
\exists (x_n,\tau_n,\tau_n^\prime)_{n\in\mathbb{N}},
x_n\rightarrow x, f(x)=0, -\tau_ndf(x_n)\rightarrow\xi, \tau_n+\tau^\prime_n\rightarrow \tau,\tau^\prime_n>0 \}\\
&=&\Gamma_f.
\end{eqnarray*}
The above yields
$\Gamma_f=WF((t+i0)^\lambda)\widehat{+}_iWF\left(\delta_{t-f}\right)|_{\{t=0\}}$.

Step 3, we evaluate the wave front set of $(f+i0)^\lambda$ viewed as the push--forward
$\pi_*\left((t+i0)^\lambda \delta_{t-f}\right)$. Outside 
$\{t=0\}$, $WF\left((t+i0)^\lambda\delta_{t-f}\right)\cap T^\bullet((\mathbb{R}\setminus \{0\})\times U) =WF\left(\delta_{t-f}\right)\cap T^\bullet((\mathbb{R}\setminus \{0\})\times U)$
and by the behaviour of the wave front set
under push--forward~\cite[Proposition]{Viet-wf2}, $\pi_*WF\left(\delta_{t-f}\right)=\emptyset$.
Hence, only the elements of $\Gamma_f$ 
of the form $(0,x;\tau=0,\xi)\in T^\bullet(\mathbb{R}\times U)$
contribute to the wave front set
of $\pi_*(\Gamma_f)$
and are calculated as follows:
\begin{eqnarray*}
\Gamma_f\cap\{(0,x;\tau=0,\xi)\}&=&\{ (0,x;0,\xi) \text{ s.t. }
x_n\rightarrow x, f(x)=0, \xi_n= -\tau_ndf(x_n)\rightarrow\xi, \tau_n+\tau^\prime_n\rightarrow 0,\tau^\prime_n>0 \}\\
&=& \{ (0,x;0,\xi) \text{ s.t. }
x_n\rightarrow x, f(x)=0, \xi_n=\tau_ndf(x_n)\rightarrow\xi, \tau_n>0 \}.
\end{eqnarray*}
Define 
\begin{equation}
\Lambda_f=\{ (x;\xi) \text{ s.t. }\exists \{(x_k,a_k)_k\}\in \left(\mathbb{R}^n\times \mathbb{R}_{>0}\right)^{\mathbb{N}} ,
x_k\rightarrow x, f(x_k)\rightarrow 0, a_kdf(x_k)\rightarrow\xi \}.
\end{equation}
by definition of $\pi_*$ it is immediate that $\Lambda_f=\pi_*\left(\Gamma_f\right)$. 
It follows that:
\begin{eqnarray*}
WF(\pi_*\left((t+i0)^\lambda\delta_{t-f}\right)
&\subset &\pi_*\left(WF(t+i0)^\lambda\widehat{+}_iWF\left(\delta_{t-f}\right) \right)\\
&=&\pi_*\left(\Gamma_f\right)=\Lambda_f.
\end{eqnarray*}
\end{proof}

\section{Functional calculus with value $\mathcal{D}^\prime_\Gamma$.}

In the sequel, for any manifold $M$, we will denote by $T^\bullet M$
the cotangent space $T^*M$ minus its zero section.
In QFT on curved analytic spacetimes, 
we will show that the meromorphically regularized
Feynman amplitudes
in position space are distributions
depending meromorphically on the regularization parameter. However in order to renormalize, 
we need to control the WF of the regularized amplitudes therefore
we are let to develop a functional calculus
for distributions with value in the space $\mathcal{D}^\prime_\Gamma$ 
of distributions whose wave front set is contained
in some closed conic set $\Gamma$ of the cotangent cone $T^\bullet\mathbb{R}^d$.

\subsubsection{The space $\mathcal{D}^\prime_\Gamma$ characterized by duality.}
We work with the space $\mathcal{D}^\prime_\Gamma$
of distributions whose wave front set is contained
in some closed conic set $\Gamma$ of the cotangent space $T^\bullet\mathbb{R}^d$ 
endowed with the normal topology constructed by Brouder Dabrowski~\cite{Viet-wf2}.
For any closed conic set
$\Gamma \subset T^*\mathbb{R}^d$, we denote by
$-\Gamma=\{(x;-\xi)\text{ s.t. }(x;\xi)\in\Gamma \}$ the antipode of $\Gamma$
and by $\Gamma^c$ the complement of $\Gamma$ in
$T^\bullet\mathbb{R}^d$. The space of compactly supported
distribution whose wave front set
is contained in some conic set $\Lambda$
will be denoted by $\mathcal{E}^\prime_\Lambda$.
The most important property for us is the following characterization
of $\mathcal{D}^\prime_\Gamma$ by duality.
\begin{prop}\label{boundeddprimegammaduality} 
A set $B$ of distributions in $\mathcal{D}^\prime_\Gamma$ 
is bounded if and only if, for
every $v\in\mathcal{E}^\prime_\Lambda$, $\Lambda=(-\Gamma)^c$, there is a constant $C>0$ such that 
$\vert \langle u,v \rangle \vert\leqslant C$ for all $u\in B$. 
\end{prop}
Such
a weakly bounded set is also strongly bounded and equicontinuous. Moreover, the
closed bounded sets of $\mathcal{D}^\prime_\Gamma$ 
are compact, complete and metrizable.
The second important property is the following sufficient condition
to describe sequential convergence in $\mathcal{D}^\prime_\Gamma$:
\begin{prop}\label{seqdualitydprimegamma}
If $u_i$ is a sequence of elements of $\mathcal{D}^\prime_\Gamma$ 
such that, for
every $v\in\mathcal{E}^\prime_\Lambda$, the sequence $\langle u_i,v \rangle$
converges to $\lambda(v)$, then $u_i$ converges to a
distribution $u$ in $\mathcal{D}^\prime_\Gamma$ and 
$\langle u,v \rangle=\lambda(v)$ for all  $v\in \mathcal{E}^\prime_\Lambda$.
\end{prop}

The above plays the same role as the characterization of sequential
convergence in $\mathcal{D}^\prime(\Omega)$ by duality, it suffices to
verify for all test function $\varphi$, $t_n(\varphi)$ converges
as a sequence of real numbers.

\subsubsection{Continuous, holomorphic functions with value in
$\mathcal{D}^\prime_\Gamma$.}
Motivated by the above characterizations
of $\mathcal{D}^\prime_\Gamma$
by duality, we can
give definitions
of being continuous or holomorphic
with value $\mathcal{D}^\prime_\Gamma$.
For applications to QFT
we need to consider
holomorphic (resp meromorphic) functions
depending on several complex variables.
\begin{defi}
A family of distributions $(t_\lambda)_\lambda$ depends continuously (resp holomorphically) 
on a complex parameter $\lambda\in\mathbb{C}^p$
with value $\mathcal{D}^\prime_\Gamma$
if for every test distribution $v\in \mathcal{E}^\prime_\Lambda,\Lambda=-\Gamma^{c}$, $t_\lambda(v)
$ is a continuous (resp holomorphic) function of $\lambda$.
We will also call such family \emph{continuous (resp holomorphic) 
with value} $\mathcal{D}^\prime_\Gamma$. 
\end{defi}

It follows from Proposition
\ref{seqdualitydprimegamma} that
\begin{prop}\label{weakintegral}
Let $\Omega$ be an open set in $\mathbb{C}^p$ and
$(t_\lambda)_{\lambda\in\Omega}$ a family of distributions in $\mathcal{D}^\prime_\Gamma$.
If $(t_\lambda)_\lambda$ depends \textbf{continuously}
on $\lambda\in\Omega\subset\mathbb{C}^p$ with value $\mathcal{D}^\prime_\Gamma$,
let $\gamma=\gamma_1\times\dots\times \gamma_p\subset\mathbb{C}^p$ be a cartesian product
where each $\gamma_i$ is a \emph{continuous curve} in $\mathbb{C}$, then
the weak integrals $\int_{\gamma\subset \mathbb{C}^p} d\lambda t_\lambda$
exists in $\mathcal{D}^\prime_\Gamma$.
\end{prop}
\begin{proof} 
For every test distribution $v\in \mathcal{E}^\prime_\Lambda,\Lambda=-\Gamma^{c}$,
the function $\lambda\in\gamma \mapsto \langle t_\lambda,v \rangle$
is continuous hence Riemann integrable.
Therefore $\int_{\gamma} d\lambda \langle t_\lambda,v \rangle$
exists as a limit of Riemann sums
and the integral $\int_{\gamma} d\lambda t_\lambda$
is well defined by the sequential
characterization of convergence in 
$\mathcal{D}^\prime_\Gamma$.
\end{proof}
In that case, we will also say that
$(t_\lambda)_\lambda$
is meromorphic with linear poles in
$\lambda$ with value $\mathcal{D}^\prime_\Gamma$.
\subsubsection{Meromorphic functions with linear poles with value $\mathcal{D}^\prime_\Gamma$.}

In the present work, we deal with families of distributions $(t_\lambda)_{\lambda\in\mathbb{C}^p}$
in $\mathcal{D}_\Gamma^\prime(U)$ depending meromorphically on
$\lambda\in\mathbb{C}^p$ with linear poles.
\begin{defi}
A family of distributions $(t_\lambda)_{\lambda\in\mathbb{C}^p}$
in $\mathcal{D}_\Gamma^\prime(U)$ depends meromorphically on
$\lambda\in\mathbb{C}^p$ with linear poles
if for every $x\in U$, there is a neighborhood $U_x$ of $x$, a collection $(L_i)_{1\leqslant i\leqslant m}\in(\mathbb{N}^p)^m\subset (\mathbb{C}^{p*})^m$ of linear functions with integer coefficients
on $\mathbb{C}^p$ such that for any element
$z=(z_1,\dots,z_p)\in \mathbb{Z}^p$,
there is a neighborhood $\Omega\subset \mathbb{C}^p$ of $z$, 
such that
\begin{equation}
\lambda\in\Omega \mapsto\prod_{i=1}^m(L_i(\lambda+z))t_\lambda|_{U_x}
\end{equation}
is holomorphic with value $\mathcal{D}^\prime_\Gamma(U_x)$.
\end{defi}

%

\subsubsection{A gain of regularity: when continuity becomes holomorphicity.}

Now we give an easy 
\begin{prop}\label{boundedbecomesholo}
Let $U$ be an open subset of $\mathbb{R}^n$, an open set $\Omega\subset \mathbb{C}^p$,
a family $(t_\lambda)_{\lambda\in\Omega}$ holomorphic in $\lambda$ with value
$\mathcal{D}^\prime(U)$.
If $(t_\lambda)_\lambda$ is  \textbf{continuous} with value $\mathcal{D}^\prime_\Gamma(U)$ 
then $(t_\lambda)_\lambda$ is \emph{holomorphic with value} $\mathcal{D}^\prime_\Gamma(U)$. 
\end{prop}
\begin{proof}
It suffices to observe that
by holomorphicity of $t$ and the multidimensional 
Cauchy's formula~\cite[p.~3]{GunningRossi} for any polydisk
$D_1\times \dots\times D_p$ such that 
$\partial D_i$ is a circle surrounding
$z_i$:
\begin{eqnarray}
t_\lambda=\frac{1}{(2i\pi)^p}\int_{\partial D_1 \times\dots\times \partial D_p} \frac{t_z dz_1\wedge \dots\wedge dz_p}{(z_1-\lambda_1)\dots (z_p-\lambda_p)}. 
\end{eqnarray}

Since $t_z$ is continuous along $\partial D_1 \times\dots\times \partial D_p$, then for any $v\in \mathcal{E}^\prime_\Lambda,\Lambda=-\Gamma^c$, 
the quantity 
\begin{eqnarray}
t_\lambda(v)=\frac{1}{(2i\pi)^p}\int_{\partial D_1 \times\dots\times \partial D_p} \frac{t_z(v) dz_1\wedge \dots\wedge dz_p}{(z_1-\lambda_1)\dots (z_p-\lambda_p)}  
\end{eqnarray}
is well defined by Proposition \ref{weakintegral}
and holomorphic in $\lambda$ by the integral representation which proves
the holomorphicity of $(t_\lambda)_\lambda$ with value $\mathcal{D}^\prime_\Gamma$.
\end{proof}
By definition of functions meromorphic  with linear poles
with value $\mathcal{D}^\prime_\Gamma$, we obtain:
\begin{coro}\label{merobecomesmero}
Let $\Omega\subset \mathbb{C}^p$, $z_1\in\mathbb{C}$,
$(t_\lambda)_{\lambda\in\Omega}$ a meromorphic family of distributions
with linear poles. Denote by $Z$ the polar set of $t$.
If $(t_\lambda)_\lambda$ is \textbf{continuous} on $\Omega\setminus Z$ with value $\mathcal{D}^\prime_\Gamma$ 
then $(t_\lambda)_\lambda$ is \emph{meromorphic with linear poles} with value $\mathcal{D}^\prime_\Gamma$. 
\end{coro}

\subsubsection{Consequences of Riemann's removable singularity Theorem.}

\begin{lemm}\label{Riemannremov}
Let $\Omega\subset \mathbb{C}^p$, $z_1\in\mathbb{C}$,
$(t_\lambda)_{\lambda\in\Omega}$ a meromorphic family of distributions
with linear poles in $Z\subset \Omega$.
If $\lambda\in\mathbb{C}^p\mapsto t_\lambda\in\mathcal{D}^\prime$ is
\textbf{locally bounded} 
then $(t_\lambda)_\lambda$ is a holomorphic family of distributions.
\end{lemm}
\begin{proof}
For every test function
$\varphi$, $\lambda\in\mathbb{C}^p \mapsto t_\lambda(\varphi)$
is meromorphic i.e. holomorphic on $\mathbb{C}^p\setminus Z$
where $Z$ is a \emph{thin set} 
and locally bounded hence
by
Riemann's removable singularity Theorem~\cite{GunningRossi}
$\lambda\in\mathbb{C}^p \mapsto t_\lambda(\varphi)$ is holomorphic.
We conclude by showing
it is a distribution
at the points in $Z$ where singularities were removed.
Let $\lambda$ be such a point, then the representation
of $t_\lambda$ by Cauchy's formula
$t_\lambda=\frac{1}{(2i\pi)^p}\int_{\partial D_1 \times\dots\times \partial D_p} \frac{t_z dz_1\wedge \dots\wedge dz_p}{(z_1-\lambda_1)\dots (z_p-\lambda_p)}$
along some contour $\partial D_1 \times\dots\times \partial D_p$  which does not intersect some neighborhood
of $\lambda$ shows that
$t_\lambda$ is a weak integral with value distribution
hence it is a distribution by Proposition \ref{weakintegral} applied
to the conic set $\Gamma=T^\bullet\mathbb{R}^n$.
\end{proof}

\subsubsection{Laurent series expansions of meromorphic distributions with 
linear poles.}
We start by examining Laurent series expansions of 
families $(t_\lambda)_\lambda$ of distributions with value
$\mathcal{D}^\prime_\Gamma$ where $\lambda$ is only one complex variable. 
We show that the coefficients
of the Laurent series expansion
of $t$ are also distributions in $\mathcal{D}^\prime_\Gamma$.
\begin{prop}\label{wavefrontlaurent}
Let $(t_\lambda)_{\lambda\in\mathbb{C}}$ be a meromorphic
family of distributions with value 
$\mathcal{D}^\prime_\Gamma$.
Then for all $z_0\in\mathbb{C}$,
there exists $\varepsilon>0$ and a bounded
set $B$ in $\mathcal{D}^\prime_\Gamma$ s.t. 
the
Laurent series expansion
of $t_\lambda$ around $z_0$
reads
\begin{equation}
t_\lambda=\sum_{k} a_k(\lambda-z)^k
\end{equation}
where for all $k$,
\begin{enumerate}
\item $a_k\in\mathcal{D}^\prime_\Gamma$
\item moreover $\frac{\varepsilon^{k}}{k!}a_k\in B$ if $z_0$ is a regular value
of $\lambda\mapsto t_\lambda$ and $\frac{\varepsilon^{k}(1+2^k)}{k!}a_k\in B$ if $z_0$ is a pole of
$\lambda\mapsto t_\lambda$.
\end{enumerate}
\end{prop}
It follows that the wave front of $a_k$
is contained in $\Gamma$.
We call such series expansion
\textbf{absolutely convergent}
with value in
$\mathcal{D}^\prime_\Gamma$.
\begin{proof}
Without loss of generality, we can
assume that 
$z_0=0$. 
First case, $0$ is not a pole of $t$.
Choose $\varepsilon>0$ such that
the disc of radius $\varepsilon$
contains only $0$ as pole and denote by
$\gamma$ the circle $\{\vert z\vert=\varepsilon\} \subset \mathbb{C}$.
Let $t_\gamma=\{t_\lambda\text{ s.t. }\lambda\in\gamma\}\subset \mathcal{D}^\prime_\Gamma$
be the curve described by $t$ in $\mathcal{D}^\prime_\Gamma$ when
$\lambda$ runs in $\gamma$, this curve is obviously a bounded subset
of $\mathcal{D}^\prime_\Gamma$ by the continuity of
$\lambda\in\gamma\mapsto t_\lambda\in\mathcal{D}^\prime_\Gamma$ and 
Proposition \ref{boundeddprimegammaduality} characterizing bounded sets by duality.
We want to consider the set $B$ defined as the closure 
of the disked hull of the curve $t_\gamma$:
\begin{equation}
B=\overline{\{\alpha t_{\lambda_1}+\beta t_{\lambda_2} \text{ s.t. } \vert\alpha\vert+\vert\beta\vert\leqslant 1,
(\alpha,\beta)\in \mathbb{C}^2, \vert\lambda_1\vert=\vert\lambda_2\vert=\varepsilon \}}.
\end{equation}
It is immediate that the disked hull is still bounded in $\mathcal{D}^\prime_\Gamma$
by the characterization of bounded sets by duality 
hence its closure $B$
is bounded in $\mathcal{D}^\prime_\Gamma$.
To summarize $B$ is a closed, bounded disk in
$\mathcal{D}^\prime_\Gamma$.
Recall $\gamma$ is the circle
$\{\vert z\vert=\varepsilon\}\subset \mathbb{C}$,
then by Cauchy's formula, we have
$$
\forall k ,a_k=\frac{k!}{2i\pi}\int_{\gamma} \frac{t_\lambda d\lambda}{\lambda^{k+1}}.$$
By the definition of the weak integral
as limit of Riemann sums, we find that:
\begin{eqnarray*} 
&&\frac{k!}{2i\pi}\int_{\gamma} \frac{t_\lambda d\lambda}{\lambda^{k+1}}\\
&=&\lim_{n}\frac{k!}{2i\pi}\frac{\sum_{j=1}^n\varepsilon^{-(k+1)} \exp(-i2\pi \frac{j(k+1)}{n})t_{\varepsilon \exp(i2\pi \frac{j}{n}) }}{\frac{n}{2\pi\varepsilon}}\\
&=&  k!\varepsilon^{-k}\lim_{n}\sum_{j=1}^n\underset{\in B}{\underbrace{\frac{  \exp(-i2\pi \frac{j(k+1)}{n})t_{\varepsilon \exp(i2\pi \frac{j}{n}) }}{in}}} 
\end{eqnarray*}
hence $\lim_{n}\frac{\sum_{j=1}^n  \exp(-i2\pi \frac{j(k+1)}{n})t_{\varepsilon \exp(i2\pi \frac{j}{n}) }}{in} $ belongs to $B$ by construction of the closed disk $B$
and
it follows that
$a_k\in \frac{k!B}{\varepsilon^k}$. 

In case $0$ is a pole, we must repeat the above proof for
a corona of the form
$\{\frac{\varepsilon}{2} \leqslant\vert z\vert\leqslant \varepsilon \}$.
So the Cauchy
formula gives an integral over two circles
of radius $\frac{\varepsilon}{2}$ and $\varepsilon$ respectively.
And the same argument as above gives that
$\frac{k!}{2i\pi}\int_{\gamma} \frac{t_\lambda d\lambda}{\lambda^{k+1}}$
belongs to $\frac{k!}{\varepsilon^k}(B+2^kB)\subset \frac{k!(1+2^k)}{\varepsilon^k}B$
since $B$ is a \textbf{disk}.
\end{proof}

The same result holds true
for holomorphic distributions
depending on several
complex variables by the same type of argument
transposed to the multivariable complex case.
\begin{prop}\label{wavefrontlaurent2}
Let $(t_\lambda)_{\lambda\in\mathbb{C}^p}$ be a holomorphic
family of distributions with value 
$\mathcal{D}^\prime_\Gamma$.
Then for all $z_0\in\mathbb{C}^p$,
there exists $\varepsilon>0$ and a bounded
set $B$ in $\mathcal{D}^\prime_\Gamma$ s.t. 
the
power series expansion
of $t_\lambda$ around $z_0$
reads
\begin{equation}
t_\lambda=\sum_{k\in\mathbb{N}^p} a_k(\lambda-z)^k
\end{equation}
where for all multi--index $k$,
\begin{enumerate}
\item $a_k\in\mathcal{D}^\prime_\Gamma$
\item and $\frac{\varepsilon^{\vert k\vert}}{k!}a_k\in B$.
\end{enumerate}
\end{prop}
The bound $\frac{\varepsilon^{\vert k\vert}}{k!}a_k\in B$
is a functional version of Cauchy's bound in our functional context.

In the meromorphic case with linear poles, we must use 
the analogue of Laurent series decomposition for meromorphic functions with linear poles
given by Theorem \ref{PaychaZhang} 
and we obtain:
\begin{thm}
Let $U$ be an open set in $\mathbb{R}^n$, $\Omega\subset \mathbb{C}^p$ open and 
$(t_\lambda)_{\lambda\in\Omega}$ a meromorphic family with linear poles
at $k\in\Omega$ with value $\mathcal{D}_\Gamma^\prime(U)$.
Then the element $t_\lambda=\frac{h}{L_1\dots L_n}$
can be written as a sum 
\begin{equation}\label{PaychaLaurentdecomp}
t=\sum_i\frac{h_i(\ell_{i(n_i+1)},\dots,\ell_{ip})}{L_{i1}^{s_{i1}}\dots L_{in_i}^{s_{in_i}}}+\phi_i(L_{i1},\dots,L_{in_i},\ell_{i(n_i+1)},\dots,\ell_{ip})
\end{equation}
where for each $i$, $(s_{i1},\dots,s_{in_i})\in\mathbb{N}^{n_i}$, 
the collection of linear forms $(L_{i1},\dots,L_{in_i})$ is a linearly independent
subset of $(L_1,\dots,L_n)$, the collection of linear forms\\
$(\ell_{i(n_i+1)},\dots,\ell_{ip})$ is a basis of the orthogonal complement 
of the subspace spanned by the $(L_{i1},\dots,L_{in_i})$ and $h_i,\phi_i$ are
\textbf{holomorphic distributions} with value $\mathcal{D}^\prime_\Gamma(U)$.
\end{thm}
\begin{proof}
This is an immediate consequence of the fact that
$h$ is holomorphic in $\lambda$ with value 
$\mathcal{D}^\prime_\Gamma$ and
that $h_i,\phi_i$ are
linear combinations in finite partial derivatives
of $h$ in $\lambda$ and are therefore
\textbf{holomorphic distributions} with value $\mathcal{D}^\prime_\Gamma(U)$.
\end{proof}


\subsubsection{When H\"ormander products of holomorphic distributions are holomorphic.}
The next proposition aims to prove the holomorphicity
of a product 
of holomorphic 
distributions
with specific conditions
on their wave front set.
\begin{prop}\label{holomorphicproductprop}
Let $(\Omega_1,\Omega_2)$
be two subsets of $(\mathbb{C}^{p_1},\mathbb{C}^{p_2})$ respectively,\\
$a(\lambda_1),b(\lambda_2)_{\lambda_1\in\Omega_1,\lambda_2\in\Omega_2}$
be two families of distributions
which are holomorphic with value 
$(\mathcal{D}^\prime_{\Gamma_1},\mathcal{D}^\prime_{\Gamma_2})$. 
If $\Gamma_1\cap-\Gamma_2=\emptyset$, set 
$\Gamma=\left(\Gamma_1+\Gamma_2\right)\cup \Gamma_1\cup\Gamma_2$
then the product $a(\lambda_1)b(\lambda_2)$
is holomorphic on $\Omega_1\times\Omega_2$
with value $\mathcal{D}^\prime_\Gamma$.
\end{prop}
\begin{proof}
First by transversality of wave front sets,
the product $a(\lambda_1)b(\lambda_2)$ is well defined
pointwise
for every $(\lambda_1,\lambda_2)\in \Omega_1\times\Omega_2$. Moreover,
by Cauchy's formula and hypocontinuity of the
product~\cite{Viet-wf2} the integral representation
$$a(\lambda_1)b(\lambda_2)=
\int_{\gamma_1} \frac{dz_1}{z_1-\lambda_1}a(z_1)\int_{\gamma_2} \frac{dz_2}{z_2-\lambda_2}  b(z_2) $$
is well defined: use Riemann sum's argument to express
the two integrals\\ 
$(\int_{\gamma_1} \frac{a(\lambda_1)d\lambda_1}{\lambda_1-z_1},\int_{\gamma_2} \frac{b(\lambda_2)d\lambda_2}{\lambda_2-z_2}  )$ as convergent sequences in
$\mathcal{D}^\prime_{\Gamma_1},\mathcal{D}^\prime_{\Gamma_2}$ respectively then the sequential continuity of the product ensures the convergence of the multiplication
$$(a,b)\in \mathcal{D}^\prime_{\Gamma_1}\times\mathcal{D}^\prime_{\Gamma_2}\mapsto (ab)\in \mathcal{D}^\prime_\Gamma  $$
hence the product $\int_{\gamma_1} \frac{dz_1}{z_1-\lambda_1}a(z_1)\int_{\gamma_2} \frac{dz_2}{z_2-\lambda_2}  b(z_2) $
is holomorphic in $(\lambda_1,\lambda_2)\in (\Omega_1\times\Omega_2)\subset\mathbb{C}^{p_1+p_2}$.
\end{proof}

\begin{prop}\label{meromproductprod}
Let $(u_{\lambda_1})_{\lambda_1},(v_{\lambda_2})_{\lambda_2}$ be two 
families of distributions
with value in
$(\mathcal{D}^\prime_{\Gamma_1},\mathcal{D}^\prime_{\Gamma_2})$
respectively
with $\Gamma_1\cap -\Gamma_2=\emptyset$ depending
meromorphically on $\lambda_1\in\mathbb{C}^{p_1},\lambda_2\in\mathbb{C}^{p_2}$ with linear poles.
Set
$\Gamma=\Gamma_1+\Gamma_2\cup \Gamma_1\cup\Gamma_2$
then the
product
$u_{\lambda_1}v_{\lambda_2}$
is meromorphic in $(\lambda_1,\lambda_2)\in\mathbb{C}^{p_1+p_2}$ 
with linear poles with value $\mathcal{D}^\prime_\Gamma$.
\end{prop}
\begin{proof}
The proof follows immediately from the 
decomposition \ref{PaychaLaurentdecomp} applied to 
both $u$ and $v$ separately
and application of Proposition \ref{holomorphicproductprop}.
\end{proof}

\subsection{Functional properties of
$((f+i0)^\lambda)_{\lambda\in\mathbb{C}}$.}

In this subsection, we use the newly defined functional
calculus to investigate 
the functional properties of the family
$((f+i0)^\lambda)_{\lambda\in\mathbb{C}}$.

\begin{prop}\label{powersfmerom}
Let $f\neq 0$ be a real valued analytic function s.t. $\{df=0\}\subset\{f=0\}$, $Z$
some discrete set which contains 
the poles of the meromorphic family
$((f+i0)^\lambda)_\lambda$
and 
\begin{equation}
\Lambda_f=\{ (x;\xi) \text{ s.t. }\exists \{(x_k,a_k)_k\}\in \left(\mathbb{R}^n\times \mathbb{R}_{>0}\right)^{\mathbb{N}} ,
x_k\rightarrow x, f(x_k)\rightarrow 0, a_kdf(x_k)\rightarrow\xi \}.
\end{equation}
Then
$((f+i0)^\lambda)_{\lambda\in \mathbb{C}\setminus Z}$
is meromorphic with value $\mathcal{D}^\prime_{\Lambda_f}$.
\end{prop}
\begin{proof}
By Theorem \ref{Atiyahmeromextension}, we already know that the family $((f+i0)^\lambda)_{\lambda\in \mathbb{C}}$
is meromorphic with value $\mathcal{D}^\prime$,
then it suffices
to show that $((f+i0)^\lambda)_{\lambda\in \mathbb{C}\setminus Z}$
is continuous in $\mathcal{D}^\prime_{\Lambda_f}$
and by Proposition \ref{merobecomesmero}, we deduce
that $((f+i0)^\lambda)_{\lambda\in \mathbb{C}} $ is meromorphic with value
$\mathcal{D}^\prime_{\Lambda_f}$.

We want to show that it is sufficient to prove that
$\lambda\in K\mapsto (f+i0)^\lambda\in 
\mathcal{D}^\prime_{\Lambda_f} $ is continuous
for arbitrary compact subsets $K\subset \mathbb{C}$ s.t. $Re(K)$
is large enough. 
Choose $Z$ to be the discrete subset of $\mathbb{C}$
defined in the proof of Lemma \ref{BernsteinSatoWF}.
For any differential operator
$P(x,\partial_x)$, note that the linear map
$u\in\mathcal{D}^\prime_{\Lambda_f}\mapsto P(x,\partial_x)u\in\mathcal{D}^\prime_{\Lambda_f}$
is continuous. If we could prove that the map $\lambda \in \{Re(\lambda)>k\} \mapsto (f+i0)^{\lambda}\in \mathcal{D}^\prime_{\Lambda_f}$
is continuous for some integer $k\in\mathbb{N}$, then
by existence of the 
functional equation for $\lambda\notin Z$, we would find some
differential operator $P$ and a polynomial $b$ such that $b(\lambda)^{-1}P(x,\partial_x)(f+i0)^{\lambda+1}=(f+i0)^\lambda$ then it follows that the map
\begin{eqnarray*}
\lambda \in \{Re(\lambda)>k-1\}\mapsto  b(\lambda)^{-1}P(x,\partial_x)(f+i0)^{\lambda+1}=(f+i0)^\lambda\in 
\mathcal{D}^\prime_{\Lambda_f}
\end{eqnarray*}
is continuous. In summary, if $\lambda \in \{Re(\lambda)>k\}\setminus Z \mapsto (f+i0)^{\lambda}\in \mathcal{D}^\prime_{\Lambda_f}$ is continuous then so is
$\lambda \in \{Re(\lambda)>k-1\}\setminus Z \mapsto (f+i0)^{\lambda}\in \mathcal{D}^\prime_{\Lambda_f}$
which means by an easy induction
that it is sufficient to prove the result
for arbitrary compact subsets $K\subset \mathbb{C}$ s.t. $Re(K)$
is large enough. 
Now if we inspect the first step of the proof
of Theorem \ref{WFf+i0}, the crucial point
relies on the product
$$\sum_i ((t+i0)^\lambda\varphi_i)(\delta_{t-f}\varphi_i) .$$
If $\lambda$ lies in a compact set
$K\subset \mathbb{C}$ s.t. $Re(K)>m_1>\frac{n+1}{2}+s$, then
$\lambda\in K\mapsto((t+i0)^\lambda\varphi_i)\in C^{m_1}$ is 
continuous hence $\lambda\in K\mapsto((t+i0)^\lambda\varphi_i)\in H^s(\mathbb{R}^{n+1})$
is continuous by
the continuous injection from Lemma \ref{ckembedding}. Now choose
$s>\frac{1}{2}$ then since $\varphi_i\delta_{t-f}$ belongs to $H^{-\frac{1}{2}-\varepsilon}(\mathbb{R}^{n+1})$
for $\varepsilon=\frac{1}{2}(s-\frac{1}{2})$
by the Sobolev trace theorems~\cite[Theorem 13.6]{Eskin}, 
and by Lemma \ref{multsobolevcontinuous} applied to 
$\left((t+i0)^\lambda\varphi_i,\varphi_i\delta_{t-f}\right)\in H^{\frac{1}{2}+2\varepsilon}\times H^{-\frac{1}{2}-\varepsilon}$, the map
$\lambda\mapsto ((t+i0)^\lambda\varphi_i)(\delta_{t-f}\varphi_i)$
is continuous with value in $\mathcal{E}^\prime(\mathbb{R}^d)$. 
By
the $u=0$ Theorem \ref{u=0thm}, 
the map $\lambda\in K\mapsto ((t+i0)^\lambda\varphi_i)(\delta_{t-f}\varphi_i)$
is bounded in $\mathcal{D}^\prime_{\Gamma_f}$ then we can conclude by following the
proof of Theorem  \ref{WFf+i0} for the family
$((t+i0)^\lambda\varphi_i)(\delta_{t-f}\varphi_i)_{\lambda\in K}$ that
$((f+i0)^\lambda)_{\lambda\in \mathbb{C}\setminus Z}$ is continuous in $\lambda\in K$ with value
$\mathcal{D}^\prime_{\Lambda_f}$. 
\end{proof}

By Lemma \ref{wavefrontlaurent}, we can deduce that:

\begin{coro}
Let $f$ be a real valued analytic function s.t. $\{df=0\}\subset\{f=0\}$,
$Z\subset \mathbb{C}$ a discrete subset containing the poles of
the meromorphic family $((f+i0)^\lambda)_\lambda$,
for all $z\in Z$,
set $a_k$ to be coefficients of the Laurent series
expansion of $\lambda\mapsto (f+i0)^\lambda$ around $z$
\begin{eqnarray*}
(f+i0)^\lambda=\sum_{k\in\mathbb{Z}} a_k(\lambda-z)^k.
\end{eqnarray*}
Then $\forall k$,
\begin{equation}
WF(a_k)\subset \Lambda_f.
\end{equation}
\end{coro}

Furthermore, we can localize the distributional support
of the coefficients
$(a_k)_k$ of the Laurent series expansion
of $((f+i0)^\lambda)_\lambda$ around poles
for negative values of $k$.
\begin{thm}
Let $f$ be a real valued analytic function s.t. $\{df=0\}\subset\{f=0\}$,
$Z\subset \mathbb{C}$ a discrete subset containing the poles of
the meromorphic family $((f+i0)^\lambda)_\lambda$. Set
$$\Lambda_f=\{ (x;\xi) \text{ s.t. }\exists \{(x_k,a_k)_k\}\in \left(\mathbb{R}^n\times \mathbb{R}_{>0}\right)^{\mathbb{N}} ,
x_k\rightarrow x, f(x_k)\rightarrow 0, a_kdf(x_k)\rightarrow\xi \} .$$
For all $z\in Z$,
let $a_k$ to be the coefficients of the Laurent series
expansion of $\lambda\mapsto (f+i0)^\lambda$ around $z$
\begin{eqnarray*}
(f+i0)^\lambda=\sum_{k\in\mathbb{Z}} a_k(\lambda-z)^k.
\end{eqnarray*}
Then for all $k\in\mathbb{Z}$, $WF(a_k)\subset\Lambda_f$
and
if $k<0$ then $a_k$ is a distribution
\textbf{supported by the critical locus} $\{df=0\}$.
\end{thm}
\begin{proof}
Let us prove that the singular terms $a_k,k<0$ in the Laurent series
expansion around $\lambda\in Z$ are distributions supported
by the critical locus $\{df=0\}$.
If $x$ is a nondegenerate point for $f$ i.e.
$f(x)=0$ but $df(x)\neq 0$, then $df\neq 0$ in some neighborhood $U_x$
of $x$ and $(f+i0)^\lambda=f^*(t+i0)^\lambda$ is well defined
by the pull--back Theorem of H\"ormander.
It is easy to check that
$\lambda\mapsto (t+i0)^\lambda\in\mathcal{D}^\prime_{T_0^*\mathbb{R}}$
is continuous for the normal topology, it follows by continuity
of the pull--back of H\"ormander for the normal topology~\cite[]{Viet-wf2}
that $(f+i0)^\lambda$ depends continuously on $\lambda$ for the normal
topology on $\mathcal{D}^\prime_{\Lambda_f}(U_x)$, therefore for any test
function $\varphi \in\mathcal{D}(U_x)$ the 
function $\lambda \mapsto (f+i0)^\lambda(\varphi)$ depends continuously on $\lambda$
and is meromorphic in $\lambda$ therefore it is holomorphic on the whole
complex plane and has no poles by the Riemann removable singularity Theorem. It follows
that if $x$ is a non degenerate point of $f$ then all terms $a_k$ for $k<0$
in the Laurent series expansion of $((f+i0)^\lambda)_\lambda$ are not supported at $x$.
\end{proof}

\section{The wave front set of  $\left(\prod_{j=1}^p
(f_j+i0)^{\lambda_j}\right)_{\lambda \in\mathbb{C}^p}$.}
Let $U$ be some open set in $\mathbb{R}^n$ and 
$\left(f_1,\dots,f_p\right)$ 
be some real valued
analytic functions on $U$  s.t. 
$\{df_j=0\}\subset \{f_j=0\}$. 
The goal of this section is to
provide a relatively simple geometric
bound on the
wave front set of the family of distributions
$\left(\prod_{j=1}^p
(f_j+i0)^{\lambda_j}\right)_{\lambda \in\mathbb{C}^p}$ 
depending meromorphically on $\lambda \in\mathbb{C}^p$.
Our proof closely follows
the case of one function $f$.

We start by recalling a particular case of some general 
result of Sabbah~\cite[Theorem 2.1]{sabbah1986polynomes} on the existence
of a multivariate Berstein Sato polynomial.
\begin{thm}\label{Sabbahmultivariatebernsteinsato}
Let $f_1,\dots,f_p$ be some analytic functions
then there exists functional relations
of the type
\begin{eqnarray*}
\forall k\in\{1,\dots,p\},
b_k(\lambda)
(f_1+i0)^{\lambda_1}\dots (f_p+i0)^{\lambda_p}=P_k(x,\partial_x,\lambda)f_k(f_1+i0)^{\lambda_1}\dots (f_p+i0)^{\lambda_p},
\end{eqnarray*}
where $\lambda=(\lambda_1,\dots,\lambda_p)$.
\end{thm}
The polynomials $(b_k)_{k\in\{1,\dots,p\}}$
are the Bernstein Sato polynomials.
The above Theorem
follows from~\cite[Theorem 2.1]{sabbah1986polynomes}
(see also~\cite{sabbah1987proximite,bahloul2005demonstration}) 
applied
to the \emph{holonomic distribution} $u=1$.
The existence of the functional equation
immediately
implies that
\begin{lemm}\label{sabbahwf}
Let $U$ be some open set in $\mathbb{R}^n$, $f_1,\dots,f_p$ be some real valued
analytic functions on $U$, $Z\subset \mathbb{C}^p$ some thin set
which contains the poles of $\prod_{j=1}^p(f_j+i0)^{\lambda_j}$
then $WF\left(\prod_{j=1}^p(f_j+i0)^{\lambda_j}\right)$ does not
depend on $\lambda\in\mathbb{C}^p\setminus Z$.
\end{lemm}

The proof of the
above Lemma is a simple
adaptation of the proof
of Lemma \ref{BernsteinSatoWF}.
In the multivariable case, the
zeros of the polynomials $(b_j)_j$
are contained in some thin set
$Z$ contained in $\mathbb{C}^p$.

\begin{thm}\label{WFprod}
Let $U$ be some open set in $\mathbb{R}^n$, $f_1,\dots,f_p$ be some real valued
analytic functions on $U$ s.t. $\{df_j=0\}\subset \{f_j=0\}$, $\prod_{j=1}^p\log^{k_j}(f_j+i0)(f_j+i0)^{\lambda_j}$ some
family of distributions depending meromorphically on $\lambda\in\mathbb{C}^p$.
Set
\begin{eqnarray}
\Gamma=\bigcup_J Z_J
\end{eqnarray}
where $J$ ranges over subsets of $\{1,\dots,p\}$ and
{\small
\begin{eqnarray*}
Z_J=\{(x;\xi)\in T^\bullet U  \text{ s.t. } , \{(x_p,a^j_p)_{j\in J}\}_p\in \left(U\times\mathbb{R}^J_{>0}\right)^{\mathbb{N}},
\forall j\in J, f_j(x)= 0,x_p\rightarrow x,\sum_{j\in J} a^j_pdf_j(x_p)\rightarrow \xi \}.
\end{eqnarray*}
}
Then there exists a thin set $Z\subset \mathbb{C}^p$
containing 
the poles of $\prod_{j=1}^p\log^{k_j}(f_j+i0)(f_j+i0)^{\lambda_j}$ such that
for all $\lambda\notin Z$:
\begin{equation}
WF\left(\prod_{j=1}^p\log^{k_j}(f_j+i0)(f_j+i0)^{\lambda_j}\right)\subset\bigcup_{J\subset \{1,\dots,p\}}Z_J.
\end{equation}
\end{thm}
\begin{proof}
It is enough to establish the Theorem
for the family $\prod_{j=1}^p(f_j+i0)^{\lambda_j}$
since $$\prod_{j=1}^p \left(\frac{d}{d\lambda_j}\right)^{k_j} \prod_{j=1}^p(f_j+i0)^{\lambda_j}
=\left(\prod_{j=1}^p\log^{k_j}(f_j+i0)(f_j+i0)^{\lambda_j}\right).$$
We follow closely the
architecture of the proof of Theorem
\ref{WFf+i0}.
First, by Lemma \ref{sabbahwf}, we can consider 
that $Re(\lambda_j)$ is chosen large enough.
We write $\prod_{j=1}^p(f_j+i0)^{\lambda_j} $ as the integral formula:
\begin{eqnarray}
\prod_{j=1}^p(f_j+i0)^{\lambda_j} = \int_{\mathbb{R}^p} dt_1\dots dt_p \left(\prod_{j=1}^p (t_j+i0)^{\lambda_j}
\prod_{j=1}^p\delta_{t_j-f_j}\right).
\end{eqnarray}
Let $\pi$ be the 
projection $\pi:(t_1,\dots,t_p,x)\in\mathbb{R}^p\times \mathbb{R}^n\mapsto x\in\mathbb{R}^n$, then the above
formula writes as the pushforward:
\begin{eqnarray}
\prod_{j=1}^p(f_j+i0)^{\lambda_j} = \pi_* \left(\prod_{j=1}^p (t_j+i0)^{\lambda_j}
\prod_{j=1}^p\delta_{t_j-f_j}\right).
\end{eqnarray}

Step 1 First, let us show that for $Re(\lambda_j),j\in\{1,\dots,p\}$ large enough
the product $\left(\prod_{j=1}^p (t_j+i0)^{\lambda_j}
\prod_{j=1}^p\delta_{t_j-f_j}\right)$ makes sense in $\mathcal{D}^\prime$.
The seperate distributional products $\prod_{j=1}^p\delta_{t_j-f_j}$
and $\prod_{j=1}^p (t_j+i0)^{\lambda_j}$ both
make sense since they satisfy the H\"ormander condition.
For $Re(\lambda_j)$ large enough, arguing as in the proof
of \ref{WFf+i0}, one can easily prove that the product
$\prod_{j=1}^p (t_j+i0)^{\lambda_j}$ can be made sufficiently regular
in the Sobolev sense so that 
the
distributional product $\left(\prod_{j=1}^p (t_j+i0)^{\lambda_j}
\prod_{j=1}^p\delta_{t_j-f_j}\right)$ makes sense.
Indeed it suffices that $\prod_{j=1}^p (t_j+i0)^{\lambda_j}\in H^{s}(\mathbb{R}^{n+p})$
for $s>\frac{p}{2}$ since $\prod_{j=1}^p\delta_{t_j-f_j}\in H^{-\frac{p}{2}-\varepsilon}(\mathbb{R}^{n+p}),\forall\varepsilon>0$ by the Sobolev trace theorem.

Step 2 We study $WF\left(\prod_{j=1}^p (t_j+i0)^{\lambda_j}
\prod_{j=1}^p\delta_{t_j-f_j}\right)$.
\begin{eqnarray*}
WF\left(\prod_{j=1}^p\delta_{t_j-f_j}\right)&=&\bigcup_{J\subset \{1,\dots,p\}} \Gamma_J,\\
\Gamma_J&=&\{(t,x;\tau,\xi) \text{ s.t. } f_j(x)=t_j, \xi=-\sum_{j\in J}\tau_j df_j ,\tau_j\neq 0 \}\\
WF\left(\prod_{j=1}^p (t_j+i0)^{\lambda_j}\right)&=&\bigcup_{J\subset \{1,\dots,p\}}\{ (t,x;\tau,0) \text{ s.t. } t_j=0,\tau_j>0\text{ if }j\in J,\tau_j=0\text{ otherwise } \}.
\end{eqnarray*}
By the $u=0$ Theorem:
\begin{eqnarray*}
WF\left(\prod_{j=1}^p (t_j+i0)^{\lambda_j}
\prod_{j=1}^p\delta_{t_j-f_j}\right)\subset WF\left(\prod_{j=1}^p\delta_{t_j-f_j}\right)\widehat{+}_iWF\left(\prod_{j=1}^p (t_j+i0)^{\lambda_j}\right)
\end{eqnarray*}
The wave front set of $WF\left(\prod_{j=1}^p (t_j+i0)^{\lambda_j}
\prod_{j=1}^p\delta_{t_j-f_j}\right)$ is not interesting
outside
$\left(\cup_j \{f_j=0\}\right)$ since it will not contribute
after push--forward by $\pi$.
\begin{equation}
\Gamma_{J,f}=\{ (t,x;\tau,\xi) \text{ s.t. }\forall j\in J, t_j=0, 
x_n\rightarrow x, f_j(x)=0, \xi_n= -\sum \tau^j_ndf_j(x_n)\rightarrow\xi, \tau^j_n<\tau^j \}.
\end{equation}
In fact, by definition of the $\widehat{+}_i$ operation
of Iagolnitzer:
\begin{eqnarray*}
\left(WF\left(\prod_{j=1}^p\delta_{t_j-f_j}\right)\widehat{+}_iWF\left(\prod_{j=1}^p (t_j+i0)^{\lambda_j}\right)\right)
\cap T_{\cup_j \{t_j=0\}}^*\left(\mathbb{R}^p\times U\right)
\subset \bigcup_{J} \Gamma_{J,f}.
\end{eqnarray*}

Step 3, we evaluate the wave front set of the push--forward.
The interesting elements of $\Gamma_{J,f}$ are the $\tau=0$ points
and are calculated as follows
\begin{eqnarray*}
&&\Gamma_{J,f}\cap\{(0,x;\tau,\xi)\text{ s.t. } \forall j\in J, \tau_j=0\}\\
&=&\{ (0,x;\tau,\xi) \text{ s.t. }
\forall j\in J, t_j=0, 
x_n\rightarrow x, f_j(x)=0, \xi_n= \sum \tau^j_ndf_j(x_n)\rightarrow\xi, \tau^j_n\geqslant 0 \}.
\end{eqnarray*}
Then it is immediate that $\pi_*(\Gamma_{J,f})=Z_J$.
\end{proof}

We can easily deduce from the
above proof and the $u=0$ Theorem \ref{u=0thm}
that when $Re(\lambda_j),\forall j$ is large enough,
the product
$\left(\prod_{j=1}^p (t_j+i0)^{\lambda_j}
\prod_{j=1}^p\delta_{t_j-f_j}\right) $ is continuous in $\lambda$
with value $\mathcal{D}^\prime_{\bigcup_{J} \Gamma_{J,f}}$ and therefore
by continuity of the pushforward~\cite{Viet-wf2}, 
the family $\left(\prod_{j=1}^p(f_j+i0)^{\lambda_j}\right)_{\lambda\in\mathbb{C}^p}$ is continuous in $\lambda$
with value $\mathcal{D}^\prime_{\bigcup_JZ_J}$ for $Re(\lambda_j)$ large enough.
We also know by Theorem
\ref{Atiyahmeromextension} that the family
$\left(\prod_{j=1}^p(f_j+i0)^{\lambda_j}\right)_{\lambda\in\mathbb{C}^p}$
is meromorphic with value
in $\mathcal{D}^\prime$ thus it is
holomorphic in $\lambda$
for $Re(\lambda)$ large enough.
The family $\prod_{j=1}^p(f_j+i0)^{\lambda}$
is both continuous in $\lambda$ with value in
$\mathcal{D}^\prime_{\bigcup_JZ_J}$ and holomorphic
with value in $\mathcal{D}^\prime$, it is thus
\textbf{holomorphic} with value in 
$\mathcal{D}^\prime_{\bigcup_JZ_J}$
by
Proposition \ref{boundedbecomesholo}.
Arguing as in the proof of Proposition
\ref{powersfmerom}
based existence 
of the Bernstein Sato polynomial we can show
that
$\prod_{j=1}^p(f_j+i0)^{\lambda}$ is meromorphic
with value $\mathcal{D}^\prime_{\bigcup_JZ_J}$.

\begin{thm}\label{WFprodbound}
Under the assumptions of the above Theorem,
the family $\prod_{j=1}^p(f_j+i0)^{\lambda}$
is meromorphic with value
$\mathcal{D}^\prime_{\bigcup_JZ_J}$.
\end{thm}
\subsection{Geometric assumptions and functional properties.}
We recall the objects of our study. 
Let $U$ be some open set in $\mathbb{R}^n$, $\left(f_1,\dots,f_p\right)$ 
be some real valued
analytic functions on $U$ s.t. $\{df_j=0\}\subset \{f_j=0\}$,
then we showed that $\prod_{j=1}^p(f_j+i0)^{\lambda_j}$ is a
family of distributions depending meromorphically on $\lambda\in\mathbb{C}^p$ with value
$\mathcal{D}^\prime_\Gamma$ where:
\begin{eqnarray}\label{defconicsetgammmaintermszj}
\Gamma=\bigcup_J Z_J
\end{eqnarray}
where $J$ ranges over subsets of $\{1,\dots,p\}$ and
{\small
\begin{eqnarray*}
Z_J=\{(x;\xi)\in T^\bullet U  \text{ s.t. } , \{(x_p,a^j_p)_{j\in J}\}_p\in \left(U\times\mathbb{R}^J_{>0}\right)^{\mathbb{N}},
\forall j\in J, f_j(x)= 0,x_p\rightarrow x,\sum_{j\in J} a^j_pdf_j(x_p)\rightarrow \xi \}.
\end{eqnarray*}
}
In this part, our goal is to add 
geometric assumptions
on the critical
loci $\bigcup\{df_j=0\}$
in order to give a nicer description of the conic set
$\Gamma$.

\subsubsection{Stratification, regularity condition and polarization.}
We define the following three geometric conditions:
\begin{enumerate}
\label{Threemainassumptions}
\item \textbf{Stratification}: The critical loci $\{df_j=0\}$ are smooth analytic submanifolds and
for every $J\subset \{1,\dots,p\}$ the submanifolds
$\{df_j=0\}$ for $j\in J$ intersect \textbf{cleanly}.
Define the submanifolds
\begin{eqnarray}
\Sigma_{J}=\bigcap_{j\in J} \{df_j=0\}
\end{eqnarray}
\item \textbf{Polarization}:
let $\Sigma=\cup_j \{df_j=0\}$, then for all $x\in \cup_{j}\{f_j=0\},x\notin\Sigma$, for all
$a_j>0$,
$\sum a_jdf_j(x)\neq 0$ and
there is a closed \textbf{convex} conic subset $\Gamma$ of $T^*\left(U\setminus \Sigma\right)$
s.t. $(x;\sum a_jdf_j(x))\in\Gamma_x$. We further assume that
$\Gamma$ 
satisfies a 
\textbf{strong
convexity} condition which reads as follows:
\begin{defi}\label{strongconvexity}
Let $U$ be an open manifold and $\Gamma\subset T^\bullet U$
a closed conic set. Then $\Gamma$ is \textbf{strongly
convex} if
for any pair of sequence $ (x_n;\xi_n),(x_n,\eta_n)$ in $\Gamma $ 
such that $(x_n;\xi_n+\eta_n)\rightarrow (x;\xi)$,
both $\vert\xi_n\vert$ and $\vert\eta_n\vert$ are bounded.
\end{defi}

\item \textbf{Regularity}: a microlocal regularity condition on the stratums
which is a particular version of Verdier's $w$ condition~\cite{verdier1976stratifications}.
\begin{eqnarray}
\forall (x,y)\in \{df_j=0\}\times\left(\{f_j=0\}\setminus \{df_j=0\}\right),
\delta(N^*_x(\{df_j=0\}),\frac{df_j(y)}{\vert df_j(y)\vert})\leqslant C\vert x-y\vert.
\end{eqnarray}
where for two vector spaces $(V,W)$,
$\delta(V,W)=\underset{x\in V,\vert x\vert=1}{\sup} \text{dist }(x,W)$.
\end{enumerate}
\begin{prop}\label{determinationwfgamma}
Let $U$ be some open set in $\mathbb{R}^n$, $\left(f_1,\dots,f_p\right)$ 
be some real valued
analytic functions on $U$ s.t. $\{df_j=0\}\subset \{f_j=0\}$
Assume the above three conditions are satisfied, then
the set $\Gamma$ defined by equation \ref{defconicsetgammmaintermszj}
satisfies the identity:
\begin{eqnarray}
\Gamma\subset \bigcup_{J}\{(x;\xi) |j\in J, f_j(x)=0,df_j(x)\neq 0, \xi=\sum_{j\in J} a_jdf_j(x),a_j>0  \}\cup N^*\Sigma_J
\end{eqnarray}
where $\Sigma_J$ is the submanifold obtained
as the clean intersection of the critical submanifolds
$\{df_j=0\},\forall j\in J$.
\end{prop}
\begin{proof}
It suffices to evaluate each set $Z_J$ separately. By polarization, on the analytic set $\cup_j \{f_j=0\}$ minus the critical locus
$\cup_{j}\{df_j=0\}$, $Z_J$
is easily calculated and equals
\begin{eqnarray}
\bigcup_{J}\{(x;\xi) |j\in J, f_j(x)=0,df_j(x)\neq 0, \xi=\sum_{j\in J} a_jdf_j(x) ,a_j>0 \}.
\end{eqnarray}
The difficulty resides in the study of $Z_J$ over the critical locus.
First use the assumption that there is some convex conic set
$\Gamma$ s.t. $\sum_{j\in J} a^j_pdf_j(x_p)\in\Gamma_{x_p}$,
the \textbf{strong convexity condition} \ref{strongconvexity} implies
that the convergence $\sum_{j\in J} a^j_pdf_j(x_p)\rightarrow \xi$ prevents
the sequences $a^j_pdf_j(x_p)$ from blowing up. Up to extraction of a convergent subsequence, 
assume w.l.o.g that
$a^j_pdf_j(x_p)\rightarrow \xi_j$,
then the regularity condition implies that $\xi_j\in N^*_{x} \left(\{df_j(x)=0\}\right)$
and
\begin{eqnarray*}
\sum_ja^j_pdf_j(x_p)&\underset{p\rightarrow +\infty}{\longrightarrow}& \sum_j\xi_j\in \sum_{j\in J}
N^*_{x} \left(\{df_j(x)=0\}\right)\\
\implies \sum_ja^j_pdf_j(x_p)&\underset{p\rightarrow +\infty}{\longrightarrow}&\xi\in N^*_x\left(\Sigma_J\right)
\end{eqnarray*}
because the
submanifolds $\{df_j=0\},j\in J$ cleanly intersect on the submanifold
$\Sigma_J$.
\end{proof}
In practical applications for QFT, we will
have to check that the above conditions
are always satisfied in order to apply the
following
Theorem:
\begin{thm}\label{functionalprod}
Under the assumptions of paragraph \ref{Threemainassumptions},
the family
$\left(\prod_{j=1}^p(f_j+i0)^{\lambda_j}\right)_{\lambda \in\mathbb{C}^p}$
depends meromorphically on $\lambda$ with linear poles
with value $\mathcal{D}^\prime_\Lambda$ where
\begin{equation}
\Lambda=\bigcup_{J}\{(x;\xi) |j\in J, f_j(x)=0,df_j(x)\neq 0, \xi=\sum_{j\in J} a_jdf_j(x),a_j>0  \}\cup N^*\Sigma_J.
\end{equation}
The distribution
\begin{equation}
\mathcal{R}_\pi\left(\prod_{j=1}^p(f_j+i0)^{k_j} \right) \in\mathcal{D}^\prime(U)
\end{equation}
is a distributional extension
of $\prod_{j=1}^p(f_j+i0)^{k_j}\in\mathcal{D}^\prime(U\setminus X)$.
\end{thm}
\begin{proof}
We already know by Theorem \ref{WFprod} that
$\Lambda=WF\left(\prod_{j=1}^p(f_j+i0)^{\lambda_j} \right)\subset \bigcup_{J}Z_J$
and $\Lambda$
is determined from Proposition \ref{determinationwfgamma}.

The meromorphicity 
with value $\mathcal{D}^\prime_\Lambda$ is a consequence of 
Proposition Theorem \ref{WFprodbound}.

Finally,
the fact
that
the singular part is supported on the critical locus
results from the fact that outside
$\Sigma=\bigcup_{j}\{df_j=0\}$, the 
distributional products
$\left(\prod_{j=1}^p (f_j+i0)^\lambda\right)$
is well defined and is bounded in $\lambda$
by \textbf{hypocontinuity} 
of the H\"ormander product~\cite{Viet-wf2} and therefore
the family
$\left(\prod_{j=1}^p (f_j+i0)^{\lambda_j}\right)_\lambda$ is \textbf{both meromorphic} in
$\lambda$ by \ref{Atiyahplus} (by the resolution of singularities of Hironaka)
and \textbf{locally bounded} it is thus holomorphic in $\lambda$ by \ref{Riemannremov}.
It follows that for all test function $\varphi\in\mathcal{D}(U\setminus \Sigma)$,
$\pi\left(\prod_{j=1}^p (f_j+i0)^{\lambda_j}(\varphi) \right)=\left(\prod_{j=1}^p (f_j+i0)^{\lambda_j}(\varphi) \right)$ since $\pi$ is a projection on holomorphic functions
and it follows that
$$\mathcal{R}_\pi\left(\prod_{j=1}^p(f_j+i0)^{k_j} \right)(\varphi)=\lim_{\lambda\rightarrow k}\left(\prod_{j=1}^p (f_j+i0)^{\lambda_j}(\varphi) \right)  $$
where the limit exists since the wave front set are transverse outside $\Sigma$.
\end{proof}

\begin{center}
 \textsc{\section*{Part II: application to meromorphic regularization in QFT.}}
\end{center}

\section{Causal manifolds and \emph{Feynman relations}.}

The goal of this part is to give a definition of Feynman propagators
which are needed to calculate vacuum expectation values (VEV)
of times ordered products ($T$-products)
in QFT. Our exposition will stress the importance
of the causal structure
of the Lorentzian manifolds considered.

To define a causal structure on a smooth manifold
$M$, we will essentially follow
Schapira's exposition~\cite{schapira2013hyperbolic,d1999global} (strongly inspired by Leray's work)
which makes use of no metric since the causal structure
is more fundamental than a metric structure and define
some cone $\gamma$ in cotangent space $T^*M$ which will induce a
partial order on $M$.
This presentation is convenient since the same cone 
will be used to describe
wave front sets of Feynman propagators
and Feynman amplitudes. 
\subsubsection{Admissible cones in cotangent space.}
For a manifold $M$ we denote by $q_1$ and $q_2$ the first and second projection
defined on $M\times M$. 
We denote by $d_2$ the diagonal of $M\times M$.
A cone $\gamma$ in a vector bundle $E\mapsto M$ is a subset of $E$ which is invariant
by the action of $\mathbb{R}_+$ on this vector bundle. We denote by $-\gamma$ the opposite
cone to $\gamma$, and by $\gamma^\circ$ the polar cone to $\gamma$, a closed convex
cone of the dual vector bundle
$\gamma^\circ = \{(x, \xi) \in E^*; \left\langle \xi , v \right\rangle \geqslant 0, \forall v\in\gamma\}.$
In all this section, we assume that $M$ is connected.
A closed relation on $M$ is a closed subset of $M\times M$.
\begin{defi}
Let $Z$ be a closed subset of $M\times M$ and
$A\subset M$ a closed set. 
\end{defi}

\begin{defi}
A cone $\gamma\subset T^*M\setminus \underline{0}$ is admissible if it is closed proper
convex, $\gamma \cap -\gamma=\emptyset$ 
and $Int(\gamma_x)\neq \emptyset$ i.e. the interior of $\gamma_x\subset T^*_xM$ is non empty for any $x\in M$.
\end{defi}

\subsubsection{A preorder relation.}
In the literature, one often encounters time-orientable Lorentzian manifolds
to which one can associate a cone in $TM$ or its polar cone in $T^*M$. Here, we
only assume that:
M is a $C^\infty$ real connected manifold and we are given an admissible 
cone $\gamma$ in $T^*M$.

\begin{defi} 
A $\gamma$-path is a continuous piecewise $C^1$-curve  $\lambda:[0,1]\mapsto M$
such that its derivative $\lambda^\prime(t)$ satisfies 
$\left\langle \lambda^\prime(t) , v \right\rangle \geqslant 0$ for all $t \in [0, 1]$ and $v \in \gamma$.
Here $\lambda^\prime(t)$ means as well the right or the left derivative, as soon as it exists
(both exist on $]0, 1[$ and are almost everywhere the same, and $\lambda^\prime_r(0)$ and $\lambda^\prime_l(1)$
exist).\\
To $\gamma$ one associates a preorder on $M$ as follows: $x\leqslant y$ 
if and only if there
exists a $\gamma$-path $\lambda$ such that $\lambda(0) = x$ and $\lambda(1) = y$.
\end{defi}

For a subset $A$ of $M$, we set:
\begin{eqnarray*}
A_{\geqslant} = \{x \in M; \exists y\in A, x \leqslant y\},\\
A_{\leqslant}= \{x \in M; \exists y\in A, x \geqslant y\}.
\end{eqnarray*}
Intuitively, $A_{\geqslant}$ (resp $A_{\leqslant}$) 
represents the \emph{past} (resp the \emph{future}) of the set $A$
for the causal relation.

\subsubsection{Topological assumptions.}

We may assume that the relation
$\leqslant $ is closed and that 
it is proper:
\begin{itemize}
\item $x_n\leqslant y_n,\forall n$ and $(x_n,y_n)\rightarrow (x,y)\implies x\leqslant y$,
\item for compact sets $A,B$, $A_{\geqslant}\cap B_{\leqslant}$ is compact. 
\end{itemize} 

\begin{defi}
A pair $\left(M,\gamma\right)$ where $\gamma\subset T^*M$ is an admissible cone whose induced preorder
relation $\leqslant$ is closed and proper is called \emph{causal}.
\end{defi}

An admissible cone $\gamma$ induces a subset $Z_\gamma\subset M\times M$
that we call the graph of the preorder relation $\leqslant$:
\begin{eqnarray}
Z_\gamma=\{(x,y)\in M\times M \text{ s.t. }  x\leqslant y  \}
\end{eqnarray}

The topological assumtions on $\leqslant $
imply that $Z_\gamma$ is closed and that 
for all compact subset $A\times B\subset M\times M$,
$q_1^{-1}(A)\cap q_2^{-1}(B)\cap Z_\gamma$ is compact.
Lorentzian manifolds are particular cases
of causal manifolds. 
The globally hyperbolic spacetimes defined by Leray are particular cases of causal manifolds
where~\cite[Definition 1.3.8 p.~23]{BGP}:
\begin{itemize}
\item the preorder relation is a \textbf{partial order relation} i.e. 
$$\left(x\leqslant y,y\leqslant x\right)\implies x=y$$
($\gamma$-paths are forbidden to describe loops),
\item the relation is \textbf{strongly causal}, for all open set $U\subset M$, for all
$x\in M$ there is some neighborhood $V$ of $x$ in $U$ such that all
causal curves whose endpoints are in $V$ are in fact contained in $V$ 
\item and the space of $\gamma$-path is \textbf{compact} in the natural topology
on the space of rectifiable curves induced from any smooth metric on $M$.
\end{itemize}

\subsection{Feynman relations and propagators.}
We assume that $(M,\gamma)$ is a causal manifold.
Relations are subsets of the cotangent space $T^*\left(M\times M\right)$.
We denote by  $N^*(d_2)$ the conormal bundle of the diagonal $d_2\subset M\times M$.
If $(M,g)$ is a Lorentzian manifold, we denote
by $(x_1;\xi_1)\sim (x_2;\xi_2)$ if the two elements $(x_1;\xi_1),(x_2;\xi_2)$
are connected by a bicharacteristic curve of $\square_g$ in cotangent space $T^*M$.
\begin{defi}\label{Feynmanrelation}
Let $(M,\gamma)$ be a causal manifold.
A subset $\Lambda\subset T^\bullet\left(M\times M\right)$ is a polarized relation 
if 
\begin{eqnarray*}
\Lambda\subset \{ x_1 < x_2\text{ and } \xi_2\in\gamma_{x_2},\xi_1\in -\gamma_{x_1} \}\cup
 \{ x_2< x_1\text{ and } \xi_1\in\gamma_{x_1},\xi_2\in -\gamma_{x_2}\}\cup N^*(d_2).
\end{eqnarray*}
 If we assume moreover that
$(M,g)$ is Lorentzian then a subset $\Lambda\subset T^\bullet\left(M\times M\right)$ is a Feynman 
relation 
if 
\begin{eqnarray*}
\Lambda\subset \{(x_1,x_2;\xi_1,\xi_2) \text{ s.t. } (x_1;\xi_1)\sim (x_2;-\xi_2) \text{ and }\xi_2\in\gamma \text{ if }x_2> x_1   \text{ and }\xi_2\in-\gamma \text{ if }x_1 > x_2  \}\cup N^*(d_2).
\end{eqnarray*}
\end{defi}
Feynman relations
are particular cases of
polarized relations.
\begin{defi}
Let $(M,g)$ be a Lorentzian manifold, $\gamma$ the corresponding 
admissible cone and $\square_g$ the corresponding
wave operator.
Then $G\in\mathcal{D}^\prime(M\times M)$ is called Feynman propagator
if $G$ is a fundamental bisolution of $\square_g+m^2$
\begin{eqnarray}
\left(\square_x+m^2\right) G(x,y)=\delta(x,y)\\
\left(\square_y+m^2\right)G(x,y)=\delta(x,y)
\end{eqnarray}
and
$WF(G)$ is a Feynman relation in $T^\bullet\left(M\times M\right)$.
\end{defi}

\subsection{Wave front set of Feynman amplitudes outside diagonals.}
We develop
a machinery
which allows us to describe
wave front sets
of Feynman amplitudes which are distributions living on configuration
spaces of causal manifolds.

\subsubsection{Configuration spaces.}

For every finite subset $I\subset \mathbb{N}$ and open subset $U\subset M$, we define
the configuration space
$U^I=\text{Maps }(I\mapsto U)=\{(x_i)_{i\in I}\text{ s.t. }x_i\in U,\forall i\in I \} $
of $\vert I\vert$ particles in $U$ 
labelled by the subset $I\subset\mathbb{N}$.
In the sequel, we will distinguish two types of diagonals
in $U^I$, the \emph{big diagonal} 
$D_I=\{(x_i)_{i\in I} \text{ s.t. }\exists (i\neq j)\in I^2, x_i=x_j \}$
which represents configurations where
at least two particles
collide, and the \emph{small diagonal}
$d_I=\{(x_i)_{i\in I} \text{ s.t. } \forall (i,j)\in I^2, x_i=x_j  \}$
where all particles in $U^I$ collapse over the same element.
The configuration space $M^{\{1,\dots,n\}}$ and the corresponding \emph{big and small} diagonals $D_{\{1,\dots,n\}},d_{\{1,\dots,n\}}$ will be denoted by
$M^n,D_n,d_n$ for simplicity.

For QFT, we are let to introduce
the concept
of \textbf{polarization}
to describe
subsets
of the cotangent of configuration 
spaces 
$T^\bullet M^n$ for all $n$ where $(M,\gamma)$
is a causal manifold:
this generalizes
the concept of positivity
of energy 
for 
the cotangent space 
of configuration space.

\subsubsection{Polarized subsets.}
 
In order to generalize this condition 
to the wave front set of Feynman amplitudes, 
we define the right concept of 
positivity of energy which is adapted 
to conic sets in $T^\bullet M^n$:
\begin{defi}\label{polarized}
Let $(M,\gamma)$ be a causal manifold.
We define a \textbf{reduced polarized part} 
(resp \textbf{reduced strictly polarized part}) 
as a conical subset
$\Xi \subset T^*M$ such that, 
if $\pi:T^*M\longrightarrow M$ 
is the natural projection, 
then $\pi(\Xi)$ 
is a finite subset $A=\{a_1,\cdots, a_r\}\subset M$
and, 
if $a\in A$ is maximal (in the sense 
there is no element $\tilde{a}$ in $A$ s.t. $\tilde{a}>a$), 
then $\left(\Xi\cap T_a^*M\right) \subset \left(\gamma\cup \underline{0}\right)$ 
(resp $\Xi\cap T_a^*M \subset \gamma$).
\end{defi}
We define
the trace
operation
as a map
which 
associates
to each element 
$p=(x_1,\dots,x_k;\xi_1,\dots,\xi_k)\in\left(T^*M\right)^k$
some finite part
$Tr(p)\subset T^* M$.
\begin{defi}
For all elements $p=((x_1,\xi_1),\cdots,(x_k,\xi_k)) \in T^*M^k$, 
we define
the \textbf{trace} $Tr(p)\subset T^*M$ 
defined by the set of elements
$(a,\eta)\in T^* M$
such that 
$\exists i\in [1,k]$ 
with the property that 
$x_i=a$, $\xi_i\neq 0$ and
$\eta=\sum_{i;x_i=a}\xi_i$.
\end{defi}
Then 
finally, 
we can define 
polarized subsets $\Gamma \subset T^*M^k$:
\begin{defi}
A conical subset $\Gamma \subset T^*M^k$ 
is \textbf{polarized} (resp strictly polarized) 
if for all $p\in\Gamma$, 
its trace $Tr(p)$  
is a reduced polarized part 
(resp reduced strictly polarized part) of $T^*M$.
\end{defi}

We enumerate easy to check properties
of polarized subsets:
\begin{itemize}
\item the union of two polarized (resp strictly polarized) subsets 
is polarized (resp strictly polarized), 
\item if a conical subset
is contained in a polarized subset
it is also polarized,
\item the projection $p:M^I\mapsto M^J$ for $J\subset I$
acts by pull--back as $p^*:T^*M^J\mapsto T^*M^I$ and sends
polarized (resp strictly polarized) subsets
to polarized (resp strictly polarized) subsets. 
\end{itemize}


The role of polarization is to control
the wave front set of the
Feynman amplitudes of the form
$\prod_{1\leqslant i<j\leqslant n} G^{n_{ij}}(x_i,x_j)\in \mathcal{D}^\prime(M^n\setminus D_n), n_{ij}\in\mathbb{N}   $
where $G$ is a Feynman propagator.

\begin{prop}\label{Feynmanrelationspolarized}
Let $(M,\gamma)$ be a causal manifold. If $\Lambda\subset T^\bullet\left(M\times M \right)$
is a Feynman relation, then $\Lambda$ is polarized and $\Lambda\cap T^\bullet\left(M^2\setminus d_2 \right)$
is strictly polarized.
\end{prop}
\begin{proof}
Obvious by definition of polarized sets and the definition of a Feynman relation.
\end{proof}

We have to check that the 
conormals of the diagonals 
$d_I$
are polarized since they are
the wave front sets
of counterterms
from the extension procedure.
\begin{prop}\label{conormpolar}
The conormal of the diagonal $d_I\subset M^I$ is polarized.
\end{prop}
\begin{proof}
Let $(x_i;\xi_i)_{i\in I}$ be 
in the conormal of
$d_I$,
let $a\in M$ s.t. 
$a=x_i,\forall i\in I$,
and $\eta=\sum \xi_i=0$
is in $\gamma_a\cup \{0\}$. Thus
the trace $Tr(x_i;\xi_i)_{i\in I}=(a;0)$
of the element $(x_i;\xi_i)_{i\in I}$ in the conormal of $d_I$
is a reduced polarized part of $T^* M$.
\end{proof}

Now we will prove the 
key theorem 
which allows 
to multiply 
two distributions
under some conditions 
of polarization on 
their
wave front sets
and deduces specific 
properties
of the wave front set
of the product:
\begin{thm}\label{polarizationthm} 
Let $u,v$ be two distributions
in $\mathcal{D}^\prime(\Omega)$, for some subset $\Omega\subset M^n$,
s.t. 
$WF(u)\cap T^\bullet \Omega$ is polarized
and $WF(v)\cap T^\bullet \Omega$ is strictly
polarized. Then
the
product $uv$
makes sense in
$\mathcal{D}^\prime(\Omega)$
and
$WF(uv)\cap T^* \Omega$ 
is polarized. Moreover,
if $WF(u)$ is also 
strictly
polarized then
$WF(uv)$ is strictly
polarized.
\end{thm}
\begin{proof}
Step 1: we prove $WF(u)+WF(v)\cap T^* \Omega$ 
does not meet 
the zero section.
For any element $p=(x_1,\dots,x_n;\xi_1,\dots,\xi_n)\in T^* M^n$
we denote by $-p$ 
the element $(x_1,\dots,x_n;-\xi_1,\dots,-\xi_n)\in T^* M^n$.
Let $p_1=(x_1,\dots,x_n;\xi_1,\dots,\xi_n)\in WF(u)$
and $p_2=(x_1,\dots,x_n;\eta_1,\dots,\eta_n)\in 
WF(v)$,
necessarily we must have 
$(\xi_1,\dots,\xi_n)\neq 0,(\eta_1,\dots,\eta_n)\neq 0$.
We will show
by a contradiction 
argument 
that the sum 
$p_1+p_2=(x_1,\dots,x_n;\xi_1+\eta_1,\dots,\xi_n+\eta_n)$
does not meet the zero section.
Assume that $\xi_1+\eta_1=0,\dots,\xi_n+\eta_n=0$
i.e. $p_1=-p_2$ then we would 
have $\xi_i=-\eta_i\neq 0$
for some $i\in\{1,\dots,n\}$ since
$(\xi_1,\dots,\xi_n)\neq 0,(\eta_1,\dots,\eta_n)\neq 0$.
We assume w.l.o.g. that $\eta_1\neq 0$, thus
$Tr(p_2)$ is non empty ! Let $B=\pi(Tr(p_1)),C=\pi(Tr(p_2))$,
we first notice $B=C$ since 
$p_2=-p_1\implies Tr(p_1)=-Tr(p_2)\implies \pi\circ Tr(p_1)=\pi\circ Tr(p_2)$.
Thus if $a$ is maximal
in $B$, $a$ is
also maximal in $C$ and we have
$$0=\sum_{x_i=a} \xi_i+\eta_i = \sum_{x_i=a} \xi_i +\sum_{x_i=a}\eta_i \in \left(\gamma_a\cup\underline{0}+\gamma_a\right)=\gamma_a,   $$
(since $p_1$ is  
polarized and $p_2$ is strictly polarized) 
contradiction ! 

 Step 2, we prove that
the set
$$\left(WF(u)+WF(v)\right) \cap T^* \Omega$$ 
is strictly polarized.
Recall $B=\pi\circ Tr(p_1)$, $C=\pi\circ Tr(p_2)$
and we denote by $A=\pi\circ Tr(p_1+p_2)$
hence in particular $A\subset B\cup C$.
We denote by $\max A$ (resp $\max B,\max C$) the set of maximal elements in $A$
(resp $B,C$). 
The key argument is to prove that $\max A=\max B\cap \max C$.
Because if $\max A=\max B\cap \max C$ holds
then for any $a\in \max A$, $\sum_{x_i=a} \xi_i+\eta_i=\sum_{x_i=a} \xi_i+\sum_{x_i=a}\eta_i\in \gamma_a$ since $a\in \max B\cap \max C$ and $Tr(p_1)$ is
a reduced polarized part and $Tr(p_2)$ is reduced strictly polarized. 
Thus $\max A=\max B\cap \max C$ implies that 
$p_1+p_2$ is strictly polarized.\\ 
We first establish the inclusion $\left(\max B\cap \max C\right)\subset \max A$. 
Let $a\in \max B\cap \max C$, then $\sum_{x_i=a} \xi_i\in \gamma_a\cup\{0\}$ and 
$\sum_{x_i=a} \eta_i\in \gamma_a$. 
Thus  
$\sum_{x_i=a} \xi_i+\eta_i\in \gamma_a\implies \sum_{x_i=a} \xi_i+\eta_i\neq 0$ 
so there must exist some $i$
for which $x_i=a$ and $\xi_i+\eta_i\neq 0$. Hence $a\in A$. Since $A\subset B\cup C$, $a\in \max B\cap \max C$, 
we deduce that $a\in\max A$ (if there were $\tilde{a}$ in $A$
greater than $a$ then $\tilde{a}\in B$ or $\tilde{a}\in C$ and $a$
would not be maximal in $B$ and $C$).

 We show the converse inclusion $\max A\subset\left(\max B\cap \max C\right)$ by contraposition.
Assume $a\notin \max B$, then there exists $x_{j_1}\in \max B$ 
s.t.
$x_{j_1}>a$ 
and $\xi_{j_1}\neq 0$. There are two cases
\begin{itemize}
\item either $x_{j_1}\in \max C$ as well, then
$\sum_{x_{j_1}=x_i}\xi_i+\eta_i\in \gamma_{x_{j_1}}
\implies \sum_{x_{j_1}=x_i}\xi_i+\eta_i\neq 0$
and there is some $i$ for which
$x_i=x_{j_1}$ and $\xi_i+\eta_i\neq 0$ thus 
$x_{j_1}\in A$ and $x_{j_1}>a$ hence $a\notin\max A$.
\item or $x_{j_1}\notin \max C$ 
then there exists $x_{j_2}\in \max C$ s.t. $x_{j_2}>x_{j_1}$ and $\eta_{j_2}\neq 0$.
Since $x_{j_1}\in \max B$,  
we must have $\xi_{j_2}=0$ so that $x_{j_2}\notin B$. 
But we also have $\xi_{j_2}+\eta_{j_2}=\eta_{j_2}\neq 0$ 
so that $x_{j_2}\in A$.
Thus $x_{j_2}\in A$ is greater than $a$ 
hence $a\notin\max A$.
\end{itemize}
We thus proved
$$a\notin\max B \implies a\notin \max A $$
and 
by symmetry of
the above arguments
in $B$ and $C$, 
we also have
$$a\notin\max C\implies a\notin \max A .$$
We established that $(\max B)^c\subset (\max A)^c$ and $(\max C)^c\subset (\max A)^c$,
thus $(\max B)^c\cup (\max C)^c\subset  (\max A)^c$ therefore 
$\max A\subset \max B\cap \max C$,
from which we deduce the equality $\max A=\max B\cap \max C$ which implies that 
$WF(u)+WF(v)$ is strictly polarized and $WF(uv)$ is polarized.
\end{proof}

An immediate corollary of the above Theorem is that
Feynman amplitudes
are well defined outside diagonals
\begin{coro}\label{wavefrontfeynmanamplitudespolarized}
Let $G\in\mathcal{D}^\prime(M^2)$
be a distribution whose wave front set
is a Feynman relation.
Then for all $n\in\mathbb{N}^*$, the distributional products
$$\prod_{1\leqslant i<j\leqslant n} G^{n_{ij}}(x_i,x_j) $$ 
are well defined in $\mathcal{D}^\prime(M^n\setminus D_n)$
and $WF\left(\prod_{1\leqslant i<j\leqslant n} G^{n_{ij}}(x_i,x_j)\right)$ is strictly
polarized on $M^n\setminus D_n$.
\end{coro}
\begin{proof}
This follows from
the fact that Feynman relations are strictly polarized outside $D_n$
hence all wave front sets are transverse by Theorem \ref{polarizationthm}
and the wave front of products are strictly polarized.
\end{proof}

\section{Meromorphic regularization of the Feynman propagator on Lorentzian manifolds.}

Let $(M,g)$ be a real analytic 
manifold with real analytic Lorentzian metric.
Our construction of meromorphic regularization
will not work on every globally hyperbolic
manifold but on a category of
``convex analytic Lorentzian spacetimes equipped with a Feynman propagator''. 
\subsection{A category from \textbf{convex} Lorentzian spacetimes.}
\label{category}
The language of category 
theory is not really necessary but 
rather convenient 
for our discussion
of the functorial
behaviour of our renormalizations.
Let us introduce the category 
$\mathbf{M}_{ca}$ which is contained in the category $\mathbf{M}_a$ 
of open analytic Lorentzian spacetimes.
An \textbf{object} $(M,g,G)$ of $\mathbf{M}_{ca}$ is
\begin{enumerate}
\item an open
real analytic manifold $M$. 
\item $M$ is endowed with a 
real analytic Lorentzian metric $g$ s.t.
$(M,g)$ is geodesically convex i.e.
for every pair $(x,y)\in M^2$, there is a 
unique geodesic of $g$ connecting
$x$ and $y$. For all $x\in M$, we denote by $\exp_x$ the exponential
map based at $x$. Since $M$ is convex, the range of $\exp_x$
is the whole manifold $M$. 
\item A Feynman propagator
$G$ which is a bisolution of the Klein Gordon
operator:
\begin{eqnarray}
\left(\square_x+m^2\right) G(x,y)&=&\delta(x,y)\\
\left(\square_y+m^2\right)G(x,y)&=&\delta(x,y)
\end{eqnarray}
and $G$ admits a \emph{representation}
for $(x,y)\in M^2$ sufficiently close:
\begin{eqnarray}
G(x,y)=\frac{U}{\Gamma +i0}+V\log\left(\Gamma+i0\right) + W
\end{eqnarray}
where $\Gamma(x,y)$ is the Synge 
world function defined as
\begin{eqnarray}
\Gamma(x,y)=\left\langle \exp^{-1}_x(y), \exp^{-1}_x(y)\right\rangle_{g_x}
\end{eqnarray}
and $\Gamma,U,V,W$ are all analytic functions.
\end{enumerate}
The \textbf{morphisms} of $\mathbf{M}_{ca}$
are defined to be the analytic embeddings $\Phi:(M,g,G)\mapsto (M^\prime,g^\prime,G^\prime)$
such that $\Phi^*g^\prime=g$, in other words 
$\Phi$ is an \textbf{isometric embedding} and $\Phi^*G^\prime=G$. 
Note that geodesics are sent to geodesics under isometries, 
hence a Lorentzian manifold isometric to a convex Lorentzian manifold
is automatically convex.

\subsection{Holonomic singularity of the Feynman propagator along diagonals.}
Once we have defined a suitable
category of spacetimes on which we could
work, we can discuss the asymptotics of Feynman propagators
near the diagonal of configuration space $M^2$.
A classical result
which goes back to Hadamard~\cite{Hadamard,BGP}
states that one can construct a Feynman propagator
$G$ which 
admits a \emph{representation}
for $(x,y)\in M^2$ sufficiently close:
\begin{eqnarray}\label{Hadamardexpansion}
G(x,y)=\frac{U}{\Gamma +i0}+V\log\left(\Gamma+i0\right) + W
\end{eqnarray}
where $\Gamma(x,y)$ is the Synge 
world function defined as
\begin{eqnarray}\label{Gamma}
\Gamma(x,y)=\left\langle \exp^{-1}_x(y), \exp^{-1}_x(y)\right\rangle_{g_x}
\end{eqnarray}
and $\Gamma,U,V,W$ are all analytic functions.
As explained in the introduction,
the key idea is that this asymptotic expansion
is of \textbf{regular holonomic type} i.e it is
in the $\mathcal{O}$ module generated by 
distributions defined as boundary values of holomorphic functions:
$(\Gamma +i0)^{-1},\log(\Gamma+i0)$.
The function $\Gamma$ should be thought of as the square
of the pseudodistance in
the pseudoriemannian setting and replaces
the quadratic form of signature $(1,3)$
used in Minkowski space $\mathbb{R}^{3+1}$.
Since $M$ belongs to the category
$\mathbf{M}_{ca}$, $M$ is convex therefore
the inverse exponential map $\exp_x^{-1}(y)$ associated
to the metric $g$ is well defined for all $(x,y)\in M^2$
and $\Gamma$ is globally defined on $M^2$.
The analytic variety
$\{\Gamma(x,y)=0\}\subset M^2$ is the null conoid associated
to the Lorentzian metric $g$.

We denote by $d_2$ the diagonal $\{x=y\}\subset M^2$
of configuration space $M^2$. 
The next step is to define the regularization $(G_\lambda)_\lambda$
of the propagator $G$. A simple solution consists in multiplying
with some complex powers of the function $\Gamma$:
\begin{defi}
Let $(M,g,G)\in\mathbf{M}_{ca}$, we define the meromorphic regularization
of $G$ as the distribution
\begin{equation}
G_\lambda=G(\Gamma+i0)^\lambda.
\end{equation}

If $M\in\mathbf{M}_a$ is not convex, then we choose a cut--off 
function $\chi$ such that $\chi=1$ in some neighborhood of the diagonal
and $\chi=0$ outside some neighborhood $V$ of the diagonal $d_2$
such that for any $(x,y)\in V$ there is a unique geodesic connecting $x$ and $y$
which implies that $\Gamma$ is well--defined on $V$.
Then define
\begin{equation}
G_\lambda=G(\Gamma+i0)^\lambda\chi+G(1-\chi).
\end{equation}
\end{defi}

Intuitively, the role of the factor 
$(\Gamma+i0)^\lambda$
is to smooth out the singularity of the Feynman propagator $G$
along the null conoid $\{\Gamma(x,y)=0\}$
when $Re(\lambda)$ is large enough.
We assume our Lorentzian manifold
to be time oriented and to be foliated by Cauchy hypersurfaces
corresponding to some time function $t$.
The Lorentzian metric $g$ induces
the existence of the natural \emph{causal
partial order relation}
$\leqslant$, and 
some convex cone $\gamma\subset T^*M$
of covectors of positive energy:
\begin{equation}
\gamma=\{(x;\xi) \text{ s.t. }g_x(\xi,\xi)\geqslant 0,dt(\xi)\geqslant 0\}.
\end{equation}

We denote by $d_2\subset M\times M$ the diagonal
$\{x=y\}$ in $M^2$. 
We describe the conic set
which contains the wave front set
of the two point functions and we study its main properties.

\begin{prop}\label{twopointpolarized}
Let $\Gamma\in C^\infty(M^2)$ be the function defined as 
\begin{eqnarray}\label{Gamma}
\Gamma(x,y)=\left\langle \exp^{-1}_x(y), \exp^{-1}_x(y)\right\rangle_{g_x}.
\end{eqnarray} 
Then:
\begin{enumerate}
\item \begin{eqnarray*}
&&\{(x,y;\xi,\eta) \text{ s.t. }\Gamma(x,y)=0, (x;\xi)\sim (y;-\eta), (x-y)^0\xi^0>0  \}\\
&=&\{(x,y;\xi,\eta) \text{ s.t. } \xi=\lambda d_x\Gamma,\eta=\lambda d_y\Gamma,\lambda\in\mathbb{R}_{>0} \}.
\end{eqnarray*}

\item Set $\Lambda_2=\{(x,y;\xi,\eta) \text{ s.t. } \xi=\lambda d_x\Gamma,\eta=\lambda d_y\Gamma,\lambda\in\mathbb{R}_{>0} \}\cup N^*\left(d_2\right)$
then $\Lambda_2$ is strictly polarized over $M^2\setminus d_2$.
\end{enumerate}
\end{prop}
\begin{proof}
It is classical and follows from the fact that
$\Gamma$ satisfies the first order differential equation
$$g^{\mu\nu}d_{x^\mu}\Gamma d_{x^\nu}\Gamma(x,y)=4\Gamma(x,y)$$
which dates back to the work of Hadamard~\cite{Hadamard,BGP}.
\end{proof}

Then we show that the 
families
$(\Gamma+i0)^{\lambda-1},(\Gamma+i0)^\lambda\log(\Gamma+i0)$
are meromorphic
with value $\mathcal{D}^\prime_{\Lambda_2}$.
\begin{prop}\label{studypowersgamma}
Let $\Gamma$ be the function defined as 
\begin{eqnarray}\label{Gamma}
\Gamma(x,y)=\left\langle \exp^{-1}_x(y), \exp^{-1}_x(y)\right\rangle_{g_x}.
\end{eqnarray} 
then
\begin{itemize}
\item the families $(\Gamma+i0)^{\lambda-1},(\Gamma+i0)^\lambda\log(\Gamma+i0)$ 
are meromorphic
of $\lambda$ with value $\mathcal{D}^\prime_{\Lambda_2}$
\item all coefficients of its Laurent series expansion around $\lambda=0$
belong to $\mathcal{D}^\prime_{\Lambda_2}$
\item its residues are \textbf{conormal} distributions supported by the diagonal $d_2$.
\end{itemize}
\end{prop}
\begin{proof}
The fact that $\Lambda_2$
is polarized and $\Lambda_2\setminus N^*\left(d_2\right)$
is strictly polarized follows from Proposition \ref{Feynmanrelationspolarized}
which is an immediate consequence of the definition of being polarized.
The three other claims
are consequences of Theorem \ref{functionalprod}, 
we have to check the three assumptions of Theorem
\ref{functionalprod}:
\begin{itemize}
\item \textbf{Stratification}: the critical manifold $\{d\Gamma=0\}$ is the diagonal $d_2\subset M^2$ and 
is a real analytic submanifold
of $\{\Gamma=0\}$
\item \textbf{Polarization}: $\Lambda_2$ is polarized by Proposition \ref{twopointpolarized}
\item \textbf{Regularity}: we perform
a local coordinate change as follows, 
$$(x,y)\in M^2 \mapsto (x,h=\exp^{-1}_x(y))\in M\times \mathbb{R}^{3+1} .$$
In this new set of coordinates $(x,h)\in M\times\mathbb{R}^{3+1}$, the
Synge world function $\Gamma$ reads $\Gamma(x,h)=h^\mu h^\nu\eta_{\mu\nu}$
where $\eta_{\mu\nu}$ is the usual symmetric tensor representing the
quadratic form
of signature $(1,3)$. It follows  that the
conormal of 
$\{\Gamma=0\}$ reads in this new coordinate system:
\begin{eqnarray}
\{(x,h;0,\xi) \text{ s.t. } \eta_{ij}h^ih^j=0, \xi=\tau\eta_{ij}h^i,\tau\neq 0 \}
\end{eqnarray}
and the diagonal $\{x=y\}$ reads $\{h=0\}$ hence
the conormal $N^*(d_2)$  reads in this new coordinate system:
\begin{eqnarray}
\{(x,0;0,\xi) \text{ s.t. } \xi\neq 0 \}.
\end{eqnarray} 
Hence it is immediate that
$\delta(d\Gamma_{(x_1,h_1)},N^*(d_2)_{(x_2,0)})=0$
and the regularity condition is thus verified
w.r.t. the conormal $N^*(d_2)$.
\end{itemize}
\end{proof}
\begin{coro}
Let $G$ be the Feynman propagator which admits an asymptotic expansion
of holonomic type \ref{Hadamardexpansion}
and $G_\lambda$ the meromorphic regularization
of $G$ defined as
\begin{equation}
G_\lambda=G(\Gamma+i0)^\lambda.
\end{equation}
Set $\Lambda_2=\{(x,y;\xi,\eta) \text{ s.t. } \xi=\lambda d_x\Gamma,\eta=\lambda d_y\Gamma,\lambda\in\mathbb{R}_{>0} \}\cup N^*\left(d_2\right)$
then the family
$\left(G_\lambda\right)_{\lambda\in\mathbb{C}}$
is meromorphic with value $\mathcal{D}^\prime_{\Lambda_2}$.
\end{coro}
\subsection{The meromorphic regularization of Feynman amplitudes.}

Our strategy to regularize a Feynman amplitude
$\prod_{1\leqslant i<j\leqslant n} G(x_i-x_j)^{n_{ij}}$
goes as follows. For every
pair of points $1\leqslant i<j\leqslant n$, let
us consider the regularized
product
\begin{equation}
G_{\lambda_{ij}}(x_i-x_j)^{n_{ij}}=G(x_i,x_j)^{n_{ij}}(\Gamma(x_i,x_j)+i0)^{n_{ij}\lambda_{ij}}
\end{equation}
depending on the complex
parameter $\lambda_{ij}\in\mathbb{C}$.
Then the regularization of the whole Feynman
amplitude reads:
\begin{equation}
\prod_{1\leqslant i<j\leqslant n} G_{\lambda_{ij}}(x_i-x_j)^{n_{ij}}
\end{equation}
which
is a family of distributions
which depends \textbf{meromorphically}
on the multivariable complex parameter 
$\lambda=(\lambda_{ij})_{1\leqslant i<j\leqslant n}\in \mathbb{C}^{\frac{n(n-1)}{2}}$
with linear poles.
This follows immediately from the existence
of the Hadamard expansion and Theorem
\ref{Atiyahmeromextension} on the analytic continuation
of complex powers of
real analytic functions.

\section{The regularization Theorem.}

Our first structure Theorem claims that
Feynman amplitudes depend meromorphically
in the \emph{complex dimensions} $(\lambda_{ij})_{1\leqslant i<j\leqslant n}$
with linear poles.
But before we prove our first main Theorem,
we need to check that the wave front sets of Feynman amplitudes on $M^n$ 
denoted by $\Lambda_n$ satisfies the strong convexity condition
of definition \ref{strongconvexity}.

\begin{defi}\label{defilambdaIcontainswffeynman}
We denote by $\Lambda_{ij}=\{(x_i,x_j;\xi_i,\xi_j) \text{ s.t. }\Gamma(x_i,x_j)=0, \xi_i=\lambda d_{x_i}\Gamma,\xi_j=\lambda d_{x_j}\Gamma,\lambda\in\mathbb{R}_{>0} \}\cup N^*\left(d_{ij}\right)$ the wave front set 
of the family $(\Gamma(x_i,x_j)+i0)^\lambda$ in $T^*(M^n\setminus D_n)$.
Define $\Lambda_I=\left((\sum_{(i<j)\in I^2} (\Lambda_{ij}+\underline{0}))\cap T^\bullet M^n\right)\cup_{J\subset I} N^*\left(d_J \right)$.
\end{defi}

\subsubsection{Strong convexity of the wave front set of Feynman amplitudes outside
$D_n$.}

We prove a fundamental Lemma about 
the conic set $\Lambda_n\cap T^*(M^n\setminus D_n)$.
Recall that the Lorentzian metric $g$ induces
the existence of the natural \emph{causal
partial order relation}
$\leqslant$, and 
some convex cone $\gamma\subset T^*M$
of covectors of positive energy:
\begin{equation}
\gamma=\{(x;\xi) \text{ s.t. }g_x(\xi,\xi)\geqslant 0,dt(\xi)\geqslant 0\}.
\end{equation}
We denote by $\Lambda_{ij}=\{(x_i,x_j;\xi_i,\xi_j) \text{ s.t. }\Gamma(x_i,x_j)=0, \xi_i=\lambda d_{x_i}\Gamma,\xi_j=\lambda d_{x_j}\Gamma,\lambda\in\mathbb{R}_{>0} \}\cup N^*\left(d_{ij}\right)$ the wave front set 
of the family $(\Gamma(x_i,x_j)+i0)^\lambda$ in $T^*(M^n\setminus D_n)$.

\begin{lemm}\label{keystrongconvexlemma}
Let $\Lambda_n=\left((\sum_{1\leqslant i<j\leqslant n} (\Lambda_{ij}+\underline{0}))\cap T^\bullet M^n\right)\cup_{J\subset I} N^*\left(d_J \right)$.
Then the conic set
$\Lambda_n\cap  T^*(M^n\setminus D_n)$ is \textbf{strongly convex}
in the sense of definition
\ref{strongconvexity}.
\end{lemm}
\begin{proof}
Let us first reformulate
the strong convexity 
condition in our case.
Let us consider
the sequences
\begin{eqnarray*}
(x_1(k),\dots,x_n(k))_k,\,\
(a_{ij}(k))_k\in \mathbb{R}_{>0}^{\mathbb{N}}
\end{eqnarray*}
and the sequence of elements of $\Lambda_n$:
$$(x_1(k),\dots,x_n(k);\sum_{1\leqslant i<j\leqslant n}a_{ij}(k)d_{x_i,x_j}\Gamma(x_i(k),x_j(k)))_{k\in\mathbb{N}}$$
such that 
$(x_1(k),\dots,x_n(k);\sum_{1\leqslant i<j\leqslant n}a_{ij}(k)d_{x_i,x_j}\Gamma(x_i(k),x_j(k)))$
converges to\\ $(x_1,\dots,x_n;\xi_1,\dots,\xi_n)\in T^*M^n$ when $k$ goes to $\infty$.
Then for all $1\leqslant i<j\leqslant n$
the sequence of covectors $a_{ij}(k)d_{x_i,x_j}\Gamma(x_i(k),x_j(k))$ remains bounded.

Without loss of generality, we assume that $(x_1(k),\dots,x_n(k))\in U^n$ for some open set $U\subset M$,
such that the cone
$\gamma|_U\subset T^*U$ satisfies the following convexity estimate:
there exists 
$\varepsilon>0$ such that
for all $((x;\xi),(x;\eta))\in \gamma^2\subset (T^*M)^2$,
$\varepsilon \left(\vert\xi\vert +\vert\eta\vert\right)\leqslant \vert\xi+\eta\vert$.

We proceed by induction on $n$.
Let us assume that the property holds true
on all configuration spaces
$M^I$ for $\vert I\vert<n$.
Let us consider
the sequences in $\Lambda_n$
\begin{eqnarray*}
(x_1(k),\dots,x_n(k))_k,\,\
(a_{ij}(k))_k\in \mathbb{R}_{>0}^{\mathbb{N}}
\end{eqnarray*}
such that $\sum_{1\leqslant i<j\leqslant n}a_{ij}(k)d_{x_i,x_j}\Gamma(x_i(k),x_j(k))$
converges to $\xi=(\xi_1,\dots,\xi_n)$ when $k$ goes to $\infty$.
By renumbering and extracting a subsequence, we can assume w.l.o.g
that $x_1(k)=\max(x_1(k),\dots,x_n(k))$ is always maximal
for the poset relation
on $M$ and that $\sum_{1\leqslant j\leqslant n} 
a_{1j}(k)d_{x_i,x_j}\Gamma(x_1(k),x_j(k))$
does not vanish for all $k$. 
\begin{eqnarray*}
&&\sum_{1\leqslant j\leqslant n} 
a_{1j}(k)d_{x_i,x_j}\Gamma(x_1(k),x_j(k))\\
&=&(\sum_{1\leqslant j\leqslant n} 
a_{1j}(k)d_{x_1}\Gamma(x_1(k),x_j(k)),\dots \sum_{1\leqslant j\leqslant n} 
a_{1j}(k)d_{x_j}\Gamma(x_1(k),x_j(k)) \dots ) 
\end{eqnarray*}
Since $\sum_{1\leqslant j\leqslant n} 
a_{1j}(k)d_{x_1}\Gamma(x_1(k),x_j(k))\rightarrow \xi_1$
and for all $k$, 
$d_{x_1}\Gamma(x_1(k),x_j(k))\in \gamma_{x_1(k)}$, 
each term
$a_{1j}(k)d_{x_1}\Gamma(x_1(k),x_j(k))$ cannot blow up.
Moreover, 
$$\forall j, \vert a_{1j}(k)d_{x_1}\Gamma(x_1(k),x_j(k)) \vert\leqslant
\varepsilon^{-1}\left(1+\vert \xi_1\vert\right)$$ for $k$ large enough
by the convexity estimate on $\gamma$.
We combine
with the fact that both elements
$(x_{1k};a_{1j}(k)d_{x_1}\Gamma(x_1(k),x_j(k)) )$ 
and
$(x_{jk};-a_{1j}(k)d_{x_j}\Gamma(x_1(k),x_j(k)) ) $
lie on the same bicharacteristic curve
which means that
$$ \sum_{1\leqslant j\leqslant n} 
a_{1j}(k)d\Gamma(x_1(k),x_j(k))\underset{k\rightarrow\infty}{\rightarrow} (\xi_1,\dots,\eta_p,\dots) .$$
It follows that 
$\sum_{2\leqslant i<j\leqslant n}a_{ij}(k)d_{x_i,x_j}\Gamma(x_i(k),x_j(k)) $
converges to $(0,\xi_2^\prime,\dots,\xi_n^\prime)\in T^{*}M^{n},
\xi^\prime_j=\xi_j-\delta_j^p\eta_p$
and
that we can identify with an element $(\xi^\prime_2,\dots,\xi^\prime_n)\in T^*(M^{n-1})$.
Then we can finish the proof using the 
inductive argument.
\end{proof}

\begin{thm}\label{regularizationthm}
Let $\prod_{1\leqslant i<j\leqslant n} G_{\lambda_{ij}}(x_i-x_j)^{n_{ij}}$
be a regularized Feynman
amplitude then
\begin{itemize}
\item the family $\Lambda_n$ is polarized in
$T^*M^n$ 
and strictly in
$T^*(M^n\setminus D_n)$ 
\item the family
of distributions 
$\left(\prod_{1\leqslant i<j\leqslant n} G_{\lambda_{ij}}(x_i-x_j)^{n_{ij}}\right)_{\lambda\in\mathbb{C}^{\frac{n(n-1)}{2}}}$
is \textbf{meromorphic with linear poles} with value $\mathcal{D}^\prime_{\Lambda_n}(M^n)$
\item the family $\left(\prod_{1\leqslant i<j\leqslant n} G_{\lambda_{ij}}(x_i-x_j)^{n_{ij}}\right)_{\lambda\in\mathbb{C}^{\frac{n(n-1)}{2}}}$ is holomorphic 
in $\lambda$ with value $\mathcal{D}^\prime_{\Lambda_n}(M^n\setminus D_n)$.
\end{itemize}
\end{thm}
\begin{proof}
The only thing we need is to check
the three assumptions, given in paragraph \ref{Threemainassumptions}, 
of Theorem \ref{functionalprod} applied to the
product:
$$\left(\prod_{ 1\leqslant i<j\leqslant n }\log^{k_{ij}}(\Gamma(x_i,x_j)+i0)(\Gamma(x_i,x_j)+i0)^{\lambda_{ij}}\right).$$
The stratification property is easy to check
since the critical locus 
of $x\mapsto \Gamma(x_i,x_j)$ is just the diagonal
$d_{ij}$ which is an analytic submanifold of $M^n$ and any finite
intersection of diagonals of the form $d_{ij}$ is a 
\emph{clean analytic submanifold}.

 Recall we denoted by $\Lambda_{ij}$ the wave front set of the family
$(\Gamma(x_i,x_j)+i0)^{\lambda_{ij}}$ in $T^*M^n$. We already know by Theorem
\ref{twopointpolarized} that $\Lambda_{ij}$ is polarized in $T^*M^n$ and strictly polarized in
$T^*\left(M^n\setminus d_n \right)$. It follows 
that
every power of Feynman propagator $G_{\lambda_{ij}}(x_i,x_j)^{n_{ij}}$
is holomorphic in $\lambda_{ij}$ with value
$\mathcal{D}_{\Lambda_{ij}}^\prime(M^n\setminus D_n)$
hence
the H\"ormander product $\prod_{ 1\leqslant i<j\leqslant n } G_{\lambda_{ij}}(x_i,x_j)^{n_{ij}}$
makes sense in $\mathcal{D}^\prime(M^n\setminus D_n)$
and by Proposition \ref{holomorphicproductprop}, 
it depends \textbf{holomorphically} in
$\lambda$ with value $\mathcal{D}^\prime_{\Lambda_n}(M^n\setminus D_n)$.  
By corollary \ref{wavefrontfeynmanamplitudespolarized}, the conic set
$\Lambda_n$ is \textbf{strictly polarized} on $M^n\setminus D_n$ 
and by Lemma \ref{keystrongconvexlemma}, $\Lambda_n$
is strongly convex therefore the product 
$\left(\prod_{ 1\leqslant i<j\leqslant n } G_{\lambda_{ij}}(x_i,x_j)^{n_{ij}}\right)_{\lambda\in\mathbb{C}^{\frac{n(n-1)}{2}}}$
satisfies the second \textbf{polarization assumption} needed
for Theorem \ref{functionalprod}.

Finally, we must check the third regularity assumption. 
The critical locus
$\{d_{x_i,x_j}\Gamma(x_i,x_j)=0\}$ is the diagonal $d_{ij}=\{x_i=x_j\}$ and we must consider its conormal
$N^*\left(d_{ij}\right)$. We must compare it with $\Lambda_{ij}=\{(y_i,y_j;\lambda d_{y_i}\Gamma,\lambda d_{y_j}\Gamma)\text{ s.t. } \Gamma(y_i,y_j)=0,\lambda >0\}$. But the regularity 
property was already checked
in the proof of Proposition \ref{studypowersgamma}.
\end{proof}
The fact that $\left(\prod_{1\leqslant i<j\leqslant n} G_{\lambda_{ij}}(x_i-x_j)^{n_{ij}}\right)_{\lambda\in\mathbb{C}^{\frac{n(n-1)}{2}}}$ is holomorphic 
in $\lambda$ with value $\mathcal{D}^\prime_{\Lambda_n}(M^n\setminus D_n)$
implies that it has a nice limit when 
$\lambda\rightarrow (0,\dots,0)\in\mathbb{C}^{\frac{n(n-1)}{2}}$,
the limit being the well defined distribution
$$\left(\prod_{1\leqslant i<j\leqslant n} G(x_i,x_j)^{n_{ij}}\right)\in\mathcal{D}^\prime(M^n\setminus D_n).$$
It follows from Theorem \ref{regularizationthm} that:
\begin{coro}\label{renormopdefwithPaycha}
Let $\mathcal{R}_\pi$ be the renormalization operator
defined in \ref{defrenormoperator} then
\begin{eqnarray*}
\mathcal{R}_\pi\left( 
\prod_{1\leqslant i<j\leqslant n}G_{0}(x_i,x_j)^{n_{ij}}\right)\in\mathcal{D}^\prime(M^n) 
\end{eqnarray*}
at
$\lambda=(0,\dots,0)$ 
is a \textbf{distributional extension} of $\left(\prod_{1\leqslant i<j\leqslant n} G(x_i,x_j)^{n_{ij}}\right)$.
\end{coro}
The above corollary gives a geometric meaning to the regularization
by analytic continuation.

\section{The renormalization Theorem.}

The goal of this section is to
prove that the renormalization operator
$\mathcal{R}_\pi$ defined in the previous section
satisfies the axioms \ref{axiomsrenormmaps} needed
for quantum field theory especially 
the factorization equation (\ref{factorizationequationQFT}).

\subsection{Renormalization maps, locality and the factorization property.}

\subsubsection{The vector subspace $\mathcal{O}(D_I,.)$ generated
by Feynman amplitudes.}\label{defifeynmanamplitudesmodule}

In QFT, renormalization is not only extension of Feynman amplitudes
in configuration space but our extension procedure
should satisfy some consistency conditions in order to
be compatible with the fundamental requirement of \textbf{locality}.

We introduce the vector space $\mathcal{O}(D_I,\Omega)$
generated by the Feynman amplitudes
\begin{equation}
\mathcal{O}(D_I,\Omega)=\left\langle\left( \prod_{i<j\in I^2} G^{n_{ij}}(x_i,x_j)\right)_{n_{ij}}\right\rangle_{\mathbb{C}}.
\end{equation}
By Corollary \ref{wavefrontfeynmanamplitudespolarized},
elements of $\mathcal{O}(D_I,\Omega)$ are distributions
in $\mathcal{D}^\prime(M^I\setminus D_I)$.
\subsubsection{Axioms for renormalization maps: factorization property as a consequence of locality.}

We define a collection of \emph{renormalization maps} 
$\left(\mathcal{R}_{\Omega\subset M^I}\right)_{\Omega,I}$ 
where $I$ runs over the finite subsets of $\mathbb{N}$
and $\Omega$ runs over the open subsets of $M^I$
which satisfy the following
axioms which are simplified versions 
of those figuring in \cite[2.3 p.~12--14]{Nikolov} \cite[Section 5 p.~33--35]{NST}:

\begin{defi}\label{axiomsrenormmaps}
For every finite subset $I\subset\mathbb{N}$, let 
$\Lambda_I$ be the conic set in $T^\bullet M^I$ 
of definition \ref{defilambdaIcontainswffeynman}.
\begin{enumerate}
\item For every $I\subset \mathbb{N},\vert I\vert<+\infty$, $\Omega\subset M^I$,
$\mathcal{R}_{\Omega\subset M^I}$ is a \textbf{linear extension operator}:
\begin{equation}
\mathcal{R}_{\Omega\subset M^I}:\mathcal{O}(D_I,\Omega)\longmapsto \mathcal{D}_{\Lambda_I}^{\prime}(\Omega).
\end{equation} 
\item For all inclusion of open subsets
$\Omega_1\subset\Omega_2\subset M^I$, we require that: 
\begin{eqnarray*}
\forall f\in\mathcal{O}(D_I,\Omega_2),\forall\varphi\in\mathcal{D}(\Omega_1)  \\
\left\langle \mathcal{R}_{\Omega_2\subset M^I}(f),\varphi\right\rangle
=\left\langle\mathcal{R}_{\Omega_1\subset M^I}(f),\varphi \right\rangle.
\end{eqnarray*}
\item The renormalization maps satisfy the \textbf{factorization property}.
If $(U,V)$ are disjoint open subsets of $M$, and $(I,J)$ are
disjoint finite subsets of $\mathbb{N}$, 
$\forall (f,g)\in \mathcal{O}(D_I,U^I)\times \mathcal{O}(D_J,V^J)$ 
and $\forall \prod_{(i,j)\in I\times J} G^{n_{ij}}(x_i,x_j),n_{ij}\in\mathbb{N}$:
\begin{eqnarray}\label{factorizationequationQFT}
&&\mathcal{R}_{(U^I\times V^J)\subset M^{I\cup J}}((f\otimes g)\prod_{(i,j)\in I\times J} G^{n_{ij}}(x_i,x_j))\\
\nonumber &=&\underset{\in \mathcal{D}_{\Lambda_{I\cup J}}^\prime(U^I\times V^J)}{\underbrace{\underset{\in \mathcal{D}_{\Lambda_I}^\prime(U^I)}{\underbrace{\mathcal{R}_{U^I\subset M^I}(f)}}\otimes \underset{\in \mathcal{D}_{\Lambda_J}^\prime(V^J)}{\underbrace{\mathcal{R}_{V^J\subset M^J}(g)}}\left(\prod_{(i,j)\in I\times J} G^{n_{ij}}(x_i,x_j)\right)}}
\end{eqnarray}
\end{enumerate}
\end{defi}

The most important property is the factorization property $(3)$ which is imposed in
\cite[equation (2.2) p.~5]{NST}.

\subsubsection{Remarks on the axioms of the Renormalization maps.}

The wave front set condition 
$$WF(\mathcal{R}_{\Omega\subset M^I}(\mathcal{O}(D_I,\Omega))\subset\Lambda_I
$$ 
is central since it allows the product
$$\underset{\in \mathcal{D}_{\Lambda_{I\cup J}}^\prime(U^I\times V^J)}{\underbrace{\underset{\in \mathcal{D}_{\Lambda_I}^\prime(U^I)}{\underbrace{\mathcal{R}_{U^I\subset M^I}(f)}}\otimes \underset{\in \mathcal{D}_{\Lambda_J}^\prime(V^J)}{\underbrace{\mathcal{R}_{V^J\subset M^J}(g)}}\left(\prod_{(i,j)\in I\times J} G^{n_{ij}}(x_i,x_j)\right)}}$$ to make sense over $U^I\times V^J$
by polarization of $\Lambda_I,\Lambda_J$ and strict polarization of 
the wave front set of $\prod_{(i,j)\in I\times J} G^{n_{ij}}(x_i,x_j)$.
 
 To define $\mathcal{R}$ on $M^I$, 
it suffices to define $\mathcal{R}_{\Omega_i\subset M^I}$ 
for an open cover $(\Omega_i)_i$
of $M^I$, by construction they necessarily coincide on the overlaps $\Omega_i\cap\Omega_j$
and the determinations can be glued together by a partition of unity. 
\subsubsection{Uniqueness property of renormalization maps.}
The following Lemma is proved in \cite[Lemmas 2.2, 2.3 p.~6]{NST}
and tells us that if a collection
of renormalization maps $(\mathcal{R}_{\Omega\subset M^I})_{\Omega,I}$ exists 
and satisfies the list of axioms \ref{axiomsrenormmaps}
then the restriction of $\mathcal{R}_{M^n}(\prod_{1\leqslant i<j\leqslant n} G^{n_{ij}}(x_i,x_j))$
on $M^n\setminus d_n$ would be 
uniquely determined
by the renormalizations $\mathcal{R}_{M^I}$ for all $\vert I\vert<n$ because of
the factorization axiom.
\begin{lemm}\label{keylemmaNST}
Let $(\mathcal{R}_{\Omega\subset M^I})_{\Omega,I}$ be a collection
of renormalization maps satisfying the axioms
\ref{axiomsrenormmaps}. Then for any Feynman amplitude
$\prod_{1\leqslant i<j\leqslant n} G^{n_{ij}}(x_i,x_j)$,
the renormalization
$\mathcal{R}_{M^n\setminus d_n\subset M^n}
(\prod_{1\leqslant i<j\leqslant n} G^{n_{ij}}(x_i,x_j))$ is uniquely determined
by the renormalizations 
$\mathcal{R}_{M^I}(\prod_{i<j\in I^2} G^{n_{ij}}(x_i,x_j))$ 
for all $\vert I\vert<n$.
\end{lemm}
\begin{proof}
See \cite[p.~6-7]{NST} for the detailed proof.
\end{proof}

Beware that the above Lemma
\textbf{does not imply the existence} of 
renormalization maps but only that
they must satisfy certain 
consistency conditions if they exist.

\subsubsection{Covering lemma.}
The following Lemma is due to Popineau and Stora \cite[Lemma 2.2 p.~6]{NST}
\cite{Stora02, Popineau} and states that
$M^n\setminus d_n$ can be partitioned as a union of open sets
on which the renormalization map $\mathcal{R}_n$ can factorize. 
\begin{lemm}\label{coveringlemmma}
Let $M$ be a smooth manifold. 
For all $I\subsetneq \{1,\dots,n\}$, 
let $C_{I}=\{(x_1,\dots,x_n)\text{ s.t. }\forall i\in I,j\notin I, x_i\neq x_j\}\subset M^n$. Then 
\begin{eqnarray}
\underset{I\subsetneq \{1,\dots,n\}}{\bigcup}C_{I}=M^n\setminus d_n.
\end{eqnarray} 
\end{lemm}
\begin{proof}
The key observation is the following,
$(x_1,\dots,x_n)\in d_n\Leftrightarrow$\\
for all neighborhood $U$ of $x_1, (x_1,\dots,x_n)\in U^n$.
On the contrary 
\begin{eqnarray*}
&&(x_1,\dots,x_n)\notin d_n\\
&\Leftrightarrow & \exists (U,V) \text{ open s.t. }\overline{U}\cap \overline{V}=\emptyset,\\
&& I\subsetneq\{1,\dots,n\},
1\leqslant\vert I\vert ,J=\{1,\dots,n\}\setminus I, \text{ s.t. } (x_1,\dots,x_n)\in U^I\times V^J.
\end{eqnarray*}
It suffices to set $\varepsilon=\underset{1<i\leqslant n }{\inf}\{d(x_i,x_1) \text{ s.t. } d(x_i,x_1)>0\}$ then let $U=\{x \text{ s.t. }d(x,x_1)< \frac{\varepsilon}{3}\}$ and $V=\{x \text{ s.t. }d(x,x_1)>\frac{2\varepsilon}{3}\}$.

It follows that the complement $M^n\setminus d_n$ of the
\emph{small diagonal} $d_n$ in $M^n$ 
is covered by open sets
of the form
$C_{I}=M^n\setminus\left(\cup_{i\in I,j\notin I}d_{ij}\right)$ where
$I\subsetneq\{1,\dots,n\}$.
\end{proof}

\subsection{Definition of the meromorphic renormalization maps.}

The Theorem \ref{regularizationthm} motivates
us to define Renormalization maps as follows.
\begin{defi}\label{defimeromrenorm}
Let $\prod_{(i<j)\in I^2}G_{\lambda_{ij}}(x_i,x_j)^{n_{ij}}$
be a Feynman amplitude in $\mathcal{O}(M^I)$. Then by Theorem \ref{regularizationthm}, 
it is a family of distributions depending meromorphically
on $\lambda\in \mathbb{C}^{\frac{n(n-1)}{2}}$ with linear poles, then we define
the action of the renormalization map $\mathcal{R}_{M^I}$ on
$\prod_{(i<j)\in I^2}G_{\lambda_{ij}}(x_i,x_j)^{n_{ij}}$ as follows:
\begin{eqnarray*}
\mathcal{R}_\pi\left( \prod_{(i<j)\in I^2}G_{0}(x_i,x_j)^{n_{ij}} \right)
\end{eqnarray*}
at
$\lambda=(0,\dots,0)$
where $\mathcal{R}_\pi$ is the regularization operator
defined in Corollary \ref{renormopdefwithPaycha}.
\end{defi}

\subsection{The main renormalization Theorem.}
We next show that the renormalization maps
$(\mathcal{R}_{M^I})_{M^I}$
defined in \ref{defimeromrenorm} satisfies
the axioms of \ref{axiomsrenormmaps}, hence
they define a genuine renormalization 
in QFT in the sense they are compatible
with the locality axioms in QFT.
\begin{thm}\label{renormthmmain}
The collection of renormalization maps
defined in \ref{defimeromrenorm} satisfies the collection
of axioms \ref{axiomsrenormmaps}. 
\end{thm}
\begin{proof}
The proof is by induction on $n$ and relies
on Theorem \ref{regularizationthm}.

We also need 
the property established in Theorem
\ref{Feynmanrelationspolarized}
that the conic set $\Lambda_2$ which contains the
wave front sets of all powers
of the regularized Feynman propagator 
is \textbf{strictly polarized}
in $T^*\left(M^2\setminus d_2\right)$ 
and polarized in $T^*M^2$.

It suffices
to check the factorization identity 
over each region
$C_I\subset M^n\setminus d_n$
of configuration space for some $
I\subsetneq \{1,\dots,n\}$ since the collection
$(C_I)_I$ forms an open cover of $M^n\setminus d_n$.
The key idea is to consider
the \textbf{formal} decomposition:
\begin{eqnarray}\label{decomp}
&&\prod_{1\leqslant i<j\leqslant n} G_{\lambda_{ij}}(x_i,x_j)^{n_{ij}}\\
\nonumber &=& \prod_{(i<j)\in I^2} G_{\lambda_{ij}}(x_i,x_j)^{n_{ij}}
\prod_{(i<j)\in I^{c2}} G_{\lambda_{ij}}(x_i,x_j)^{n_{ij}}
\prod_{(i<j)\in I\times I^c} G_{\lambda_{ij}}(x_i,x_j)^{n_{ij}}
\end{eqnarray}
that we write shortly as:
\begin{eqnarray}
t_n(\lambda_n) = t_I(\lambda_I)t_{I^c}(\lambda_{I^c})t_{I,I^c}(\lambda_{I,I^c}) \\
\nonumber t_n =\prod_{1\leqslant i<j\leqslant n} G_{\lambda_{ij}}(x_i,x_j)^{n_{ij}},\,\
t_I = \prod_{(i<j)\in I^2} G_{\lambda_{ij}}(x_i,x_j)^{n_{ij}},\\
\nonumber t_{I^c}=\prod_{(i<j)\in I^{c2}} G_{\lambda_{ij}}(x_i,x_j)^{n_{ij}},\,\
t_{I,I^c}(\lambda_{I,I^c})=\prod_{(i<j)\in I\times I^c} G_{\lambda_{ij}}(x_i,x_j)^{n_{ij}}\\
\nonumber \lambda_n=(\lambda_{ij})_{1\leqslant i<j\leqslant n},\,\
\lambda_{I}=(\lambda_{ij})_{(i<j)\in I^2},\\
\nonumber\lambda_{I^c}=(\lambda_{ij})_{(i<j)\in I^{c2}},\,\
\lambda_{I,I^c}=(\lambda_{ij})_{(i<j)\in I\times I^{c}}.
\end{eqnarray}
Let
us explain how to make sense of this decomposition.
By Theorem \ref{regularizationthm},
the left hand side
$t_n(\lambda_n)$ is meromorphic in $\lambda_n$
with value $\mathcal{D}_{\Lambda_n}^\prime$, and so are
each terms $t_I,t_{I^c},t_{I,I^c}$ w.r.t. the variables 
$\lambda_I,\lambda_{I^c},\lambda_{I,I^c}$.

The product on the right hand side makes sense
since: 
\begin{enumerate}
\item  By Theorem \ref{regularizationthm}, $t_{I}(\lambda_I)$ is meromorphic with value $\mathcal{D}^\prime_{\Lambda_I}$,
$t_{I^c}(\lambda_{I^c})$ is meromorphic with value $\mathcal{D}^\prime_{\Lambda_{I^c}}$
and $\Lambda_I,\Lambda_{I^c}$ are polarized
\item  the interaction term $\left(\prod_{(i<j)\in I\times I^c} G_{\lambda_{ij}}(x_i,x_j)^{n_{ij}}\right)$
is holomorphic with value $\mathcal{D}^\prime_{\Lambda_{I,I^c}}$
where $\Lambda_{I,I^c}=\sum_{(i<j)\in I\times I^c} 
(\Lambda_{ij}+\underline{0})\cap T^\bullet M^n$ is strictly polarized 
\end{enumerate}
therefore 
the conic sets
$\Lambda_I,\Lambda_{I^c},\Lambda_{I,I^c}$ are transverse
in $T^*C_I$ by Theorem \ref{polarizationthm} which implies
that the distributional product $t_It_{I^c}\left(\prod_{(i<j)\in I\times I^c} G_{\lambda_{ij}}(x_i,x_j)^{n_{ij}}\right)$ makes sense in $\mathcal{D}_{\Lambda_n}^\prime$
for every $\lambda_n$ avoiding the poles. Moreover by proposition
\ref{meromproductprod}, the product is \textbf{meromorphic}
in $\lambda_n$ with value $\mathcal{D}^\prime_{\Lambda_n}$ 
hence equation (\ref{decomp}) holds true in the
sense of distributions depending meromorphically 
on $\lambda_n$. 
In order to conclude, 
we make two central observations:
\begin{itemize}
\item on $C_I$, for every $(i,j)\in I\times I^c$, 
the Feynman
propagator $G_{\lambda_{ij}}(x_i,x_j)$
is holomorphic
in $\lambda_{ij}$
with value
$\mathcal{D}^\prime_{\Lambda_{ij}}(C_I)$.
Hence by strict polarization of $\Lambda_{ij}\cap T^\bullet C_I$
and Proposition \ref{holomorphicproductprop},
$t_{I,I^c}$ is holomorphic in $\lambda_{I,I^c}$ with
value $\mathcal{D}^\prime_{\Lambda_n}(C_I)$.
\item By Theorem \ref{PaychaZhang}, there exists a projection $\pi$ 
from meromorphic functions with linear poles on holomorphic functions 
satisfying the
factorization property of definition \ref{factconditionmeroabstract}
and used
to construct the renormalization operator $\mathcal{R}_\pi$, 
hence:
\begin{eqnarray*}
\pi\left( t_n(\lambda_n) \right)&=&\pi\left( t_I(\lambda_I)t_{I^c}(\lambda_{I^c})t_{I,I^c}(\lambda_{I,I^c}) \right)\\
&=& \pi\left( t_I(\lambda_I)\right) \pi\left( t_{I^c}(\lambda_{I^c})\right)
\pi\left(t_{I,I^c}(\lambda_{I,I^c}) \right)\\&& 
\text{ by factorization property and Proposition \ref{meromproductprod}}\\
&=& \pi\left( t_I(\lambda_I)\right)  \pi\left(t_{I^c}(\lambda_{I^c})\right)
t_{I,I^c}(\lambda_{I,I^c})
\end{eqnarray*}
since $t_{I,I^c}$ holomorphic and $\pi$ acts as the identity map
on holomorphic functions thus
\begin{eqnarray*}
\lim_{\lambda_n\rightarrow 0}\pi\left( t_n(\lambda_n) \right)&=&\lim_{\lambda_I\rightarrow 0} \pi\left( t_I(\lambda_I)\right) \lim_{\lambda_{I^c}\rightarrow 0} \pi\left(t_{I^c}(\lambda_{I^c})\right)
t_{I,I^c}(\lambda_{I,I^c}).
\end{eqnarray*} 
\end{itemize}
It follows by definition
of the renormalization maps that
$$\mathcal{R}_{M^n}\left(t_n\right)|_{C_I}=\mathcal{R}_{M^I}\left(t_I\right)\mathcal{R}_{M^{I^c}}\left(t_{I^c}\right)t_{I,I^c}$$
which exactly means that 
$\mathcal{R}$ factorizes on $M^n\setminus d_n$ since
$(C_I)_I$ forms an open cover of $M^n\setminus d_n$.
\end{proof}

\section{The functorial behaviour of renormalizations.}

In this last section, we investigate
the functorial behaviour of the renormalization maps
previously constructed.
We can add a new axiom
on renormalization maps which states
that renormalizations should behave functorially
w.r.t. morphisms of our category 
$\mathbf{M}_{ca}$. 

\begin{prop}
Given $(M,g,G),(M^\prime,g^\prime,G^\prime)\in \mathbf{M}_{ca}^2$ and a morphism
$$\Phi:(M^\prime,g^\prime,G^\prime)\mapsto (M,g,G),$$
then
\begin{enumerate}
\item $\Phi^*\Gamma=\Gamma^\prime$.
\item $\Phi$ acts by pull--back on $\mathcal{O}(M^I)$ and sends the Feynman amplitudes in $\mathcal{O}(M^I)$ to
Feynman amplitudes in $\mathcal{O}((M^\prime)^I)$.
\end{enumerate}
\end{prop}
\begin{proof}
The above claims are straightforward consequences
from the fact that $\Gamma$ depends only on the metric $g$ via the exponential map and from
the definition of morphisms which gives $\Phi^*G=G^\prime$.
\end{proof}

What follows is a definition of \textbf{covariant} renormalizations 
in the spirit of the seminal works~\cite{Brunetti-03,Hollands,Hollands2}
\begin{defi}
A family of collection of renormalization maps $\left((\mathcal{R}_{M^I})_{I}\right)_{(M,g,G)\in \mathbf{M}_{ca}}$
indexed by $(M,g,G)\in \mathbf{M}_{ca}$ is \textbf{covariant}
if for all morphisms $\Phi:(M^\prime,g^\prime,G^\prime)\mapsto (M,g,G)$ where 
$(M^\prime,g^\prime,G^\prime),(M,g,G)\in\mathbf{M}_{ca}^2$:
\begin{equation}
\forall t\in\mathcal{O}(M^I),
\mathcal{R}_{(M^\prime)^I} \Phi^*t=\Phi^*\left(\mathcal{R}_{M^I} t \right).
\end{equation}
\end{defi}

In section 10, all renormalization maps constructed
depend only on the element
$(M,g,G)$ in the category $\mathbf{M}_{ca}$ since the only ingredients we used
were the Feynman propagator $G$ and the Synge world function $\Gamma$ which depends only on the metric $g$.
Therefore, it follows that:
\begin{thm}
The family of collection of renormalization maps $\left((\mathcal{R}_{M^I})_{I}\right)_{(M,g,G)\in \mathbf{M}_{ca}}$
indexed by $(M,g,G)\in \mathbf{M}_{ca}$ constructed in Theorem \ref{renormthmmain} is \textbf{covariant}.
\end{thm}

\end{document}